\documentclass{fundam}

\usepackage{url} 
\usepackage[ruled,lined]{algorithm2e}
\usepackage{float,etoolbox}
\usepackage{amsfonts} 
\usepackage{graphicx}
\usepackage{caption}
\usepackage{setspace}

\usepackage{tikz}
\usetikzlibrary{graphs,quotes,decorations.shapes,decorations.pathreplacing,petri}

\newcommand{\escale}[1]{\ensuremath{\scalebox{0.8}{#1}}}
\newcommand{\nscale}[1]{\ensuremath{\scalebox{0.8}{#1}}}

\newcommand{\myEdge}[2]{ \tikz[baseline=-1pt]{
\draw[#2,line width=0.3pt] (0,0) -- ++(0.6,0) node[anchor=base, yshift=3pt, pos=0.5] {\escale{$#1$}};
}}
\newcommand{\myEEdge}[2]{ \tikz[baseline=-1pt]{
\draw[#2,line width=0.3pt] (0,0) -- ++(0.6,0) node[anchor=base, yshift=5pt, pos=0.5] {\escale{$#1$}};
}}
\newcommand{\mylEdge}[2]{ \tikz[baseline=-1pt]{
\draw[#2,line width=0.3pt] (0,0) -- ++(0.9,0) node[anchor=base, yshift=5pt, pos=0.5] {\escale{$#1$}};
}}
\newcommand{\myllEdge}[2]{ \tikz[baseline=-1pt]{
\draw[#2,line width=0.3pt] (0,0) -- ++(1.4,0) node[anchor=base, yshift=5pt, pos=0.5] {\escale{$#1$}};
}}

\newcommand{\edge}[1]{\myEdge{#1}{->}}

\newcommand{\Edge}[1]{\myEEdge{#1}{->}}
\newcommand{\ledge}[1]{\mylEdge{#1}{->}}

\newcommand{\lledge}[1]{\myllEdge{#1}{->}}

\newcommand{\nop}{\ensuremath{\textsf{nop}}}
\newcommand{\inp}{\ensuremath{\textsf{inp}}}
\newcommand{\out}{\ensuremath{\textsf{out}}}
\newcommand{\set}{\ensuremath{\textsf{set}}}
\newcommand{\res}{\ensuremath{\textsf{res}}}
\newcommand{\swap}{\ensuremath{\textsf{swap}}}
\newcommand{\free}{\ensuremath{\textsf{free}}}
\newcommand{\used}{\ensuremath{\textsf{used}}}

\usepackage{tabularx,lipsum,environ}
\makeatletter
\newcommand{\problemtitle}[1]{\gdef\@problemtitle{#1}}
\newcommand{\probleminput}[1]{\gdef\@probleminput{#1}}
\newcommand{\problemquestion}[1]{\gdef\@problemquestion{#1}}
\NewEnviron{decisionproblem}{
  \problemtitle{}\probleminput{}\problemquestion{}
  \BODY
  \par\addvspace{.5\baselineskip}
  \noindent
  \begin{tabularx}{0.97\textwidth}{@{\hspace{\parindent}} l X c}
    \multicolumn{2}{@{\hspace{\parindent}}l}{\@problemtitle} \\
    \emph{Input:} & \@probleminput \\
    \emph{Question:} & \@problemquestion
  \end{tabularx}
  \par\addvspace{.5\baselineskip}
}
\makeatother

\makeatletter
\newcommand{\problemtask}[1]{\gdef\@problemtask{#1}}
\NewEnviron{searchproblem}{
  \problemtitle{}\probleminput{}\problemtask{}
  \BODY
  \par\addvspace{.5\baselineskip}
  \noindent
  \begin{tabularx}{0.97\textwidth}{@{\hspace{\parindent}} l X c}
    \multicolumn{2}{@{\hspace{\parindent}}l}{\@problemtitle} \\
    \emph{Input:} & \@probleminput \\
    \emph{Task:} & \@problemquestion
  \end{tabularx}
  \par\addvspace{.5\baselineskip}
}
\makeatother

  \usepackage{hyperref}

\begin{document}

\setcounter{page}{125}
\publyear{2021}
\papernumber{2084}
\volume{183}
\issue{1-2}

  \finalVersionForARXIV

\title{The Complexity of Synthesis of $b$-Bounded Petri Nets}

\author{Ronny Tredup\thanks{Address for correspondence: Universit\"at Rostock,  Institut f\"ur Informatik,
                   Theoretische Informatik, Albert-Einstein-Stra\ss e 22, 18059, Rostock, Germany}
\\
  Institut f\"ur Informatik,  Theoretische Informatik\\
  Universit\"at Rostock\\
Albert-Einstein-Stra\ss e 22, 18059, Rostock, Germany\\
ronny.tredup@uni-rostock.de
}

\maketitle

\runninghead{R. Tredup}{The Complexity of Synthesis of $b$-Bounded Petri Nets}

\begin{abstract}
For a fixed type of Petri nets $\tau$, \textsc{$\tau$-Synthesis} is the task of finding for a given transition system $A$ a Petri net $N$ of type $\tau$ ($\tau$-net, for short) whose reachability graph is isomorphic to $A$ if there is one.
The decision version of this search problem is called \textsc{$\tau$-Solvability}.
If an input $A$ allows a positive decision, then it is called $\tau$-solvable and a sought net $N$ $\tau$-solves $A$.
As a well known fact, $A$ is $\tau$-solvable if and only if it has the so-called $\tau$-\emph{event state separation property} ($\tau$-ESSP, for short) and the $\tau$-\emph{state separation property} ($\tau$-SSP, for short).
The question whether $A$ has the $\tau$-ESSP or the $\tau$-SSP defines also decision problems.
In this paper, for all $b\in \mathbb{N}$, we completely characterize the computational complexity of \textsc{$\tau$-Solvability}, \textsc{$\tau$-ESSP} and \textsc{$\tau$-SSP} for the types of pure $b$-bounded Place/Transition-nets, the $b$-bounded Place/Transition-nets and their corresponding $\mathbb{Z}_{b+1}$-extensions.
\end{abstract}


\section{Introduction}\label{introduction}%

The task of system \emph{analysis} is to examine the behavior of a system and to derive its behavioral properties.
Its counterpart, \emph{synthesis}, is the task of automatically finding an implementing system for a given behavioral specification.
A valid synthesis procedure then computes a system that is correct by design if it exists.

In this paper we investigate a certain instance of synthesis:
For a fixed type of Petri nets $\tau$, \textsc{$\tau$-Synthesis} is the task to find, for a given directed labeled graph $A$, called transition system (TS, for short), a Petri net $N$ of type $\tau$ ($\tau$-net, for short) whose state graph is isomorphic to $A$ if such a net exists.
The decision version of \textsc{$\tau$-Synthesis} is called $\tau$-\textsc{Solvability}.

Synthesis for Petri nets has been investigated and applied for many years and in various fields:
It is used to extract concurrency and distributability data from sequential specifications like transition systems or languages \cite{DBLP:journals/fac/BadouelCD02}.
Synthesis has applications in the field of process discovery to reconstruct a model from its execution traces \cite{DBLP:books/daglib/0027363}.
In \cite{DBLP:journals/deds/HollowayKG97}, it is employed in supervisory control for discrete event systems.
It is also used for the synthesis of speed-independent circuits \cite{DBLP:journals/tcad/CortadellaKKLY97}.
In this paper, we investigate the computational complexity of synthesis for certain types of \emph{bounded} Petri nets, that is, Petri nets for which there is a positive integer $b$ that restricts the number of tokens on every place in every reachable marking.

In \cite{DBLP:conf/tapsoft/BadouelBD95,DBLP:series/txtcs/BadouelBD15}, synthesis has been shown to be solvable in polynomial time for bounded and pure bounded \emph{Place/Transition}-nets (P/T-nets, for short).
The approach provided in \cite{DBLP:conf/tapsoft/BadouelBD95,DBLP:series/txtcs/BadouelBD15} guarantees a (pure) bounded P/T-net to be output if such a net exists.
Unfortunately, it does not work for preselected bounds.
In fact, in \cite{DBLP:journals/tcs/BadouelBD97} it has been shown that solvability is NP-complete for $1$-bounded P/T-nets (there referred to as \emph{elementary net systems}), that is, if the bound $b=1$ is chosen in advance.
In \cite{DBLP:conf/stacs/Schmitt96}, the type of pure $1$-bounded P/T-nets is extended by the additive group $\mathbb{Z}_2$ of integers modulo $2$ (there referred to as \emph{flip-flop nets}).
Transitions of these nets can simulate the addition of integers modulo 2.
The result of \cite{DBLP:conf/stacs/Schmitt96} shows that this suffices to bring the complexity of synthesis down to polynomial time.
In \cite{DBLP:conf/tamc/TredupR19,DBLP:conf/apn/TredupR19}, we progressed the approach of examining the effects of the presence and absence of different interactions on the complexity of synthesis for the broader class of \emph{Boolean} Petri nets that enable independence between places and transitions.
This class also contains the type of $1$-bounded P/T-nets and its $\mathbb{Z}_2$-extension.
Although~\cite{DBLP:conf/tamc/TredupR19,DBLP:conf/apn/TredupR19} show that synthesis remains hard for 75 of the 128 possible Boolean types (allowing independence), \cite{DBLP:conf/tamc/TredupR19} also discovers 36 types for which synthesis is doable in polynomial time.
The latter applies in particular for the $\mathbb{Z}_2$ extension of $1$-bounded P/T-nets.
As another aspect that possibly might influence the complexity of synthesis of (pure) $1$-bounded P/T-nets, the grade $g$ of a TS $A$ as has been introduced in \cite{DBLP:conf/apn/TredupRW18}:
A TS $A$ is $g$-\emph{grade} if every state of $A$ has at most $g$ incoming and at most $g$ outgoing labeled edges.
There we showed that synthesis of pure $1$-bounded P/T-nets remains NP-complete even for acyclic $1$-grade TS.
In \cite{DBLP:journals/corr/abs-1911-05834}, for any fixed $g\in \mathbb{N}$, we completely characterize the computational complexity of synthesis from $g$-grade TS for all Boolean Petri net types that enable independence.
Surprisingly enough, for many other Boolean types, synthesis remains hard for all $g\geq 1$.
For example, this applies to the type of \emph{inhibitor nets} and the type of \emph{contextual nets}, which have originally been introduced in~\cite{DBLP:conf/apn/Pietkiewicz-Koutny97} and~\cite{DBLP:journals/acta/MontanariR95} and are referred to as $\{\nop, \inp, \out, \free\}$ and $\{\nop, \inp, \out, \used, \free\}$ in~\cite{DBLP:journals/corr/abs-1911-05834}, respectively.
However, there are several types for which the complexity changes when $g$ becomes small enough.
This applies in particular to the Boolean type of \emph{trace nets} that has originally been introduced in \cite{DBLP:journals/acta/BadouelD95} and is referred to as $\{\nop,\inp,\out,\res,\set,\used,\free\}$ in \cite{DBLP:journals/corr/abs-1911-05834}.
Synthesis for this type is hard if $g\geq2$, but polynomial for $g < 2$.
The same is true for the type of \emph{set nets} that has originally been introduced in~\cite{DBLP:journals/acta/KleijnKPR13} and is referred to as $\{\nop, \inp, \set, \used\}$ in~\cite{DBLP:journals/corr/abs-1911-05834}.

However, some questions in the area of synthesis for Petri nets are still open.
Recently, the complexity status of synthesis for (pure) $b$-bounded P/T-nets, where $b\geq 2$, has been reported as unknown~\cite{DBLP:conf/concur/SchlachterW17}.
Furthermore, it has not yet been analyzed whether extending (pure) $b$-bounded P/T-nets by the group $\mathbb{Z}_{b+1}$ provides also a tractable superclass if $b\geq 2$.

Let $b\in\mathbb{N}^+$.
In this paper, we show that solvability for (pure) $b$-bounded P/T-nets is NP-complete even if the input is an acyclic $1$-grade TS.
Moreover, for $b\geq 2$, we introduce (pure) $\mathbb{Z}_{b+1}$-extended $b$-bounded P/T-nets.
This type originates from (pure) $b$-bounded P/T-nets by adding interactions between places and transitions simulating the addition of integers modulo $b+1$.
This extension is a natural generalization of Schmitt's approach that does this for $b=1$~\cite{DBLP:conf/stacs/Schmitt96}.
In contrast to Schmitt's result~\cite{DBLP:conf/stacs/Schmitt96}, in this paper, we show that solvability for (pure) $\mathbb{Z}_{b+1}$-extended $b$-bounded P/T-nets remains NP-complete for all $b\geq 2$ even if the input is restricted to $g$-grade TS where $g\geq 2$.
In particular, this makes the synthesis of all of these $b$-bounded Petri net types NP-hard.
The question arises whether there are also types of $b$-bounded P/T-nets for which synthesis is tractable if $b\geq 2$.
We affirm this question and propose the type of restricted $\mathbb{Z}_{b+1}$-extended $b$-bounded P/T-nets.
This paper shows, that synthesis is solvable in polynomial time for this type.

To prove the NP-completeness of solvability we use its well known close connection to the so-called \emph{event state separation property} (ESSP, for short) and \emph{state separation property} (SSP, for short).
In fact, a TS $A$ is solvable with respect to a Petri net type if and only if it has the type related ESSP \emph{and} SSP \cite{DBLP:series/txtcs/BadouelBD15}.
The question of whether a TS $A$ has the ESSP or the SSP also defines decision problems.
The possibility to efficiently decide if $A$ has at least one of both properties serves as quick-fail pre-processing mechanisms for solvability.
Moreover, if $A$ has the ESSP then synthesizing Petri nets up to language equivalence is possible \cite{DBLP:series/txtcs/BadouelBD15}.
This makes the decision problems ESSP and SSP worth to study.
In \cite{DBLP:journals/tcs/Hiraishi94}, both problems have been shown to be NP-complete for pure $1$-bounded P/T-nets.
This has been confirmed for almost trivial inputs in \cite{DBLP:conf/apn/TredupRW18,DBLP:conf/concur/TredupR18}.

In this paper, for all $b\in \mathbb{N}^+$, we show that ESSP and SSP are NP-complete for (pure) $b$-bounded P/T-nets even if the input is an acyclic $1$-grade TS.
Moreover, for all $b\geq 2$, the ESSP is shown to remain NP-complete for (pure) $\mathbb{Z}_{b+1}$-extended $b$-bounded P/T-nets for $g$-grade TS where $g\geq 2$.
By way of contrast, in this paper, we show that SSP is decidable in polynomial time for the type of (pure) $\mathbb{Z}_{b+1}$-extended $b$-bounded P/T-nets, for all $b\in \mathbb{N}$.
To the best of our knowledge, so far, this is the first net family where the provable computational complexity of SSP is different to solvability and ESSP.

All presented NP-completeness proofs base on a reduction from the monotone one-in-three 3-SAT problem that is known to be NP-complete~\cite{DBLP:journals/dcg/MooreR01}.
Every reduction starts from a given boolean input expression $\varphi$ and results in an accordingly restricted $g$-grade TS $A$.
The expression $\varphi$ belongs to monotone one-in-three 3-SAT if and only if $A$ has the ESSP or the SSP or the solvability, depending on which of the properties is queried.

The proofs of the announced polynomial time results base on a generalization of Schmitt's approach~\cite{DBLP:conf/stacs/Schmitt96} that reduces ESSP and SSP to systems of linear equations modulo $b+1$.
It exploits that the solvability of such systems is decidable in polynomial time.

This paper is organized as follows:
Section~\ref{sec:preliminaries} introduces necessary definitions and provides them with illustrating examples.
Moreover, it also presents some basic results that are used throughout the paper.
Section~\ref{sec:unions} introduces the concept of unions applied by the proofs of our hardness results.
Section~\ref{sec:hardness_results} provides the NP-completeness results and presents the corresponding reductions that prove their validity.
Section~\ref{sec:poly_results} provide the announced tractability results.
Finally, Section~\ref{sec:conclusion} closes the paper.
This paper is an extended version of~\cite{DBLP:conf/apn/Tredup19,DBLP:conf/apn/Tredup19a}.

\section{Preliminaries}\label{sec:preliminaries}%
In this section, we introduce necessary notions and provide some basic results that we use throughout the paper as well as some examples.

\begin{definition}[Transition System]\label{def:ts}
A (deterministic) \emph{transition system} (TS, for short) $A=(S,E, \delta)$ is a directed labeled graph with states $S$, events $E$ and partial \emph{transition function} $\delta: S\times E \longrightarrow S$, where $\delta(s,e)=s'$ is interpreted as the \emph{edge} $s\edge{e}s'$.
For $s\edge{e}s'$ we say $s$ is a source and $s'$ is a target of $e$, respectively.
An event $e$ \emph{occurs} at a state $s$, denoted by $s\edge{e}$, if $\delta(s,e)$ is defined.
A word $w=e_0\dots e_n\in E^*$ \emph{occurs} at a state $s$, denoted by $s\edge{w}$, if it is the empty word $\varepsilon$ or there are states $q_0,\dots, q_n$ such that $s=q_0$ and $\delta(q_i,e_{i+1})=q_{i+1}$ is defined for all $i\in \{0,\dots,n-1\}$.
An \emph{initialized} TS $A=(S,E,\delta, \iota)$ is a TS with a distinct state $\iota \in S$ such that every state $s\in S$ is \emph{reachable} from $\iota$ by a directed labeled path.
The language of $A$ is the set $L(A)=\{w\in E^* \mid \iota \edge{w}\}$.
\end{definition}

In the remainder of this paper, if not explicitly stated otherwise, we assume all TS to be initialized and if a TS $A$ is not explicitly defined, then we refer to its components consistently by $S(A)$ (states) and $E(A)$ (events) and $\delta_A$ (transition function) and $\iota_A$ (initial states).

\begin{definition}[$g$-grade, linear]\label{def:grade}
Let $g\in \mathbb{N}$.
A TS $A=(S,E,\delta,\iota)$ is \emph{$g$-grade} if, for every state $s\in S$, the number of incoming and outgoing labeled edges at $s$ is at most $g$: $\vert \{e\in E\mid \edge{e}s\}\vert\leq g$ and $\vert \{e\in E\mid s\edge{e}\}\vert\leq g$.
If a TS is $1$-grade and cycle free, that is, there are pairwise distinct states $s_0,\dots, s_m$ such that $A=s_0\edge{e_1}\dots\edge{e_m}s_m$, then we say $A$ is \emph{linear};
we call $s_m$ the \emph{terminal state} of $A$ and, for all $i < j\in \{1,\dots, m\}$, we say $e_j$ and $s_j$ occur after $e_i$.
\end{definition}


In this paper, we deal with (different kinds of Petri) nets.
Nets have places, transitions, a flow and an initial marking.
Places can contain \emph{tokens}.
A global marking of a net defines for every place $p$ how many tokens it contains initially.
The firing of a transition can change locally the content of some places and thus globally the marking of the net.
The flow defines the relations between places and transitions:
how many token must a place contain to allow the firing of a transition and in which way changes the firing of a transition the content of a place.
Nets are classified by the number of tokens that a place can maximally contain (markings) and according to how places and transitions may influence each other (flow).
This way to classify nets leads to infinite many different classes of nets.
In order to deal with these classes in a uniform way, the notion of types of nets has been developed in~\cite{DBLP:series/txtcs/BadouelBD15}:

\begin{definition}[Type of nets]\label{def:type_of_nets}
A type of nets $\tau$ is a (non-initialized) TS $\tau=(S_\tau, E_\tau,\delta_\tau)$ with $S_\tau\subseteq \mathbb{N}$.
\end{definition}

Based on this notion, we are now able to define $\tau$-nets, where the states $S_\tau$ of $\tau=(S_\tau,E_\tau,\delta_\tau)$ correspond to possible contents of places, the events $E_\tau$ correspond to possible relations between places and transitions and the partial transition function $\delta_\tau$ describe how the contents of places can be changed by the firing of a transition and, moreover, which contents can inhibit such a firing:

\begin{definition}[$\tau$-Nets]\label{def:tau_nets}
Let $\tau=(S_\tau, E_\tau, \delta_\tau)$ be a type of nets.
A Petri net $N = (P, T, M_0, f)$ of type $\tau$, ($\tau$-net, for short) is given by finite and disjoint sets $P$ of places and $T$ of transitions, an initial marking $M_0: P\longrightarrow  S_\tau$, and a (total) flow function $f: P \times T \rightarrow E_\tau$.
A $\tau$-net realizes a certain behavior by firing sequences of transitions:
A transition $t \in T$ can fire in a marking $M: P \longrightarrow  S_\tau$ if $\delta_\tau(M(p), f(p,t))$ is defined for all $p\in P$.
By firing, $t$ produces the next marking $M' : P\longrightarrow  S_\tau$ where $M'(p)=\delta_\tau(M(p), f(p,t))$ for all $p\in P$.
This is denoted by $M \edge{t} M'$.
Given a $\tau$-net $N=(P, T, M_0, f)$, its behavior is captured by a transition system $A_N$, called the reachability graph of $N$.
The state set of $A_N$ is the reachability set $RS(N)$, that is, the set of all markings that, starting from initial state $M_0$, are reachable by firing a sequence of transitions.
For every reachable marking $M$ and transition $t \in T$ with $M \edge{t} M'$ the state transition function $\delta_{A_N}$ of $A_N$ is defined by $\delta_{A_N}(M,t) = M'$.
\end{definition}

Let $b\in \mathbb{N}^+$ be arbitrary but fixed.
In this paper, the following types of ($b$-bounded Petri) nets are the subject of our investigations:

\begin{definition}[$\tau_{PT}^b$]
The type of \emph{$b$-bounded P/T-nets} $\tau_{PT}^b=(S_{\tau_{PT}^b}, E_{\tau_{PT}^b}, \delta_{\tau_{PT}^b})$ has the state set $S_{\tau_{PT}^b}=\{0,\dots, b\}$ and  the event set $E_{\tau_{PT}^b}=\{0,\dots, b\}^2$ and, for all $s\in S_{\tau_{PT}^b}$ and all $(m,n)\in E_{\tau_{PT}^b}$, the transition function is defined by $\delta_{\tau_{PT}^b}(s,(m,n))=s-m+n$ if $s\geq m$ and $ s-m+n \leq b$, and undefined otherwise.
\end{definition}

\begin{definition}[$\tau_{PPT}^b$]
The type $\tau_{PPT}^b=(S_{\tau_{PPT}^b}, E_{\tau_{PPT}^b}, \delta_{\tau_{PPT}^b})$ of \emph{pure $b$-bounded P/T-nets} is a restriction of $\tau_{PT}^b$ that discards all events $(m,n)$ from $E_{\tau_{PT}^b}$ where both $m$ and $n$ are positive.
To be exact, $S_{\tau_{PPT}^b}=S_{\tau_{PT}^b}$ and $E_{\tau_{PPT}^b}=E_{\tau_{PT}^b} \setminus \{(m,n) \mid 1 \leq m,n \leq b\}$ and, for all $s\in S_{\tau_{PPT}^b}$ and all $e\in E_{\tau_{PPT}^b}$, we have $\delta_{\tau_{PPT}^b}(s,e)=\delta_{\tau_{PT}^b}(s,e)$.
\end{definition}

\begin{definition}[$\tau_{\mathbb{Z}PT}^b$]
\begin{spacing}{1.06}
The type $\tau_{\mathbb{Z}PT}^b=(S_{\tau_{\mathbb{Z}PT}^b}, E_{\tau_{\mathbb{Z}PT}^b}, \delta_{\tau_{\mathbb{Z}PT}^b})$ of \emph{$\mathbb{Z}_{b+1}$-extended $b$-bounded P/T-nets} originates from $\tau_{PT}^b$ by extending the event set $E_{\tau_{PT}^b}$ with the elements $0,\dots, b$.
The transition function additionally simulates the addition modulo (b+1).
More exactly,  $S_{\tau_{\mathbb{Z}PT}^b}=S_{\tau_{PT}^b}$ and $E_{\tau_{\mathbb{Z}PT}^b}=(E_{\tau_{PT}^b}\setminus \{(0,0)\}) \cup \{0,\dots, b\}$ and, for all $s\in S_{\tau_{\mathbb{Z}PT}^b}$ and all $e\in E_{\tau_{\mathbb{Z}PT}^b}$ we have that $\delta_{\tau_{\mathbb{Z}PT}^b}(s,e)=\delta_{\tau_{PT}^b}(s,e)$ if $e\in E_{\tau_{PT}^b}$, else $\delta_{\tau_{\mathbb{Z}PT}^b}(s,e)=(s+e) \text{ mod } (b+1)$.
\end{spacing}\vspace*{-2mm}
\end{definition}

\begin{definition}[$\tau_{\mathbb{Z}PPT}^b$]
The type $\tau_{\mathbb{Z}PPT}^b=(S_{\tau_{\mathbb{Z}PPT}^b}, E_{\tau_{\mathbb{Z}PPT}^b}, \delta_{\tau_{\mathbb{Z}PPT}^b})$ of pure \emph{$\mathbb{Z}_{b+1}$-extended $b$-bounded P/T-nets} is a restriction of $\tau_{\mathbb{Z}PT}^b$ such that $S_{\tau_{\mathbb{Z}PPT}^b}=S_{\tau_{\mathbb{Z}PT}^b}$ and $E_{\tau_{\mathbb{Z}PPT}^b}=E_{\tau_{\mathbb{Z}PT}^b}\setminus \{(m,n) \mid 1 \leq m,n \leq b\}$ and, for all $s\in S_{\tau_{\mathbb{Z}PPT}^b}$ and all $e\in E_{\tau_{\mathbb{Z}PPT}^b}$, the transition function is defined by $\delta_{\tau_{\mathbb{Z}PPT}^b}(s,e)=\delta_{\tau_{\mathbb{Z}PT}^b}(s,e)$.
\end{definition}

\begin{definition}[$\tau_{R\mathbb{Z}PT}^b$]
\begin{spacing}{1.1}
The type of \emph{restricted $\mathbb{Z}_{b+1}$-extended $b$-bounded P/T-nets} $\tau_{R\mathbb{Z}PT}^b=(S_{\tau_{R\mathbb{Z}PT}^b},E_{\tau_{R\mathbb{Z}PT}^b}, \delta_{\tau_{R\mathbb{Z}PT}^b})$ has the same state set $S_{\tau_{R\mathbb{Z}PT}^b}=S_{\tau_{\mathbb{Z}PT}^b}$ and the same event set $E_{\tau_{R\mathbb{Z}PT}^b}=E_{\tau_{\mathbb{Z}PT}^b}$ as $\tau_{\mathbb{Z}PT}^b$, but a restricted transition function $\delta_{\tau_{R\mathbb{Z}PT}^b}$.
In particular, $\delta_{\tau_{R\mathbb{Z}PT}^b}$ restricts $\delta_{\tau_{\mathbb{Z}PT}^b}$ in such a way that for $s\in S_{\tau_{R\mathbb{Z}PT}^b}$ and $(m,n)\in E_{\tau_{R\mathbb{Z}PT}^b}$ we have that $\delta_{\tau_{R\mathbb{Z}PT}^b}(s,(m,n))=\delta_{\tau_{\mathbb{Z}PT}^b}(s,(m,n))$ if $s=m$;
otherwise if $s\not=m$, then $\delta_{\tau_{R\mathbb{Z}PT}^b}(s,(m,n))$ remains undefined.
Hence, every $(m,n)\in E_{\tau_{R\mathbb{Z}PT}^b}$ occurs exactly once in $\tau_{R\mathbb{Z}PT}^b$.
Furthermore, if $(s,e)\in \{0,\dots,b\}^2$ then $\delta_{\tau_{R\mathbb{Z}PT}^b}(s,e)=\delta_{\tau_{\mathbb{Z}PT}^b}(s,e)$.
\end{spacing}
\end{definition}

\begin{figure}[h!]
\vspace*{-2mm}
\begin{minipage}{\textwidth}
\begin{center}
\begin{tikzpicture}[scale = 0.95]
\begin{scope}%
\node (type) at (-2.5,0) {$\tau_{PT}^b$:};
\node (0) at (0,0) {\nscale{$0$}};
\node (1) at (5,0) {\nscale{$1$}};
\node (2) at (10,0) {\nscale{$2$}};

\path (0) edge [<-, out=120,in=-120,looseness=10] node[left] {\escale{$(0,0)$}} (0);
\path (1) edge [->, out=225,in=-45,looseness=4] node[below] { \escale{$(1,1)$}  } (1);
\path (1) edge [->, out=-225,in=45,looseness=4] node[above] {  \escale{$(0,0)$}  } (1);
\path (2) edge [->, out=-60,in=60,looseness=10] node[right, align= left] {\escale{$(0,0)$} \\ \escale{$(1,1)$} \\ \escale{$(2,2)$}  } (2);
\path (1) edge [->, bend left=10] node[above, align= left] {   \escale{$(0,1), (1,2)$}}   (2);
\path (1) edge [<-, bend right=10] node[below, align= left] {   \escale{$(1,0),(2,1)$}   }   (2);
\graph{

(0) ->[bend left= 10, "\escale{$(0,1)$}"] (1);
(0) ->[bend left= 30, "\escale{$(0,2)$}"] (2);
(0) <-[bend right= 10,swap,  "\escale{$(1,0)$}"] (1);
(0) <-[bend right= 30, swap, "\escale{$(2,0)$}"] (2);
};
\end{scope}

\end{tikzpicture}
\end{center}
\end{minipage}
\begin{minipage}{\textwidth}
\begin{center}
\begin{tikzpicture}[scale = 0.95]
\begin{scope}%
\node(type) at (-2.5, 0) { $\tau_{\mathbb{Z}PT}^2$:};

\node (0) at (0,0) {\nscale{$0$}};
\node (1) at (5,0) {\nscale{$1$}};
\node (2) at (10,0) {\nscale{$2$}};

\path (0) edge [<-, out=120,in=-120,looseness=10] node[left] {\escale{$0$}} (0);
\path (1) edge [->, out=225,in=-45,looseness=4] node[below] { \escale{$(1,1)$}  } (1);
\path (1) edge [->, out=-225,in=45,looseness=4] node[above] {  \escale{$0$}  } (1);
\path (2) edge [->, out=-60,in=60,looseness=10] node[right, align= left] {\escale{$ 0$} \\ \escale{$(1,1)$} \\ \escale{$(2,2)$}  } (2);
\path (1) edge [->, bend left=10] node[above, align= left] {   \escale{$(0,1), (1,2), 1$}}   (2);
\path (1) edge [<-, bend right=10] node[below, align= left] {   \escale{$(1,0), (2,1), 2$}   }   (2);
\graph{

(0) ->[bend left= 10, "\escale{$(0,1),1$}"] (1);
(0) ->[bend left= 30, "\escale{$(0,2),2$}"] (2);
(0) <-[bend right= 10,swap,  "\escale{$(1,0), 2$}"] (1);
(0) <-[bend right= 30, swap, "\escale{$(2,0) ,1$}"] (2);
};
\end{scope}

\end{tikzpicture}
\end{center}
\end{minipage}
\begin{minipage}{\textwidth}
\begin{center}
\begin{tikzpicture}[scale = 0.95]
\begin{scope}%
\node(type) at (-2.5, 0) { $\tau_{R\mathbb{Z}PT}^2$:};

\node (0) at (0,0) {\nscale{$0$}};
\node (1) at (5,0) {\nscale{$1$}};
\node (2) at (10,0) {\nscale{$2$}};

\path (0) edge [<-, out=120,in=-120,looseness=10] node[left] {\escale{$0$}} (0);
\path (1) edge [->, out=225,in=-45,looseness=4] node[below] { \escale{$(1,1)$}  } (1);
\path (1) edge [->, out=-225,in=45,looseness=4] node[above] {  \escale{$0$}  } (1);
\path (2) edge [->, out=-60,in=60,looseness=10] node[right, align= left] {\escale{$ 0$} \\ \escale{$(2,2)$}  } (2);
\path (1) edge [->, bend left=10] node[above, align= left] {   \escale{$(1,2), 1$}}   (2);
\path (1) edge [<-, bend right=10] node[below, align= left] {   \escale{$(2,1), 2$}   }   (2);
\graph{

(0) ->[bend left= 10, "\escale{$(0,1),1$}"] (1);
(0) ->[bend left= 25, "\escale{$(0,2),2$}"] (2);
(0) <-[bend right= 10,swap,  "\escale{$(1,0), 2$}"] (1);
(0) <-[bend right= 25, swap, "\escale{$(2,0) ,1$}"] (2);
};
\end{scope}

\end{tikzpicture}
\end{center}\vspace*{-2mm}
\caption{The $2$-bounded types $\tau_{PT}^2$ (top) and $\tau_{\mathbb{Z}PT}^2$ (middle) and $\tau_{R\mathbb{Z}PT}^2$ (bottom), respectively, with states $0$ and $1$ and $2$.
Multiple edges with the same source and target (but different events) are represented as one edge with multiple events.
For example, $\tau_{PT}^b$ has three edges from $2$ to $2$: one labeled $(0,0)$, another $(1,1)$, and a third labeled $(2,2)$.
Similarly, $\tau_{\mathbb{Z}PT}^2$ has two edges from $2$ to $0$: one labeled $(2,0)$ and another labeled $1$.
}\label{fig:types}
\end{minipage}
\end{figure}
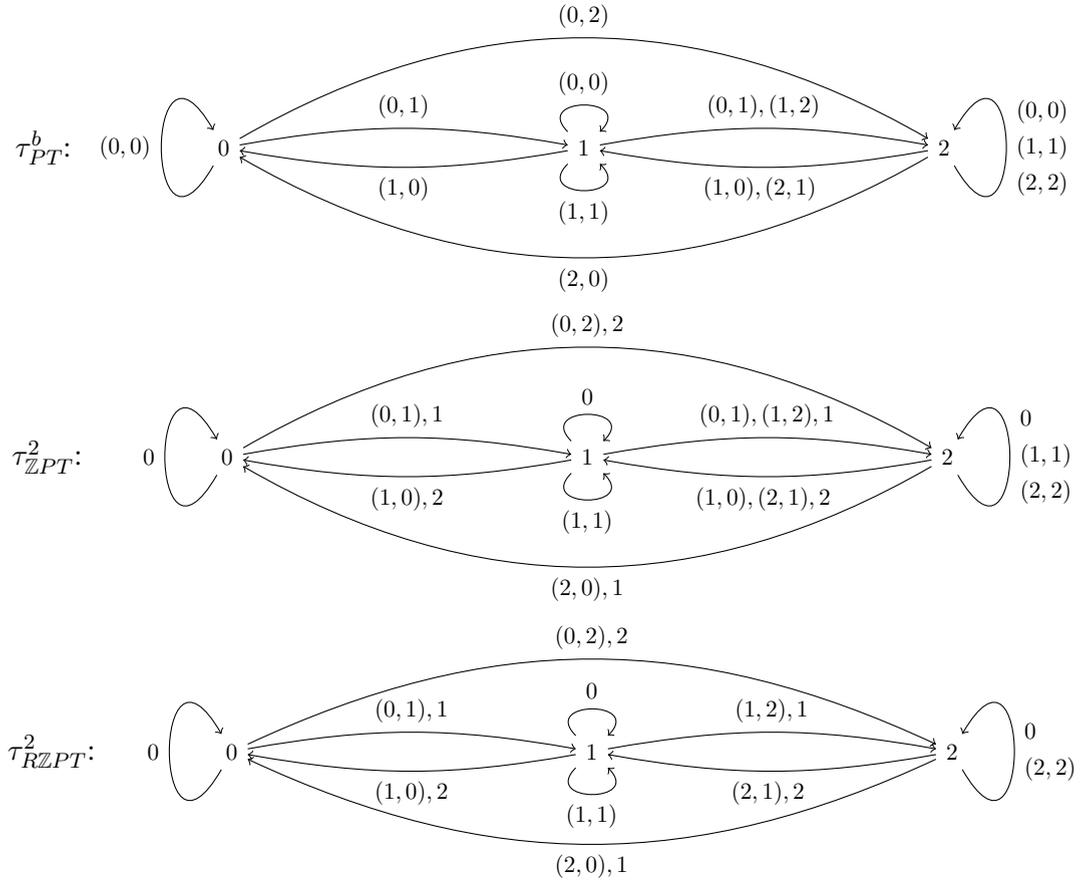

\begin{example}
Figure~\ref{fig:types} sketches $\tau_{PT}^2$ (top) and $\tau_{\mathbb{Z}PT}^2$ (middle).
Events separated by commas label different edges.
Omitting the events $(1,1)$, $(1,2)$, $(2,1)$ and $(2,2)$ and the corresponding edges yields $\tau_{PPT}^2$ and  $\tau_{\mathbb{Z}PPT}^2$, respectively.
Moreover, Figure~\ref{fig:types} sketches $\tau_{R\mathbb{Z}PT}^2$ (bottom).
\end{example}

\begin{example}\label{ex:tau_net}
Figure~\ref{fig:admissible_set} sketches the $\tau_{PPT}^2$-net $N_1$ and its reachability graph $A_{N_1}$.
$N_1$ has places $R_1$ and $R_2$, transitions $a$ and $b$ and flow $f(R_1,a)=(1,0)$, $f(R_2,b)=(1,0)$ and $f(R_1,b)=f(R_2,a)=(0,0)$ and initial marking $M_0(R_1)=M_0(R_2)=1$.
The $(0,0)$-labeled edges are omitted.
\end{example}

\begin{figure}[H]
\vspace*{-2mm}
\begin{center}
\begin{tikzpicture}[new set = import nodes]
\begin{scope}[yshift=-1cm,nodes={set=import nodes}]
		\node (T) at (-1.25,0.5) {$A_1:$};
		\node (0) at (0,0) {\nscale{$s_0$}};
		\node (1) at (1.8,0) {\nscale{$s_1$}};
		\node (2) at (0,1) {\nscale{$s_2$}};
		\node (3) at (1.8,1) {\nscale{$s_3$}};
		\draw[-latex](-0.6,0)to node[]{}(0);
\graph {(import nodes);
			0 ->[swap,"\escale{$a$}"]1->[swap, "\escale{$b$}"]3;
			0->[ "\escale{$b$}"]2->["\escale{$a$}"]3;
		};
\end{scope}

\begin{scope}[xshift=5.5cm, node distance=1.75cm,bend angle=45,auto]
\node (T) at (-1.75,-0.5) {$N_1:$};
\node [place] (R1)[tokens =1,label=left: $R_1$, node distance =1.5cm]{};
\node [transition] (a)[right of =R1]{$a$}
	 edge [pre, above] node {$\nscale{(1,0)}$}  (R1);

\node [place] (R3)[tokens =1,below of=R1, label=left: $R_2$, node distance =1cm]{};
\node [transition] (b)[right of =R3]{$b$}
	 edge [pre, above] node {$\nscale{(1,0)}$}  (R3);

\end{scope}
\begin{scope}[xshift=10.5cm, yshift=-1cm,nodes={set=import nodes}]
		\node (T) at (-1.25,0.5) {$A_{N_1}:$};
		\node (0) at (0,0) {\nscale{$11$}};
		\node (1) at (1.8,0) {\nscale{$01$}};
		\node (2) at (0,1) {\nscale{$10$}};
		\node (3) at (1.8,1) {\nscale{$00$}};
		\draw[-latex](-0.6,0)to node[]{}(0);
\graph {(import nodes);
			0 ->[swap, "\escale{$a$}"]1->[swap, "\escale{$b$}"]3;
			0->["\escale{$b$}"]2->["\escale{$a$}"]3;
		};
\end{scope}
\end{tikzpicture}
\end{center}\vspace*{-2mm}
\caption{The TS $A_1$, the $\tau_{PPT}^1$-Net $N_1$ and the reachability graph $A_{N_1}$ of $N_1$.}\label{fig:admissible_set}
\end{figure}
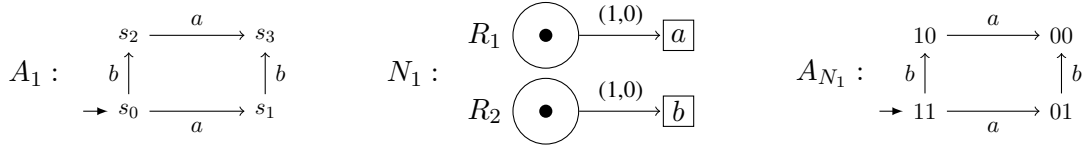

According to Definition~\ref{def:tau_nets}, for every $\tau$-net $N$, there is always a TS $A_N$, that reflects the global behavior of $N$, namely the corresponding reachability graph.
Moreover, by firing all possible sequences of transitions, the reachability graph $A_N$ can be computed effectively.
Naturally, this raises the question whether a given TS $A$ corresponds to the behavior of a $\tau$-net $N$.
Furthermore, in case of a positive decision, $N$ should be constructed.
This is the subject of the following search problem:

\noindent
\fbox{\begin{minipage}[t][1.7\height][c]{0.97\textwidth}
\begin{searchproblem}
  \problemtitle{\textsc{$\tau$-Synthesis}}
  \probleminput{A TS $A=(S,E,\delta, \iota)$.}
  \problemquestion{Find a $\tau$-net $N$ whose reachability graph is isomorphic to $A$ if it exists.}
\end{searchproblem}
\end{minipage}}
\bigskip

If an input $A=(S,E,\delta,\iota)$ of $\tau$-Synthesis allows a positive decision, then we want to construct a corresponding $\tau$-net $N$ purely from $A$.
Since $A$ and the reachability graph $A_N$ of $N$ shall be isomorphic, the events $E$ of $A$ become transitions of $N$.
The places, the flow function and the initial marking of $N$ originate from so-called $\tau$-regions of $A$.

\begin{definition}[$\tau$-Regions]\label{def:region}
Let $\tau\in \{\tau^b_0,\tau^b_1,\tau^b_2,\tau^b_3,\tau^b_4\}$ and $A=(S,E,\delta,\iota)$ be a TS.
A $\tau$-region of $A$ is a pair $(sup, sig)$ of \emph{support} $sup: S \rightarrow S_\tau $ and \emph{signature} $sig: E\rightarrow E_\tau $ such that for every edge $s \edge{e} s'$ of $A$ the image $sup(s) \ledge{sig(e)} sup(s')$ is present in $\tau$.
If $sig(e)=(m,n)$, then we define $sig^-(e)=m$ and $sig^+(e)=n$ and $\vert sig(e)\vert =0$, and if $sig(e)\in \{0,\dots, b\}$, then we define $sig^-(e)=sig^+(e)=0$ and $\vert sig(e)\vert =sig(e)$.
\end{definition}

A region $(sup, sig)$ models a place $p$ and its initial marking $M_0(p)$ as well as the corresponding part of the flow function $f(p,\cdot)$ of a sought $\tau$-net if it exist.
In particular, $sig(e)$ models $f(p,e)$ and $sup(\iota)$ models the number of tokens that $p$ contains initially and, more generally, $sup(s)$ models the number of tokens $M(p)$ in the marking $M$ that corresponds to the state $s$ according to the isomorphism $\varphi$ that justifies $A\cong A_N$.

\begin{definition}[Synthesized net]\label{def:synthesized_net}
Every set $\mathcal{R} $ of $\tau$-regions of $A=(S,E,\delta,\iota)$ defines the \emph{synthesized $\tau$-net} $N^{\mathcal{R}}_A=(\mathcal{R}, E, f, M_0)$ with set of places $\mathcal{R}$, set of transitions $E$, flow function $f((sup, sig),e)=sig(e)$ and initial marking $M_0((sup, sig))=sup(\iota)$ for all $(sup, sig)\in \mathcal{R}$ and all $e\in E$.
\end{definition}

To make sure that a synthesized net $N$ realizes the behavior of a TS exactly, distinct states $s$ and $s'$ of $A$ must correspond to different markings $M$ and $M'$ of the net.
Moreover, the firing of a transition $e$ needs to be inhibited at a marking $M$, when the event $e$ does not occur at the state $s$ that corresponds to $M$ by the isomorphism $\varphi$.
This is stated by so-called separation atoms and separation properties.

\begin{definition}[$\tau$-State Separation]\label{def:state_separation}
Let $\tau$ be a type of nets and $A=(S,E,\delta,\iota)$ a TS.
A pair $(s, s')$ of distinct states of $A$ defines a \emph{state separation atom} (SSA, for short).
A $\tau$-region $R=(sup, sig)$ \emph{solves} $(s,s')$ if $sup(s)\not=sup(s')$.
The meaning of $R$ is to ensure that $N^{\mathcal{R}}_A$ contains at least one place $R$ such that $M(R)\not=M'(R)$ for the markings $M$ and $M'$ corresponding to $s$ and $s'$, respectively.
If $s\in S$ is a state of $A$ and, for all states $s'\in S$ such that $s'\not=s$, there is a $\tau$-region that solves $(s,s')$ then $s$ is called \emph{$\tau$-solvable}.
If every state of $A$ or, equivalently, every SSA of $A$ is $\tau$-solvable, then $A$ has the \emph{$\tau$-state separation property} ($\tau$-SSP, for short).
\end{definition}

\begin{definition}[$\tau$-Event State Separation]\label{def:event_state_separation}
Let $\tau$ be a type of nets and $A=(S,E,\delta,\iota)$ a TS.
A pair $(e,s)$ of event $e\in E $ and state $s\in S$ where $e$ does not occur at $s$, that is $\neg s\edge{e}$, defines an \emph{event state separation atom} (ESSA atom, for short).
A $\tau$-region $R=(sup, sig)$ \emph{solves} $(e,s)$ if $sig(e)$ is not defined at $sup(s)$ in $\tau$, that is, $\neg \delta_\tau(sup(s), sig(e))$.
The meaning of $R$ is to ensure that there is at least one place $R$ in $N^{\mathcal{R}}_A$ such that $\delta_\tau(M(R),f(R,e))$ is not defined for the marking $M$ that corresponds to $s$ via the isomorphism, that is, $e$ cannot fire in $M$.
If, for all $s\in S$ such that $\neg s\edge{e}$, there is a $\tau$-region that solves $(e,s)$, then $e$ is called \emph{$\tau$-solvable}.
If every event of $A$ or, equivalently, every ESSA of $A$ is $\tau$-solvable, then $A$ has the \emph{$\tau$-event state separation property} ($\tau$-ESSP, for short).
\end{definition}

\begin{definition}[Witness, $\tau$-admissible set]\label{def:witness}
A set $\mathcal{R}$ of $\tau$-region is a ($\tau$-) \emph{witness} of the $\tau$-(E)SSP of $A$ if it contains for every (E)SSA a $\tau$-region that solves it.
If $A$ has the $\tau$-SSP and the $\tau$-ESSP, then $A$ is called $\tau$-solvable.
A set $\mathcal{R}$ that is a witness of both the $\tau$-SSP and the $\tau$-ESSP of $A$ is called $\tau$-\emph{admissible}.
\end{definition}

The following lemma, borrowed from \cite[p.163]{DBLP:series/txtcs/BadouelBD15}, summarizes the already implied connection between the existence of $\tau$-admissible sets of $A$ and (the solvability of) $\tau$-synthesis:
\begin{lemma}[\cite{DBLP:series/txtcs/BadouelBD15}]\label{lem:admissible}
Let $A$ be a TS and $\tau$ a type of nets.
The reachability graph $A_N$ of a $\tau$-net $N$ is isomorphic to $A$ if and only if there is a $\tau$-admissible set $\mathcal{R}$ of $A$ such that $N=N^{\mathcal{R}}_A$.
\end{lemma}

\begin{example}\label{ex:admissible_set}
Let $\tau\in \{\tau_{PPT}^b,\tau_{PT}^b \mid b\in \mathbb{N}^+ \}$.
The TS $A_1$ of Figure~\ref{fig:admissible_set} has the $\tau$-ESSP and the $\tau$-SSP:
The region $R_1=(sup_1,sig_1)$, which is defined by $sup_1(s_0)=sup_1(s_2)=1$ and $sup_1(s_1)=sup_1(s_3)=0$ and $sig(a)=(1,0)$ and $sig(b)=(0,0)$, solves the ESSA $(a,s_1)$ and $(a,s_3)$ as well as the SSA $(s_0,s_1)$ and $(s_0,s_3)$ and $(s_2,s_1)$ and $(s_2,s_3)$.
Moreover, the region $R_2=(sup_2,sig_2)$, which is defined by $sup_2(s_0)=sup_2(s_1)=1$, $sup_2(s_2)=sup_2(s_3)=0$, $sig(a)=(0,0)$ and $sig(b)=(1,0)$ solves the remaining ESSA $(b,s_2)$ and $(b,s_3)$ as well as the SSA $(s_0,s_2)$ and $(s_1,s_3)$ of $A_1$.
Since $R_1$ and $R_2$ solve all SSA and ESSA, $\mathcal{R}=\{R_1,R_2\}$ is a $\tau$-admissible set.
Figure~\ref{fig:admissible_set} sketches the synthesized net $N_1=N^{\mathcal{R}}_{A_1}$, where $(0,0)$-labeled flow edges are omitted, and its reachability graph $A_{N_1}$.
The isomorphism $\varphi$ between $A_1$ and $A_{N_1}$ is given by $\varphi(s_0)=11$, $\varphi(s_1)=01$, $\varphi(s_2)=10$ and $\varphi(s_3)=00$.
\end{example}

\begin{figure}[H]
\vspace*{-7mm}
\begin{center}
\begin{tikzpicture}[new set = import nodes]

\begin{scope}
\node () at (-2.75,0) {$ A_2:$};
\node (s2) at (0,0) {\nscale{$s_2$}};
\node (s0) at (-1.4, -0.75) {\nscale{$s_0$}};
\node (s1) at (-1.4,0.75) {\nscale{$s_1$}};
\draw[-latex](-2,-0.75)to node[]{}(s0);

\path (s0) edge [->, bend left =20] node[pos=0.4,left] {\escale{$a$}} (s1);
\path (s1) edge [->, bend left =20] node[pos=0.6,above] {\escale{$a$}} (s2);
\path (s2) edge [->, bend left =20] node[pos=0.4,below] {\escale{$a$}} (s0);
\end{scope}

\begin{scope}[xshift=3.75cm, node distance=1.5cm,bend angle=45,auto]
\node (T) at (-1.75,0) {$N_2:$};
\node [place] (R1)[label=left: $R$, node distance =1.5cm]{};
\node [transition] (a)[right of =R1]{$a$}
	 edge [pre, above] node {$\nscale{1}$}  (R1);

\end{scope}
\begin{scope}[xshift=9.75cm,nodes={set=import nodes}]
\node () at (-2.75,0) {$ A_{N_2}:$};
\node (s2) at (0,0) {\nscale{$2$}};
\node (s0) at (-1.4, -0.75) {\nscale{$0$}};
\node (s1) at (-1.4,0.75) {\nscale{$1$}};
\draw[-latex](-2,-0.75)to node[]{}(s0);

\path (s0) edge [->, bend left =20] node[pos=0.4,left] {\escale{$a$}} (s1);
\path (s1) edge [->, bend left =20] node[pos=0.6,above] {\escale{$a$}} (s2);
\path (s2) edge [->, bend left =20] node[pos=0.4,below] {\escale{$a$}} (s0);
\end{scope}
\end{tikzpicture}
\end{center}\vspace*{-3mm}
\caption{The TS $A_2$, the $\tau_{\mathbb{Z}PPT}^2$-Net $N_2$ and the reachability graph $A_{N_2}$ of $N_2$.}\label{fig:admissible_set_2}\vspace*{-2mm}
\end{figure}
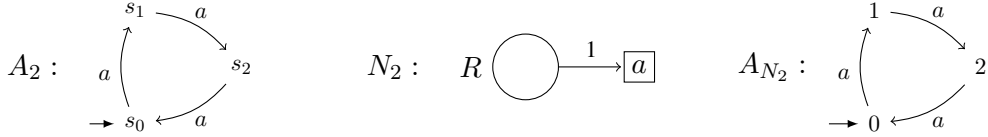

\begin{example}\label{ex:admissible_set_2}
The TS $A_2$ of Figure~\ref{fig:admissible_set_2} has no ESSA, since the only event $a$ occurs at every state of $A_2$.
Consequently, $A_2$ has the $\tau$-ESSP for all types of nets.
However, $A_2$ has the SSA $(s_0,s_1), (s_0,s_2)$ and $(s_1,s_2)$.
If $\tau\in \{\tau_{PPT}^b,\tau_{PT}^b \mid b\in \mathbb{N}^+ \}$, then neither of these atoms is $\tau$-solvable, since every $\tau$-region $R=(sup, sig)$ of $A_2$ satisfies $sup(s_0)=sup(s_0)-2sig^-(a)+2sig^+(a)$, which implies $sig(a)=(0,0)$ and thus $sup(s_0)=sup(s_1)=sup(s_2)$.
Nevertheless, if $\tau\in \{\tau_{\mathbb{Z}PPT}^b,\tau_{\mathbb{Z}PT}^b \mid b\geq 2 \}$, then $A_2$ has the $\tau$-SSP, since the following $\tau$-Region $R=(sup, sig)$ solves all SSA in one blow:
$sup(s_0)=0$, $sup(s_1)=1$ and $sup(s_2)=2$ and $sig(a)=1$.
Since $A_2$ has also the $\tau$-ESSP, $\mathcal{R}=\{R\}$ is a $\tau$-admissible set of $A_2$.

Figure~\ref{fig:admissible_set_2} sketches the synthesized net $N_2=N^{\mathcal{R}}_{A_2}$ and its reachability graph $A_{N_2}$.
This example also shows that the group-extended types $\tau_{\mathbb{Z}PPT}$ and $\tau_{\mathbb{Z}PT}$ are strictly more powerful than the types $\tau_{PPT}$ and $\tau_{PT}$.
\end{example}

A purpose of this paper is to characterize the computational complexity of $\tau$-\textsc{Synthesis} for all introduced $b$-bounded types of nets completely.
Since the corresponding complexity classes are defined for decision problems, we restrict our investigations to the decision version of $\tau$-\textsc{Synthesis} that is called $\tau$-\textsc{Solvability}.
By Lemma~\ref{lem:admissible}, there is a $\tau$-admissible set $\mathcal{R}$ of $A$ if and only if there is a $\tau$-net $N$ whose reachability graph is isomorphic to $A$.
This allows us to formulate the solvability problem for $\tau$-nets as follows:

\smallskip\noindent
\fbox{\begin{minipage}[t][1.7\height][c]{0.97\textwidth}
\begin{decisionproblem}
  \problemtitle{\textsc{$\tau$-Solvability}}
  \probleminput{A TS $A=(S,E,\delta, \iota)$.}
  \problemquestion{Does there exist a $\tau$-admissible set $\mathcal{R}$ of $A$?}
\end{decisionproblem}
\end{minipage}}
\bigskip

Although we are mainly interested in synthesis, the $\tau$-SSP and the $\tau$-ESSP are also interesting on their own.
This is because, for example, an algorithm that decides in polynomial time whether $A$ has the $\tau$-SSP or the $\tau$-ESSP could serve as a pre-synthesis method, which rejects inputs that does not have the property in question.
This leads to the following decision problems:

\smallskip\noindent
\fbox{\begin{minipage}[t][1.7\height][c]{0.97\textwidth}
\begin{decisionproblem}
  \problemtitle{\textsc{$\tau$-SSP}}
  \probleminput{A TS $A=(S,E,\delta, \iota)$.}
  \problemquestion{Does there exist a witness $\mathcal{R}$ for the $\tau$-SSP of $A$?}
\end{decisionproblem}
\end{minipage}}
\bigskip

\noindent
\fbox{\begin{minipage}[t][1.7\height][c]{0.97\textwidth}
\begin{decisionproblem}
  \problemtitle{\textsc{$\tau$-ESSP}}
  \probleminput{A TS $A=(S,E,\delta, \iota)$.}
  \problemquestion{Does there exist a witness $\mathcal{R}$ for the $\tau$-ESSP of $A$?}
\end{decisionproblem}
\end{minipage}}
\bigskip

In~\cite{DBLP:journals/tcs/BadouelBD97}, it was originally shown that \textsc{$\tau_{PPT}^1$-Solvability} (there referred to as \emph{elementary net synthesis}) is NP-complete.
In~\cite{DBLP:conf/apn/TredupRW18,DBLP:conf/concur/TredupR18}, we have shown that this remains true even for strongly restricted inputs and applies also to \textsc{$\tau_{PPT}^1$-SSP} and \textsc{$\tau_{PPT}^1$-ESSP}.
Moreover, the type $\tau_{\mathbb{Z}PPT}^1$ coincides with Schmitt's type (\emph{flip-flop nets}) for which the considered decision problems are tractable~\cite{DBLP:conf/stacs/Schmitt96}.
In \cite[p.~619]{DBLP:conf/tamc/TredupR19}, this characterization was found to be true for $\tau_{\mathbb{Z}PT}^1$ (there referred to as the Boolean type of nets $\tau=\{\nop,\inp,\out,\used,\swap\}$) as well.
In this paper, we complete the complexity characterization of $\tau$-\textsc{Solvability}, $\tau$-\textsc{SSP} and $\tau$-\textsc{ESSP} for all introduced $b$-bounded types of nets and all $b\in \mathbb{N}^+$.
(Observe that the problems are trivial if $b=0$.)
Figure~\ref{fig:overview} provides an overview over our findings and shows, depending on $\tau$ and $b$, which of the problems are NP-complete (NPC) and which are solvable in polynomial time (P).

\begin{figure}[H]
\centering
\newcommand{\tEntry}[3]{
\fill[#3] (#1) +(-0.5,-0.25) rectangle +(0.5,0.25);
\node at (#1) {#2};
}
\begin{tikzpicture}[scale=1.3]
\tikzstyle{P} = [black!32]
\tikzstyle{NPC}=[black!08]
\tEntry{-3.25,1}{Bound }{white}
\tEntry{-1.625,1}{Problem }{white}
\tEntry{0,1}{$\tau_{PPT}^b$}{white}
\tEntry{1,1}{$\tau_{PT}^b$}{white}
\tEntry{2,1}{$\tau_{\mathbb{Z}PPT}^b$}{white}
\tEntry{3,1}{$\tau_{\mathbb{Z}PT}^b$}{white}
\tEntry{4,1}{$\tau_{R\mathbb{Z}PT}^b$}{white}
%
\foreach \j in {1} {\tEntry{-3.25,-1.5*\j+1.5}{$b=1$}{white}}
\foreach \j in {2} {\tEntry{-3.25,-1.5*\j+1.5}{$b\geq 2$}{white}}
\foreach \j in {1,...,2} {\tEntry{-1.625,-1.5*\j+2}{\textsc{SSP}}{white}}
\foreach \j in {1,...,2} {\tEntry{-1.625,-1.5*\j+1.5}{\textsc{ESSP}}{white}}
\foreach \j in {1,...,2} {\tEntry{-1.625,-1.5*\j+1}{\textsc{Solvability}}{white}}
%
\foreach \j in {-4,...,1} {\tEntry{0,0.5*\j}{NPC}{NPC}}
%
\foreach \j in {-4,...,1} {\tEntry{1,0.5*\j}{NPC}{NPC}}
\foreach \j in {-2,...,1} {\tEntry{2,0.5*\j}{P}{P}}
\foreach \j in {-4,...,-3} {\tEntry{2,0.5*\j}{NPC}{NPC}}
\foreach \j in {-2,...,1} {\tEntry{3,0.5*\j}{P}{P}}
\foreach \j in {-4,...,-3} {\tEntry{3,0.5*\j}{NPC}{NPC}}
\foreach \j in {-4,...,1} {\tEntry{4,0.5*\j}{P}{P}}
\tikzstyle{dottedline} = [black!64,dotted]
\tikzstyle{boldline} = [black,line width = 1pt]
\draw[-,boldline] (-3.75,0.75) -- ++(8.25,0);
\foreach \j in {0,-1.5} {
\draw[boldline] (-3.75,\j-0.75) -- ++(8.25,0);
\draw[dottedline] (-2.5,\j-0.25) -- ++(7,0);
\draw[dottedline] (-2.5,\j+0.25) -- ++(7,0);
}
\draw[-,boldline] (-2.75,1.25) -- ++(0,-3.5);
\foreach \i in {0,...,4} {\draw[boldline] (\i-0.5,1.25) -- ++(0,-3.5);}
\end{tikzpicture}
\caption{Overview of the computational complexity of \textsc{$\tau$-Solvability}, \textsc{$\tau$-SSP} and \textsc{$\tau$-ESSP} for all $\tau\in \{\tau_{PPT}^b,\tau_{PT}^b, \tau_{\mathbb{Z}PPT}^b, \tau_{\mathbb{Z}PT}^b,\tau_{R\mathbb{Z}PT}^b\}$ and all $b\in \mathbb{N}^+$.}
\label{fig:overview}
\end{figure}

In the following, if not explicitly stated otherwise, for all $\tau\in \{\tau_{PT}^b,\tau_{PPT}^b\}$, we let $b\in \mathbb{N}^+$ and, for all $\tau\in \{\tau_{\mathbb{Z}PT}^b,\tau_{\mathbb{Z}PPT}^b\}$, we let $2\leq b\in \mathbb{N}$ be arbitrary but fixed, since the case $b=1$ is already solved for the latter.
The observations of the next lemma are used to simplify our proofs:

\begin{lemma}\label{lem:observations} Let $\tau \in \{\tau_{PPT}^b, \tau_{PT}^b, \tau_{\mathbb{Z}PT}^b, \tau_{\mathbb{Z}PPT}^b\}$ and $A=(S,E,\delta,\iota)$ be a TS.
\begin{enumerate}
\item\label{lem:sig_summation_along_paths}
Two mappings $sup: S\longrightarrow S_\tau$ and $sig: E\longrightarrow E_\tau$ define a $\tau$-region of $A$ if and only if for every directed labeled path $q_0\edge{e_1}\dots\edge{e_m}q_m$ of $A$ holds $sup(q_{i})=sup(q_{i-1})-sig^-(e_i)+sig^+(e_i)+\vert sig(e_i)\vert$ for all $i\in \{1,\dots, \ell\}$, where this equation is to consider modulo $(b+1)$.
In particular, every region $(sup, sig)$ is implicitly completely defined by $sig$ and $sup(\iota)$.
\item\label{lem:absolute_value}
If $s_{0}, s_{1},\dots, s_{b}\in S$, $e\in E$ and $s_{0}\edge{e} \dots \edge{e} s_b$ then a $\tau$-region $(sup, sig)$ of $A$ satisfies $sig(e)= (m,n)$ with $m\not=n$ if and only if $(m,n) \in \{(1,0),(0,1)\}$.
If $sig(e)=(0,1)$ then $sup(s_{0})=0$ and $sup(s_b)=b$.
If $sig(e)=(1,0)$ then $sup(s_0)=b$ and $sup(s_b)=0$.
\end{enumerate}
\end{lemma}
\begin{proof}
(\ref{lem:sig_summation_along_paths}):
The first claim follows directly from the definitions of $\tau$ and $\tau$-regions.
For the second claim, we observe that every state $s\in S$ is reachable by a directed labeled path $q_0\edge{e_1}\dots\edge{e_m}q_m$, where $q_0=\iota$ and $q_m=s$.
Thus, if $sup(\iota)$ and a valid signature $sig$ are given, then, by the first claim, we get $sup(s)$ by $sup(s)=sup(\iota)+ \sum_{i=1}^{m} (-sig^-(e_i)+sig^+(e_i)+\vert sig(e_i)\vert)$.

(\ref{lem:absolute_value}):
The \textit{If}-direction is trivial.
For the \textit{Only-if}-direction we show that the assumption $(m,n)\not\in \{(1,0),(0,1)\}$ yields a contradiction.

\medskip
By (\ref{lem:sig_summation_along_paths}), we have that $sup(s_b)=sup(s_0) + b\cdot(n-m)$.
If $\vert n-m\vert > 1$, then either $b\cdot(n-m) < - b$ or $b\cdot(n-m) > b$;
since $0\leq sup(s_0) \leq b$, the first case contradicts $sup(s_b)\geq 0$, and the latter case contradicts $\sup(s_b)\leq b$, respectively.
Hence, if $n\not=m$ then $\vert n-m\vert =1$.
For a start, we show that $ m > n$ implies $m=1$ and $n=0$.
By $n \leq m-1$ and $sup(s_0)\leq b$ we obtain the estimation
\[
sup(s_{b-1}) = sup(s_0) +(b-1)(n-m)\ \leq \ b + (b-1)(m-1-m) = 1
\]
By $n < m \leq sup(s_{b-1})\leq 1$ we have $(m,n)=(1,0)$.
Similarly, we obtain that $(m,n)=(0,1)$ if $m < n$.
Hence, if $sig(e)=(m,n)$ and $n\not=m$ then $sig(e) \in \{(1,0),(0,1)\}$.

The second statement follows directly from (\ref{lem:sig_summation_along_paths}).
\end{proof}


The following lemma shows that if $A$ is a linear TS and $\mathcal{R}$ is a witness of the $\tau$-ESSP of $A$, then $\mathcal{R}$ witnesses also the $\tau$-SSP of $A$.
In particular, this implies that a linear TS $A$ is $\tau$-solvable if and only if it has the $\tau$-ESSP.
Notice that Lemma~\ref{lem:essp_implies_ssp} provides a very general result, since its statement is independent from the actual choice of $\tau$.

\begin{lemma}[ESSP implies SSP for Linear TS]\label{lem:essp_implies_ssp}
Let $\tau$ be a type of nets and let $ A= s_0 \edge{e_1}  \dots  \edge{e_n} s_n$ be a a linear TS and let $\mathcal{R}$ be a set of $\tau$-regions of $A$.
If $\mathcal{R}$ is a witness of the $\tau$-ESSP of $A$, then $\mathcal{R}$ witnesses also the $\tau$-SSP of $A$.
\end{lemma}
\begin{proof}
Let $\mathcal{R}$ be a witness of the $\tau$-ESSP of $A$.
Assume that there is an SSA that can not be solved by a region of $\mathcal{R}$.
Then there is an SSA $\alpha=(z_{i_j}, z_{i_k})$ of $A$, where $i_j,i_k\in \{0,\dots, n\}$, such that the state $z_{i_k}$ has the maximum index among all states of $A$ that participate at SSA of $A$ that can not be solved by a region of $\mathcal{R}$:
if $(z_{i_\ell}, z_{i_m})$ is an SSA of $A$ that can not be solved by regions of $\mathcal{R}$, then $i_\ell\leq i_k$ and $i_m\leq i_k$.
In particular, this implies $i_j < i_k$.
Since $i_j < i_k$, there is the edge $z_{i_j}\ledge{e_{i_j+1}}z_{i_j+1}$ in $A$.
Since $\alpha$ is not solvable by regions of $\mathcal{R}$, we have $sup(z_{i_j})=sup(z_{i_k})$ for all $(sup, sig)\in \mathcal{R}$.
This implies, that the event $e_{i_j+1}$ occurs at $z_{i_k}$, since the ESSA  $(e_{i_j+1}, z_{i_k})$ would not be solvable otherwise:
$sup(z_{i_j})\lledge{sig(e_{i_j+1})}$ and $\neg sup(z_{i_k})\lledge{sig(e_{i_j+1})}$ implies the contradiction $sup(z_{i_j})\not=sup(z_{i_k})$.
Hence, $z_{i_k}\ledge{e_{i_j+1}}z_{i_k+1}$ is an edge in $A$.
Since $sup(z_{i_j})=sup(z_{i_k})$ for all $(sup, sig)\in \mathcal{R}$ and since $\delta_\tau$ is a function, by $z_{i_j}\ledge{e_{i_j+1}}z_{i_j+1}$  and $z_{i_k}\ledge{e_{i_j+1}}z_{i_k+1}$, we get $sup(z_{i_j+1})=sup(z_{i_k+1})$ for all $(sup, sig)\in \mathcal{R}$.
In particular, the SSA $(z_{i_j+1}, z_{i_k+1})$ is not solvable by regions of $\mathcal{R}$.
Since $i_k < i_k+1$, this contradicts the choice of  $\alpha$.
Consequently, $\alpha$ does not exist and $\mathcal{R}$ witnesses the $\tau$-SSP of $A$.
\end{proof}

\section{The concept of unions}\label{sec:unions}%

For our reductions, we use the technique of \emph{component design} \cite{DBLP:books/fm/GareyJ79}.
Every implemented constituent is a TS (in the context of the reduction also referred to as gadget) that locally ensures the satisfaction of some constraints.
Commonly, all constituents are finally joined together in a target instance (TS) such that all required constraints are properly globally translated.
However, the concept of unions saves us the need to actually create the target instance:

\begin{definition}[Union]
If $A_0=(S_0,E_0,\delta_0,\iota_0), \dots, A_n=(S_n,E_n,\delta_n,\iota_n)$ are TS with pairwise disjoint states (but not necessarily disjoint events) then we call $U(A_0, \dots, A_n)$ their \emph{union} with set of states $S(U)=\bigcup_{i=0}^n S_i$ and set of events $E(U)=\bigcup_{i=0}^n E_i$.
\end{definition}

Let $\tau=(Z_\tau,E_\tau,\delta_\tau)$ be a type of nets and $U=U(A_0, \dots, A_n)$ a union, where $A_i=(S_i,E_i,\delta_i,\iota_i)$ for all $i\in\{0,\dots,n\}$.
The concepts of SSA, ESSA, $\tau$-regions, $\tau$-SSP, and $\tau$-ESSP as defined in the preliminaries are transferred to $U$ as follows:

\begin{definition}[Region of a Union]
A pair $(sup, sig)$ of mappings $sup: S(U)\rightarrow S_\tau$ and $sig:E(U)\rightarrow E_\tau$ is called a \emph{$\tau$-region} (of $U$), if $s\edge{e}s'\in A_i$ implies $sup(s)\ledge{sig(e)}sup(s')\in \tau$ for all $i\in \{0,\dots, n\}$.
\end{definition}

\begin{definition}[$\tau$-State Separation in Unions]
A pair $(s,s')$ of distinct states $s, s' \in S(U)$ of the \emph{same} TS $A_i$, where $i\in \{0,\dots, n\}$, defines an SSA of $U$.
A  $\tau$-region $(sup,sig)$ of $U$ solves $(s,s')$, if $sup(s) \not= sup(s')$.
$U$ has the $\tau$-SSP, if all of its SSA are $\tau$-solvable.
\end{definition}

\begin{definition}[$\tau$-Event State Separation in Unions]
A pair $(e,s)$ of event $e\in E(U)$ and state $s\in S(U)$ such that $\neg s\edge{e}$ defines an ESSA of $U$.
A $\tau$-region $(sup, sig)$ of $U$ solves it, if $\neg sup(s)\edge{e}$.
$U$ has the $\tau$-ESSP if all of its ESSA are $\tau$-solvable.
\end{definition}

In the same way, the notion of \emph{witness} and \emph{$\tau$-admissible set} and \emph{$\tau$-solvable} are transferred to unions.
From the perspective of $\tau$-SSP and $\tau$-ESSP, unions are intended to treat a lot of unjoined TS as if they were joined to a TS.
To be able to do so, in the following, we introduce the \emph{linear joining} $LJ(U)$ and the \emph{joining} $J(U)$ of a union $U$ and argue that $LJ(U)$ or $J(U)$ has the $\tau$-(E)SSP if and only if $U$ has the $\tau$-(E)SSP.

\begin{definition}[Linear Joining]
Let $U = U(A_0, \dots, A_n)$ be a union such that, for all $i\in \{0,\dots, n\}$, the TS $A_i=(S_i,E_i,\delta_i,\iota_i)$ is linear and its terminal state is $t_i$ and let $Q=\{q_1,\dots, q_n\}$ be a set of states, which is disjoint with $S(U)$, and $W=\{w_1,\dots, w_n\}$ and $Y=\{y_1,\dots, y_n\}$ be sets of events which are disjoint with $E(U)$.
The \emph{linear joining} of $U$ is the linear TS $LJ(U)=(S(U)\cup Q,E(U)\cup W\cup Y,\delta, \iota_0)$ with transition function $\delta$ that is, for all $e\in E(U)\cup W\cup Y$ and all $s\in S(U)\cup Q$, defined as follows:
\[\delta(s,e)=
\begin{cases}
\delta_i(s,e), &\text{if } s\in S_i \text{ and } e\in E_i \text{ and } i\in \{0,\dots, n\}\\
q_{i+1}, &\text{if } s=t_i \text{ and } e=w_{i+1} \text{ and } i\in \{0,\dots, n-1\}\\
\iota_i, &\text{if } s=q_i \text{ and } e=y_i \text{ and } i\in \{1,\dots, n\}\\
\text{undefined}, &\text{otherwise}
\end{cases}
\]
\end{definition}

\begin{remark}
The linear joining $LJ(U)$ of $U$ can be sketched as follows:
\begin{center}
\begin{tikzpicture}[new set = import nodes]
\begin{scope}[nodes={set=import nodes}]%
		\node (init) at (-1,0) {$LJ(U)=$};
		\foreach \i in {0,...,4} { \coordinate (\i) at (\i*1.4cm,0) ;}
		\node (0) at (0) {$A_0$};
		\node (1) at (1) {\nscale{$q_1$}};
		\node (2) at (2) {$A_1$};
		\node (3) at (3) {\nscale{$q_2$}};
		\node (4) at (4) {};
		\node (5) at (6.2,0) {$\dots$};
		\node (6) at (6.8,0) {};
		\node (7) at (8.2,0) {\nscale{$q_n$}};
		\node (8) at (9.6,0) {$A_n$};
\graph {
	(import nodes);
			0 ->["\escale{$w_1$}"]1->["\escale{$y_1$}"]2 ->["\escale{$w_2$}"]3 ->["\escale{$y_2$}"]4;
			6  ->["\escale{$w_n$}"]7->["\escale{$y_n$}"]8;
};
\end{scope}
\end{tikzpicture}
\end{center}
\end{remark}
%

\begin{definition}[Joining]
Let $U = U(A_0, \dots, A_n)$ be a union of TS $A_i=(S_i,E_i,\delta_i,\iota_i)$ for all $i\in \{0,\dots, n\}$, and let $Q=\{q_0,\dots, q_n\}$ be a set of states, which is disjoint with $S(U)$, and $W=\{w_1,\dots, w_n\}$ and $Y=\{y_0,\dots, y_n\}$ be sets of events, which are disjoint with $E(U)$.
The \emph{joining} of $U$ is the TS $J(U) = (S(U) \cup Q, E(U) \cup W \cup Y , \delta, q_0 )$ with transition function $\delta$ that is, for all $e\in E(U)\cup W\cup Y$ and all $s\in S(U)\cup Q$, defined as follows:
\[\delta(s,e) =
\begin{cases}
\delta_i(s,e), & \text{if } s \in S_i \text{ and } e \in E_i \text{ and }  i\in \{0,\dots, n\}\\
q_{i+1}, & \text{if } s = q_i \text{ and } e=w_{i+1} \text{ and }  i\in \{0,\dots, n-1\}\\
\iota_i, & \text{if } s = q_i \text{ and } e=y_i \text{ and } i\in \{0,\dots, n\}\\
\text{undefined}, &\text{otherwise}
\end{cases}
\]
\end{definition}
\begin{remark}
The joining $J(U)$ of $U$ can be sketched as follows:
\begin{center}
\begin{tikzpicture}
\node (t0) at (0,0) {\nscale{$q_0$}};
\node (t1) at (1.5,0) {\nscale{$q_1$}};
\coordinate (t2) at (3,0) ;
\node (dots_1) (d1) at (3.5,0) {$\dots$} ;
\coordinate (t_n_1) at (4,0) ;
\node (tn) at (5.5,0) {\nscale{$q_n$}};
\node (a0) at (0,-1.1) {$A_0$};
\node (a1) at (1.5,-1.1) {$A_1$};
\coordinate (a2) at (3.2,-1.1) ;
\coordinate (a_n_1) at (4.5,-1.1) ;
\node (an) at (5.5,-1.1) {$A_n$};

\path (t0) edge [->] node[pos=0.5,above] {\escale{$w_1$}} (t1);
\path (t1) edge [->] node[pos=0.5,above] {\escale{$w_2$}} (t2);
\path (t_n_1) edge [->] node[pos=0.5,above] {\escale{$w_n$}} (tn);
\path (t0) edge [->] node[pos=0.5,left] {\escale{$y_0$}} (a0);
\path (t1) edge [->] node[pos=0.5,left] {\escale{$y_1$}} (a1);
\path (tn) edge [->] node[pos=0.5,left] {\escale{$y_n$}} (an);
\end{tikzpicture}
\end{center}
\end{remark}

The following lemma proves the announced functionality of unions.
For technical reasons, we restrict ourselves to unions $U$ where for every event $e\in E(U)$ there is at least one ESSA $(e,s)$ to solve.
The unions of our reductions satisfy this property, which is used to ensure that if $U$ has the $\tau$-ESSP, then $LJ(U)$ and $J(U)$ have the $\tau$-ESSP, too.
Moreover, our reductions ensure that if $\tau\in\{\tau_{PT}^b,\tau_{PPT}^b\}$, then only the linear joining $LJ(U)$ is to consider, and if $\tau\in\{\tau_{\mathbb{Z}PT}^b,\tau_{\mathbb{Z}PPT}^b\}$, then only the joining $J(U)$ is used.
Thus, for the sake of simplicity, the lemma is formulated accordingly.
\begin{lemma}\label{lem:joining}\hfill\break

\vspace*{-7mm}
\begin{enumerate}
\item
Let $U = U(A_0, \dots, A_n)$ be a union of linear TS such that $t_i$ is the terminal state of $A_i=(S_i,E_i,\delta_i,\iota)$ for all $i\in \{0,\dots, n\}$, and  for every event $e\in E(U)$ there is a state $s\in S(U)$ with $\neg s\edge{e}$.
If $\tau\in \{\tau_{PT}^b, \tau_{PPT}^b\}$, then $U$ has the $\tau$-ESSP, respectively the $\tau$-SSP, if and only if $LJ(U)$ has the $\tau$-ESSP, respectively the $\tau$-SSP.

\item
Let $U = U(A_0, \dots, A_n)$ be a union such that $A_i=(S_i,E_i,\delta_i, \iota_i)$ for all $i\in \{0,\dots, n\}$, and for every event $e\in E(U)$ there is a state $s\in S(U)$ with $\neg s\edge{e}$.
If $\tau \in \{\tau_{\mathbb{Z}PT}^b, \tau_{\mathbb{Z}PPT}^b\}$, then $U$ has the $\tau$-ESSP, respectively the $\tau$-SSP, if and only if $J(U)$ has the $\tau$-ESSP, respectively the $\tau$-SSP.
\end{enumerate}
\end{lemma}
\begin{proof}
(1): The \emph{if}-direction is trivial.

\emph{Only-if}:
Let $R=(sup, sig)$ be a $\tau$-region of $U$, which solves an ESSA $(a,z)$ or an SSA $(z,z')$ of $U$.
We can extended $R$ to a $\tau$-region $R'=(sup', sig')$ of $LJ(U)$ that also solves these atoms, by defining $R'$ for all $s\in S(U)\cup Q$ and all $e\in E(U)\cup W\cup Y$ as follows:
\begin{align*}
sup'(s) &= \begin{cases}
sup(s), & \text{if } s \in S(U),\\
sup(z), & \text{if } s \in Q
\end{cases}\\
sig'(e) &= \begin{cases}
sig(e), & \text{if } e \in E(U),\\
(sup(t_i)-sup(z),0) & \text{if }  e = w_{i+1} \text{ and } sup(t_i)  > sup(z) \text{ and } i\in \{0,\dots, n-1\} \\
(0, sup(z)-sup(t_i)) & \text{if }  e = w_{i+1} \text{ and } sup(t_i)  \leq sup(z) \text{ and } i\in \{0,\dots, n-1\} \\
(0, sup(\iota_i)-sup(z)) & \text{if }  e = y_i \text{ and } sup(\iota_i)  > sup(z) \text{ and } i\in \{1,\dots, n\} \\
(sup(z)-sup(\iota_i),0) & \text{if }  e = y_i \text{ and } sup(\iota_i)  \leq sup(z) \text{ and } i\in \{1,\dots, n\} \\
\end{cases}
\end{align*}

Notice that this extension also $\tau$-solves $(e, q_i)$ for all $i\in \{1,\dots, n\}$, since $sup(q_i)=sup(z)$.
Since, for every event $e\in E(U)$ there is an atom $(e,s)$ to solve, this implies that all events of $U$ are $\tau$-solvable in $LJ(U)$.
Moreover, it is easy to see that the connector states $q_1,\dots, q_n$ and the connector events $y_1,\dots, y_n$ are $\tau$-solvable:
If $i\in \{1,\dots, n\}$ is arbitrary but fixed then the following region $R=(sup, sig)$ (by Lemma~\ref{lem:observations}, completely defined) $\tau$-solves $q_i$ and $y_i$:
$sup(\iota_0)=b$;
for all $e\in E(LJ(U))$, if $e=y_i$, then $sig(e)=(0,b)$;
if $e=w_i$, then $sig(e)=(b,0)$;
otherwise $sig(e)=(0,0)$.

So far, we have already proven that if $U$ has the $\tau$-SSP, then $LJ(U)$ has the $\tau$-SSP, too.
Thus, to prove that the $\tau$-ESSP of $U$ implies $\tau$-ESSP of $LJ(U)$, it remains to show that $w_1,\dots, w_n$ are solvable if $U$ has the $\tau$-ESSP.
Let $i\in \{1,\dots, n\}$ be arbitrary but fixed.
The following region $R=(sup, sig)$ solves $(w_i,s)$ for all $s\in S(LJ(U))\setminus S_{i-1}$:
if $i=1$, then $sup(\iota_0)=0$, otherwise $sup(\iota_0)=b$;
for all $e\in E(LJ(U))$, if $i\not=1$ and $e=y_{i-1}$, then $sig(e)=(b,0)$;
if $e=w_i$, then $sig(w_i)=(0,b)$;
otherwise, $sig(e)=(0,0)$.

It remains to argue that $(w_i,s)$ is $\tau$-solvable for all $s\in S_{i-1}$.
Since $U$ has the $\tau$-ESSP, there is a set $\mathcal{R}$ of regions that witnesses the $\tau$-ESSP.
In particular, for every ESSA $(e,s)$ of $A_{i-1}$ there is a region of $\mathcal{R}$ that solves it.
Restricting the corresponding regions to $A_{i-1}$ yields a set of regions that witnesses the $\tau$-ESSP of $A_{i-1}$.
Since $A_i$ is linear, by Lemma~\ref{lem:essp_implies_ssp}, these regions witness also the $\tau$-SSP of  $A_{i-1}$.
Consequently, for every state $s\in S_{i-1}\setminus\{t_{i-1}\}$, there is a region $(sup, sig)\in \mathcal{R}$ such that $sup(s)\not=sup(t_{i-1})$.
We extend this region to a region of $LJ(U)$ that solves $(w_i, s)$ as follows:
Besides of $w_i, q_i$ and $y_i$, the extension of $(sup, sig)$ is defined as $R'$ above;
if $sup(s)>sup(t_{i-1})$, then $sup(q_i)=b$, otherwise $sup(q_i)=0$;
if $sup(q_i)=b$, then $sig(w_i)=(0, b-sup(t_{i-1}))$; otherwise $sig(w_i)=(sup(t_{i-1}), 0)$;
finally, if $sup(q_i)=b$, then $sig(y_i)=(b-sup(\iota_i), 0)$; otherwise $sig(y_i)=(0,sup(\iota_i))$.

(2):
\emph{If}: Again, the \emph{if}-direction is trivial.

\emph{Only-if}:
Let $R=(sup, sig)$ be a $\tau$-region of $U$, which solves an ESSA $(a,z)$ or an SSA $(z,z')$ of $U$.
We can extended $R$ to a $\tau$-region $R'=(sup', sig')$ of $J(U)$ that also solves these atoms, by defining $R'$ for all $s\in S(U)\cup Q$ and all $e\in E(U)\cup W\cup Y$ as follows:
\begin{align*}
sup'(s) &= \begin{cases}
sup(s), & \text{if } s \in S(U),\\
sup(z), & \text{if } s \in Q
\end{cases}\\
sig'(e) &= \begin{cases}
sig(e), & \text{if } e \in E(U),\\
0, & \text{if } e \in W ,\\
(sup(z) - sup(\iota_i),0) & \text{if } e = y_i \text{ and } sup(\iota_i)  < sup(z) \text{ and } i\in \{0,\dots, n\} \\
( 0, sup(\iota_i ) - sup(z)) & \text{if } e = y_i \text{ and } sup(\iota_i)  \geq  sup(z) \text{ and } i\in \{0,\dots, n\}
\end{cases}
\end{align*}

Notice that $R'$ also solves $(a,q_i)$ for all $i\in \{0,\dots, n\}$ as $sup(q_i)=sup(z)$.
Consequently, since there is at least one state $s\in  S(U)$ for every event $e\in E(U)$ such that $(e,s)$ is an ESSA of $U$, the atom $(e,q_i)$ is solvable for every $e\in E(U)$ and every $i\in \{0,\dots, n\}$.
As a result, to prove the $\tau$-(E)SSP for $J(U)$ it remains to argue that the remaining SSA and ESSA at which the states of $Q$ and the events of $W\cup Y$ participate are solvable in $J(U)$.
If $i\in \{0,\dots, n\}$ and $s\in S(J(U))$ and $e\in E(J(U))$ then the following region $(sup, sig)$ simultaneously solves every valid atom $(y_i, \cdot)$, $(q_i,\cdot) $ and $(w_{i+1}, \cdot)$ in $J(U)$ (if the latter exists):
\[
sup(s) =\begin{cases}
	0, & \text{if } s=q_i\\
	b, & \text{otherwise }
	\end{cases}\\
\text{  \text{ } }
sig(e)=\begin{cases}
	(0,b), & \text{if } e = y_i \text{ or } ( i < n \text{ and } e=w_{i+1})\\
	(b,0), & \text{if } 1 \leq i  \text{ and } e=w_{i-1}\\
	0, & \text{ otherwise} \\
	\end{cases}
\]

\vspace*{-7mm}
\end{proof}

\section{NP-completeness results}\label{sec:hardness_results}%

The following theorem is the main contribution of this section:
\begin{theorem}\label{the:hardness_results}
\begin{enumerate}
\item\label{the:hardness_results_1}
If $\tau \in \{\tau_{PT}^b, \tau_{PPT}^b\}$, then $\tau$-\textsc{Solvability} and $\tau$-\textsc{ESSP} and $\tau$-\textsc{SSP} are NP-complete, even when restricted to linear TS.
\item\label{the:hardness_results_2}
Let $\tau \in \{\tau_{\mathbb{Z}PT}^b, \tau_{\mathbb{Z}PPT}^b\}$.
For any fixed $g\geq 2$, $\tau$-\textsc{Solvability} and $\tau$-\textsc{ESSP} are NP-complete, even when restricted to $g$-grade TS.
\end{enumerate}
\end{theorem}
For the proof of Theorem~\ref{the:hardness_results}, on the one hand, we have to argue that $\tau$-\textsc{Solvability}, $\tau$-\textsc{ESSP} and $\tau$-\textsc{SSP} are in NP.
This can be seen as follows.
By Definition~\ref{def:state_separation} and Definition~\ref{def:event_state_separation}, a TS $A=(S,E,\delta,\iota)$ has at most $\vert S\vert^2$ SSA and at most $\vert S\vert\cdot \vert E\vert $ ESSA, respectively.
This implies that if a TS $A$ is $\tau$-solvable or has the $\tau$-SSP or the $\tau$-ESSP, then there is a set of $\tau$-regions $\mathcal{R}$ of $A$ of size at most $\vert S\vert^2 + \vert S\vert \cdot \vert E\vert $ that witnesses the corresponding property of $A$.
Consequently, there is a non-deterministic Turing machine that guesses $\mathcal{R}$ in a non-deterministic computation and verifies the validity of $\mathcal{R}$ in (deterministic) polynomial time.

On the other hand, we have to argue that the decision problems are NP-hard for accordingly restricted input TS.
The NP-hardness proofs base on polynomial-time reductions of the following decision problem, which has been shown to be NP-complete in~\cite{DBLP:journals/dcg/MooreR01}:

\smallskip\noindent
\fbox{\begin{minipage}[t][1.9\height][c]{0.97\textwidth}
\begin{decisionproblem}
  \problemtitle{\textsc{Cubic Monotone One-In-Three 3-SAT} (\textsc{CM1in33Sat})}
  \probleminput{A boolean expression $\varphi=\{\zeta_0,\dots, \zeta_{m-1}\}$ of 3-clauses such that, for all $i\in \{0,\dots, m-1\}$, the clause $\zeta_i=\{X_{i_0}, X_{i_1}, X_{i_2}\}$ contains $3$ distinct non-negated variables, where $i_0 < i_1 < i_2$;
every variable $X\in V(\varphi)$ occurs in exactly three distinct clauses, where $V(\varphi)=\bigcup_{i=0}^{m-1}\zeta_i$ denotes the set of all variables of $\varphi$.
}
  \problemquestion{Does there exist a one-in-three model of $\varphi$, that is, a subset $M\subseteq V(\varphi)$ such that $\vert M\cap\zeta_i\vert =1 $ for all $i\in \{0,\dots, m-1\}$?}
\end{decisionproblem}
\end{minipage}}
\bigskip

Notice that the characterization of the input $\varphi$ implies $\vert V(\varphi)\vert=m$.
The following example provides --up to renaming-- the smallest instance of \textsc{CM1in33Sat} that allows a positive decision:

\begin{example}\label{ex:varphi}
The boolean expression $\varphi=\{\zeta_0,\dots, \zeta_{5}\}$ with clauses $\zeta_0=\{X_0,X_1,X_2\},\ \zeta_1= \{X_0,X_2,X_3\},\ \zeta_2= \{X_0,X_1,X_3\},\ \zeta_3= \{X_2,X_4,X_5\},\ \zeta_4=\{X_1,X_4,X_5\},\ \zeta_5= \{X_3,X_4,X_5\}$ is a well-defined input of  \textsc{CM1in33Sat} and has the one-in-three model $M=\{X_0,X_4\}$.
\end{example}

\textbf{General reduction approach.}
Let $\tau\in \{\tau_{PT}^b,\tau_{PPT}^b, \tau_{\mathbb{Z}PT}^b,\tau_{\mathbb{Z}PPT}^b\}$.
For the proof of the NP-hardness of \textsc{$\tau$-Solvability} and \textsc{$\tau$-ESSP} we reduce $\varphi$ to a union $U_\tau$ of gadget TS.
The index $\tau$ emphasizes that the actual peculiarity of the union depends on $\tau$.
In particular, if $\tau\in \{\tau_{PT}^b,\tau_{PPT}^b\}$, then all these TS are linear, and $JL(U_\tau)$ is a well defined linear TS.
Otherwise, if $\tau\in \{\tau_{\mathbb{Z}PT}^b,\tau_{\mathbb{Z}PPT}^b\}$, then these gadgets are 2-grade TS where no initial state has an incoming edge, which implies that $J(U_\tau)$ is a
2-grade TS.

\medskip
In $U_\tau$, the variables of $\varphi$ are represented by events and the clauses of $\varphi$ are represented by paths on which the variables of the clauses occur as events.
More exactly, for every $i\in \{0,\dots, m-1\}$ and clause $\zeta_i=\{X_{i_0}, X_{i_1}, X_{i_2}\}$, the union $U_\tau$ contains (a gadget with) a directed labeled path $P_i=\dots\Edge{X_{i_0}}\dots\Edge{X_{i_1}}\dots\Edge{X_{i_2}}\dots $ on which the variables $X_{i_0}, X_{i_1}$ and $X_{i_2}$ of $\zeta_i$ occur as events.
Moreover, by construction, the union $U_\tau$ provides an ESSA $\alpha$ whose $\tau$-solvability is connected with the existence of a one-in-three model of $\varphi$.
In particular, we build the union $U_\tau$ in a way such that there is a subset $\mathfrak{E}\subseteq E_\tau$ of events of $\tau$ so that the following properties are satisfied:
If $R=(sup, sig)$ is a $\tau$-region of $U_\tau$ that solves $\alpha$, then the variable events whose signature belongs to $\mathfrak{E}$ define a one-in-three model of $\varphi$, that is, the set $M=\{X\in V(\varphi)\mid sig(X)\in \mathfrak{E}\}$ satisfies $\vert M\cap \zeta_i\vert =1$ for all $i\in \{0,\dots, m-1\}$.
Hence, if $U_\tau$ has the $\tau$-ESSP, then $\alpha$ is $\tau$-solvable and $\varphi$ allows a positive decision.
Moreover, the construction of $U_\tau$ ensures that if $\varphi$ has a one-in-three model, then $\alpha$ as well as all the other ESSA and SSA of $U_\tau$ are $\tau$-solvable.
Thus, $U_\tau$ has the $\tau$-ESSP if and only if $\varphi$ is one-in-three satisfiable if and only if $U_\tau$ has both the $\tau$-ESSP and the $\tau$-SSP.
Since Lemma~\ref{lem:joining} lifts these implications to the linear joining $LJ(U_\tau)$, if $\tau\in \{\tau_{PT}^b,\tau_{PPT}^b\}$, and to the joining $J(U_\tau)$, if  $\tau\in \{\tau_{\mathbb{Z}PT}^b,\tau_{\mathbb{Z}PPT}^b\}$, this proves the NP-hardness of the \textsc{$\tau$-ESSP} and \textsc{$\tau$-Solvability} for accordingly restricted TS.

\eject
Let $\tau\in \{\tau_{PT}^b,\tau_{PPT}^b\}$.
For the proof of the NP-hardness of \textsc{$\tau$-SSP} we reduce $\varphi$ to a union $U$ of linear TS.
Since this union is the same for both $\tau_{PT}^b$ and $\tau_{PPT}^b$, $U$ needs no index.
Using essentially the same approach as just sketched, the union $U$ provides an SSA $\alpha$ that is $\tau$-solvable if and only if $\varphi$ has a one-in-three model.
Moreover, if $\alpha $ is $\tau$-solvable, then $U$ has the $\tau$-SSP.
Consequently, again by Lemma~\ref{lem:joining}, this implies that $LJ(U)$ has the $\tau$-SSP if and only if $\varphi$ has a one-in-three model.
This proves the NP-hardness of \textsc{$\tau$-SSP} for linear inputs.

\subsection{NP-hardness of \textsc{$\tau_{PPT}^b$-solvability} and \textsc{$\tau_{PPT}^b$-ESSP}}\label{sec:tau_ppt_solvability}%

\begin{figure}[t!]
\begin{center}
\begin{tikzpicture}
\begin{scope}
\foreach \i in {0,...,11} {\coordinate (\i) at (\i*1.325cm,0);}
\foreach \i in {0,...,11} {\node (h\i) at (\i) {\nscale{$h_{1,\i}$}};}
\foreach \i in {0,...,11} {\coordinate (s\i) at (\i*1.325cm,-0.5);}
\foreach \i in {0,4,9} {\node[opacity=0.7] (s\i) at (s\i) {\nscale{$[0]$}};}
\foreach \i in {1,5,10} {\node[opacity=0.7] (s\i) at (s\i) {\nscale{$[1]$}};}
\foreach \i in {2,3,6,7,8,11} {\node[opacity=0.7] (s\i) at (s\i) {\nscale{$[2]$}};}
\foreach \i in {0,...,10} {\coordinate (g\i) at (\i*1.325cm+0.6125cm ,-0.5cm);}
\foreach \i in {0,1,4,5,9,10} {\node[opacity=0.7](g\i) at (g\i) {\nscale{$(0,1)$}};}
\foreach \i in {2,6,7} {\node[opacity=0.7](g\i) at (g\i) {\nscale{$(0,0)$}};}
\foreach \i in {3,8} {\node[opacity=0.7](g\i) at (g\i) {\nscale{$(2,0)$}};}
\graph { (h0) ->["\escale{$k$}"] (h1) ->["\escale{$k$}"] (h2) ->["\escale{$y_0$}"] (h3) ->["\escale{$o_0$}"] (h4) ->["\escale{$k$}"] (h5)->["\escale{$k$}"](h6)->["\escale{$y_1$}"](h7)->["\escale{$y_0$}"](h8)->["\escale{$o_1$}"](h9)->["\escale{$k$}"](h10)->["\escale{$k$}"](h11);
};
\end{scope}
\begin{scope}[yshift=-2cm]
\begin{scope}%
\foreach \i in {0,...,3} {\coordinate (\i) at (\i*1.4cm,0);}
\foreach \i in {0,...,3} {\node (h\i) at (\i) {\nscale{$d_{j,\i}$}};}
\foreach \i in {0,...,3} {\coordinate (s\i) at (\i*1.4cm,-0.5);}
\foreach \i in {0,2} {\node[opacity=0.7] (s\i) at (s\i) {\nscale{$[2]$}};}
\foreach \i in {1,3} {\node[opacity=0.7] (s\i) at (s\i) {\nscale{$[0]$}};}
\foreach \i in {0,...,2} {\coordinate (g\i) at (\i*1.4cm+0.7cm ,-0.5cm);}
\foreach \i in {0,2} {\node[opacity=0.7](g\i) at (g\i) {\nscale{$(2,0)$}};}
\foreach \i in {1} {\node[opacity=0.7](g\i) at (g\i) {\nscale{$(0,2)$}};}
\graph { (h0) ->["\escale{$o_0$}"] (h1) ->["\escale{$k_j$}"] (h2) ->["\escale{$o_1$}"] (h3);
};
\end{scope}
\begin{scope}[xshift=5.75cm ]%
\foreach \i in {0,...,2} {\coordinate (\i) at (\i*1.4cm,0);}
\foreach \i in {0,...,2} {\node (h\i) at (\i) {\nscale{$f_{\ell,\i}$}};}
\foreach \i in {0,...,2} {\coordinate (s\i) at (\i*1.4cm,-0.5);}
\foreach \i in {1,2} {\node[opacity=0.7] (s\i) at (s\i) {\nscale{$[2]$}};}
\foreach \i in {0} {\node[opacity=0.7] (s\i) at (s\i) {\nscale{$[0]$}};}
\foreach \i in {0,...,2} {\coordinate (g\i) at (\i*1.4cm+0.7cm ,-0.5cm);}
\foreach \i in {0} {\node[opacity=0.7](g\i) at (g\i) {\nscale{$(0,2)$}};}
\foreach \i in {1} {\node[opacity=0.7](g\i) at (g\i) {\nscale{$(0,0)$}};}
\graph { (h0) ->["\escale{$k_0$}"] (h1) ->["\escale{$z_\ell$}"] (h2);
};
\end{scope}
\begin{scope}[xshift=10cm ]%
\foreach \i in {0,...,2} {\coordinate (\i) at (\i*1.4cm,0);}
\foreach \i in {0,...,2} {\node (h\i) at (\i) {\nscale{$g_{\ell,\i}$}};}
\foreach \i in {0,...,2} {\coordinate (s\i) at (\i*1.4cm,-0.5);}
\foreach \i in {0,1} {\node[opacity=0.7] (s\i) at (s\i) {\nscale{$[2]$}};}
\foreach \i in {2} {\node[opacity=0.7] (s\i) at (s\i) {\nscale{$[0]$}};}
\foreach \i in {0,...,2} {\coordinate (g\i) at (\i*1.4cm+0.7cm ,-0.5cm);}
\foreach \i in {1} {\node[opacity=0.7](g\i) at (g\i) {\nscale{$(2,0)$}};}
\foreach \i in {0} {\node[opacity=0.7](g\i) at (g\i) {\nscale{$(0,0)$}};}
\graph { (h0) ->["\escale{$z_\ell$}"] (h1) ->["\escale{$o_0$}"] (h2);
};
\end{scope}
\end{scope}

\begin{scope}[yshift=-4cm]
\begin{scope}
\foreach \i in {0,...,3} {\coordinate (\i) at (\i*1.4cm,0);}
\foreach \i in {0,...,3} {\node (h\i) at (\i) {\nscale{$m_{0,\i}$}};}
\foreach \i in {0,...,3} {\coordinate (s\i) at (\i*1.4cm,-0.5);}
\foreach \i in {0,3} {\node[opacity=0.7] (s\i) at (s\i) {\nscale{$[0]$}};}
\foreach \i in {2} {\node[opacity=0.7] (s\i) at (s\i) {\nscale{$[1]$}};}
\foreach \i in {1} {\node[opacity=0.7] (s\i) at (s\i) {\nscale{$[2]$}};}
\foreach \i in {0,...,2} {\coordinate (g\i) at (\i*1.4cm+0.7cm ,-0.5cm);}
\foreach \i in {0} {\node[opacity=0.7](g\i) at (g\i) {\nscale{$(0,2)$}};}
\foreach \i in {1,2} {\node[opacity=0.7](g\i) at (g\i) {\nscale{$(1,0)$}};}
\graph { (h0) ->["\escale{$k_1$}"] (h1) ->["\escale{$X_0$}"] (h2) ->["\escale{$X_0$}"] (h3);
};
\end{scope}
\begin{scope}[xshift=5.25cm]
\foreach \i in {0,...,3} {\coordinate (\i) at (\i*1.4cm,0);}
\foreach \i in {0,...,3} {\node (h\i) at (\i) {\nscale{$m_{1,\i}$}};}
\foreach \i in {0,...,3} {\coordinate (s\i) at (\i*1.4cm,-0.5);}
\foreach \i in {0} {\node[opacity=0.7] (s\i) at (s\i) {\nscale{$[0]$}};}
\foreach \i in {1,2,3} {\node[opacity=0.7] (s\i) at (s\i) {\nscale{$[2]$}};}
\foreach \i in {0,...,2} {\coordinate (g\i) at (\i*1.4cm+0.7cm ,-0.5cm);}
\foreach \i in {0} {\node[opacity=0.7](g\i) at (g\i) {\nscale{$(0,2)$}};}
\foreach \i in {1,2} {\node[opacity=0.7](g\i) at (g\i) {\nscale{$(0,0)$}};}
\graph { (h0) ->["\escale{$k_1$}"] (h1) ->["\escale{$X_1$}"] (h2) ->["\escale{$X_1$}"] (h3);
};
\end{scope}
\begin{scope}[xshift=10.5cm]
\foreach \i in {0,...,3} {\coordinate (\i) at (\i*1.4cm,0);}
\foreach \i in {0,...,3} {\node (h\i) at (\i) {\nscale{$m_{2,\i}$}};}
\foreach \i in {0,...,3} {\coordinate (s\i) at (\i*1.4cm,-0.5);}
\foreach \i in {0} {\node[opacity=0.7] (s\i) at (s\i) {\nscale{$[0]$}};}
\foreach \i in {1,2,3} {\node[opacity=0.7] (s\i) at (s\i) {\nscale{$[2]$}};}
\foreach \i in {0,...,2} {\coordinate (g\i) at (\i*1.4cm+0.7cm ,-0.5cm);}
\foreach \i in {0} {\node[opacity=0.7](g\i) at (g\i) {\nscale{$(0,2)$}};}
\foreach \i in {1,2} {\node[opacity=0.7](g\i) at (g\i) {\nscale{$(0,0)$}};}
\graph { (h0) ->["\escale{$k_1$}"] (h1) ->["\escale{$X_2$}"] (h2) ->["\escale{$X_2$}"] (h3);
};
\end{scope}
\begin{scope}[yshift=-1.5cm]
\begin{scope}
\foreach \i in {0,...,3} {\coordinate (\i) at (\i*1.4cm,0);}
\foreach \i in {0,...,3} {\node (h\i) at (\i) {\nscale{$m_{3,\i}$}};}
\foreach \i in {0,...,3} {\coordinate (s\i) at (\i*1.4cm,-0.5);}
\foreach \i in {0} {\node[opacity=0.7] (s\i) at (s\i) {\nscale{$[0]$}};}
\foreach \i in {1,2,3} {\node[opacity=0.7] (s\i) at (s\i) {\nscale{$[2]$}};}
\foreach \i in {0,...,2} {\coordinate (g\i) at (\i*1.4cm+0.7cm ,-0.5cm);}
\foreach \i in {0} {\node[opacity=0.7](g\i) at (g\i) {\nscale{$(0,2)$}};}
\foreach \i in {1,2} {\node[opacity=0.7](g\i) at (g\i) {\nscale{$(0,0)$}};}
\graph { (h0) ->["\escale{$k_1$}"] (h1) ->["\escale{$X_3$}"] (h2) ->["\escale{$X_3$}"] (h3);
};
\end{scope}
\begin{scope}[xshift=5.25cm]
\foreach \i in {0,...,3} {\coordinate (\i) at (\i*1.4cm,0);}
\foreach \i in {0,...,3} {\node (h\i) at (\i) {\nscale{$m_{4,\i}$}};}
\foreach \i in {0,...,3} {\coordinate (s\i) at (\i*1.4cm,-0.5);}
\foreach \i in {0,3} {\node[opacity=0.7] (s\i) at (s\i) {\nscale{$[0]$}};}
\foreach \i in {1} {\node[opacity=0.7] (s\i) at (s\i) {\nscale{$[2]$}};}
\foreach \i in {2} {\node[opacity=0.7] (s\i) at (s\i) {\nscale{$[1]$}};}
\foreach \i in {0,...,2} {\coordinate (g\i) at (\i*1.4cm+0.7cm ,-0.5cm);}
\foreach \i in {0} {\node[opacity=0.7](g\i) at (g\i) {\nscale{$(0,2)$}};}
\foreach \i in {1,2} {\node[opacity=0.7](g\i) at (g\i) {\nscale{$(1,0)$}};}
\graph { (h0) ->["\escale{$k_1$}"] (h1) ->["\escale{$X_4$}"] (h2) ->["\escale{$X_4$}"] (h3);
};
\end{scope}
\begin{scope}[xshift=10.5cm]
\foreach \i in {0,...,3} {\coordinate (\i) at (\i*1.4cm,0);}
\foreach \i in {0,...,3} {\node (h\i) at (\i) {\nscale{$m_{5,\i}$}};}
\foreach \i in {0,...,3} {\coordinate (s\i) at (\i*1.4cm,-0.5);}
\foreach \i in {0} {\node[opacity=0.7] (s\i) at (s\i) {\nscale{$[0]$}};}
\foreach \i in {1,2,3} {\node[opacity=0.7] (s\i) at (s\i) {\nscale{$[2]$}};}
\foreach \i in {0,...,2} {\coordinate (g\i) at (\i*1.4cm+0.7cm ,-0.5cm);}
\foreach \i in {0} {\node[opacity=0.7](g\i) at (g\i) {\nscale{$(0,2)$}};}
\foreach \i in {1,2} {\node[opacity=0.7](g\i) at (g\i) {\nscale{$(0,0)$}};}
\graph { (h0) ->["\escale{$k_1$}"] (h1) ->["\escale{$X_5$}"] (h2) ->["\escale{$X_5$}"] (h3);
};
\end{scope}
\end{scope}
\end{scope}
\begin{scope}[yshift=-7.5cm]
\begin{scope}
\foreach \i in {0,...,10} {\coordinate (\i) at (\i*1.4cm,0);}
\foreach \i in {0,...,10} {\node (h\i) at (\i) {\nscale{$t_{0,\i}$}};}
\foreach \i in {0,...,10} {\coordinate (s\i) at (\i*1.4cm,-0.5);}
\foreach \i in {0,3,4,5,6,7,8,9} {\node[opacity=0.7] (s\i) at (s\i) {\nscale{$[0]$}};}
\foreach \i in {2} {\node[opacity=0.7] (s\i) at (s\i) {\nscale{$[1]$}};}
\foreach \i in {1,10} {\node[opacity=0.7] (s\i) at (s\i) {\nscale{$[2]$}};}
\foreach \i in {0,...,10} {\coordinate (g\i) at (\i*1.4cm+0.7cm ,-0.5cm);}
\foreach \i in {3,4,5,6,7,8} {\node[opacity=0.7](g\i) at (g\i) {\nscale{$(0,0)$}};}
\foreach \i in {1,2} {\node[opacity=0.7](g\i) at (g\i) {\nscale{$(1,0)$}};}
\foreach \i in {0,9} {\node[opacity=0.7](g\i) at (g\i) {\nscale{$(0,2)$}};}
\graph { (h0) ->["\escale{$k_2$}"] (h1) ->["\escale{$X_0$}"] (h2) ->["\escale{$X_0$}"] (h3) ->["\escale{$z_0$}"] (h4) ->["\escale{$X_1$}"] (h5)->["\escale{$X_1$}"](h6)->["\escale{$z_1$}"](h7)->["\escale{$X_2$}"](h8)->["\escale{$X_2$}"](h9)->["\escale{$k_3$}"](h10);
};
\end{scope}
\begin{scope}[yshift=-1.5cm]
\foreach \i in {0,...,10} {\coordinate (\i) at (\i*1.4cm,0);}
\foreach \i in {0,...,10} {\node (h\i) at (\i) {\nscale{$t_{1,\i}$}};}
\foreach \i in {0,...,10} {\coordinate (s\i) at (\i*1.4cm,-0.5);}
\foreach \i in {0,3,4,5,6,7,8,9} {\node[opacity=0.7] (s\i) at (s\i) {\nscale{$[0]$}};}
\foreach \i in {2} {\node[opacity=0.7] (s\i) at (s\i) {\nscale{$[1]$}};}
\foreach \i in {1,10} {\node[opacity=0.7] (s\i) at (s\i) {\nscale{$[2]$}};}
\foreach \i in {0,...,10} {\coordinate (g\i) at (\i*1.4cm+0.7cm ,-0.5cm);}
\foreach \i in {3,4,5,6,7,8} {\node[opacity=0.7](g\i) at (g\i) {\nscale{$(0,0)$}};}
\foreach \i in {1,2} {\node[opacity=0.7](g\i) at (g\i) {\nscale{$(1,0)$}};}
\foreach \i in {0,9} {\node[opacity=0.7](g\i) at (g\i) {\nscale{$(0,2)$}};}
\graph { (h0) ->["\escale{$k_2$}"] (h1) ->["\escale{$X_0$}"] (h2) ->["\escale{$X_0$}"] (h3) ->["\escale{$z_2$}"] (h4) ->["\escale{$X_2$}"] (h5)->["\escale{$X_2$}"](h6)->["\escale{$z_3$}"](h7)->["\escale{$X_3$}"](h8)->["\escale{$X_3$}"](h9)->["\escale{$k_3$}"](h10);
};
\end{scope}
\begin{scope}[yshift=-3cm]
\foreach \i in {0,...,10} {\coordinate (\i) at (\i*1.4cm,0);}
\foreach \i in {0,...,10} {\node (h\i) at (\i) {\nscale{$t_{2,\i}$}};}
\foreach \i in {0,...,10} {\coordinate (s\i) at (\i*1.4cm,-0.5);}
\foreach \i in {0,3,4,5,6,7,8,9} {\node[opacity=0.7] (s\i) at (s\i) {\nscale{$[0]$}};}
\foreach \i in {2} {\node[opacity=0.7] (s\i) at (s\i) {\nscale{$[1]$}};}
\foreach \i in {1,10} {\node[opacity=0.7] (s\i) at (s\i) {\nscale{$[2]$}};}
\foreach \i in {0,...,10} {\coordinate (g\i) at (\i*1.4cm+0.7cm ,-0.5cm);}
\foreach \i in {3,4,5,6,7,8} {\node[opacity=0.7](g\i) at (g\i) {\nscale{$(0,0)$}};}
\foreach \i in {1,2} {\node[opacity=0.7](g\i) at (g\i) {\nscale{$(1,0)$}};}
\foreach \i in {0,9} {\node[opacity=0.7](g\i) at (g\i) {\nscale{$(0,2)$}};}
\graph { (h0) ->["\escale{$k_2$}"] (h1) ->["\escale{$X_0$}"] (h2) ->["\escale{$X_0$}"] (h3) ->["\escale{$z_4$}"] (h4) ->["\escale{$X_1$}"] (h5)->["\escale{$X_1$}"](h6)->["\escale{$z_5$}"](h7)->["\escale{$X_3$}"](h8)->["\escale{$X_3$}"](h9)->["\escale{$k_3$}"](h10);
};
\end{scope}
\begin{scope}[yshift=-4.5cm]
\foreach \i in {0,...,10} {\coordinate (\i) at (\i*1.4cm,0);}
\foreach \i in {0,...,10} {\node (h\i) at (\i) {\nscale{$t_{3,\i}$}};}
\foreach \i in {0,...,10} {\coordinate (s\i) at (\i*1.4cm,-0.5);}
\foreach \i in {0,6,7,8,9} {\node[opacity=0.7] (s\i) at (s\i) {\nscale{$[0]$}};}
\foreach \i in {5} {\node[opacity=0.7] (s\i) at (s\i) {\nscale{$[1]$}};}
\foreach \i in {1,2,3,4,10} {\node[opacity=0.7] (s\i) at (s\i) {\nscale{$[2]$}};}
\foreach \i in {0,...,10} {\coordinate (g\i) at (\i*1.4cm+0.7cm ,-0.5cm);}
\foreach \i in {1,2,3,6,7,8} {\node[opacity=0.7](g\i) at (g\i) {\nscale{$(0,0)$}};}
\foreach \i in {4,5} {\node[opacity=0.7](g\i) at (g\i) {\nscale{$(1,0)$}};}
\foreach \i in {0,9} {\node[opacity=0.7](g\i) at (g\i) {\nscale{$(0,2)$}};}
\graph { (h0) ->["\escale{$k_2$}"] (h1) ->["\escale{$X_2$}"] (h2) ->["\escale{$X_2$}"] (h3) ->["\escale{$z_6$}"] (h4) ->["\escale{$X_4$}"] (h5)->["\escale{$X_4$}"](h6)->["\escale{$z_7$}"](h7)->["\escale{$X_5$}"](h8)->["\escale{$X_5$}"](h9)->["\escale{$k_3$}"](h10);
};
\end{scope}
\begin{scope}[yshift=-6cm]
\foreach \i in {0,...,10} {\coordinate (\i) at (\i*1.4cm,0);}
\foreach \i in {0,...,10} {\node (h\i) at (\i) {\nscale{$t_{4,\i}$}};}
\foreach \i in {0,...,10} {\coordinate (s\i) at (\i*1.4cm,-0.5);}
\foreach \i in {0,6,7,8,9} {\node[opacity=0.7] (s\i) at (s\i) {\nscale{$[0]$}};}
\foreach \i in {5} {\node[opacity=0.7] (s\i) at (s\i) {\nscale{$[1]$}};}
\foreach \i in {1,2,3,4,10} {\node[opacity=0.7] (s\i) at (s\i) {\nscale{$[2]$}};}
\foreach \i in {0,...,10} {\coordinate (g\i) at (\i*1.4cm+0.7cm ,-0.5cm);}
\foreach \i in {1,2,3,6,7,8} {\node[opacity=0.7](g\i) at (g\i) {\nscale{$(0,0)$}};}
\foreach \i in {4,5} {\node[opacity=0.7](g\i) at (g\i) {\nscale{$(1,0)$}};}
\foreach \i in {0,9} {\node[opacity=0.7](g\i) at (g\i) {\nscale{$(0,2)$}};}
\graph { (h0) ->["\escale{$k_2$}"] (h1) ->["\escale{$X_1$}"] (h2) ->["\escale{$X_1$}"] (h3) ->["\escale{$z_8$}"] (h4) ->["\escale{$X_4$}"] (h5)->["\escale{$X_4$}"](h6)->["\escale{$z_9$}"](h7)->["\escale{$X_5$}"](h8)->["\escale{$X_5$}"](h9)->["\escale{$k_3$}"](h10);
};
\end{scope}
\begin{scope}[yshift=-7.5cm]
\foreach \i in {0,...,10} {\coordinate (\i) at (\i*1.4cm,0);}
\foreach \i in {0,...,10} {\node (h\i) at (\i) {\nscale{$t_{5,\i}$}};}
\foreach \i in {0,...,10} {\coordinate (s\i) at (\i*1.4cm,-0.5);}
\foreach \i in {0,6,7,8,9} {\node[opacity=0.7] (s\i) at (s\i) {\nscale{$[0]$}};}
\foreach \i in {5} {\node[opacity=0.7] (s\i) at (s\i) {\nscale{$[1]$}};}
\foreach \i in {1,2,3,4,10} {\node[opacity=0.7] (s\i) at (s\i) {\nscale{$[2]$}};}
\foreach \i in {0,...,10} {\coordinate (g\i) at (\i*1.4cm+0.7cm ,-0.5cm);}
\foreach \i in {1,2,3,6,7,8} {\node[opacity=0.7](g\i) at (g\i) {\nscale{$(0,0)$}};}
\foreach \i in {4,5} {\node[opacity=0.7](g\i) at (g\i) {\nscale{$(1,0)$}};}
\foreach \i in {0,9} {\node[opacity=0.7](g\i) at (g\i) {\nscale{$(0,2)$}};}
\graph { (h0) ->["\escale{$k_2$}"] (h1) ->["\escale{$X_3$}"] (h2) ->["\escale{$X_3$}"] (h3) ->["\escale{$z_{10}$}"] (h4) ->["\escale{$X_4$}"] (h5)->["\escale{$X_4$}"](h6)->["\escale{$z_{11}$}"](h7)->["\escale{$X_5$}"](h8)->["\escale{$X_5$}"](h9)->["\escale{$k_3$}"](h10);
};
\end{scope}
\end{scope}
\end{tikzpicture}
\end{center}\vspace*{-4mm}
\caption{The gadgets of the union $U_{\tau_{PPT}^2}$ that originates from the input $\varphi$ of Example~\ref{ex:varphi}; we assume $j\in \{0,1,2,3\}$ and $\ell\in \{0,\dots, 11\}$.
A number $[i]$ and and a pair $(k,\ell)$ below a state $s$ and event $e$ define the support $sup(s)=i$ and the signature $sig(e)=(k,\ell)$ in correspondence to the region $R=(sup, sig)$, which is defined to prove the $\tau_{PPT}^2$-solvability of $\alpha=(k,h_{1,8})$.}\label{fig:example_for_pure_essp}
\end{figure}
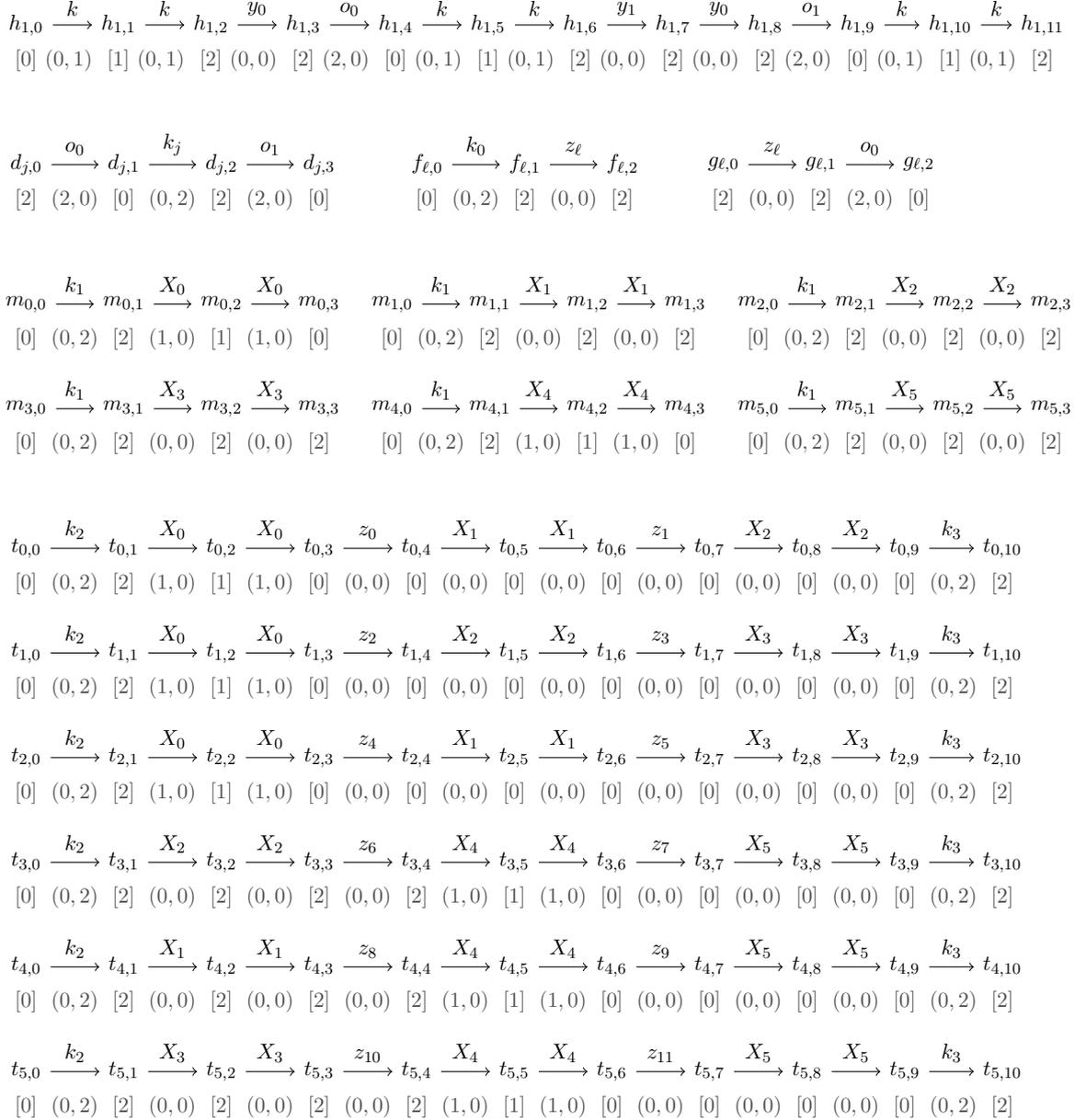\vspace*{-2mm}

In the remainder of this section, unless explicitly stated otherwise, let $\tau=\tau_{PPT}^b$.
In the following, we first introduce the gadgets (TS) of the union $U_\tau$ and the atom $\alpha$.
Figure~\ref{fig:example_for_pure_essp} presents a concrete example of $U_{\tau_{PPT}^2}$, where $\varphi$ corresponds to Example~\ref{ex:varphi}.
Secondly, we argue that these gadgets collaborate in a way such that if $\alpha$ is $\tau$-solvable, then $\varphi$ has a one-in-three model.
Finally, we show that if $\varphi$ is one-in-three satisfiable, then $U^\tau_\varphi$ is $\tau$-solvable.

\medskip
The union $U_\tau$ has the following gadget $H_1$ that provides the announced ESSA $\alpha=(k, h_{1,2b+4})$:
\begin{center}
\begin{tikzpicture}
\node (h0) at (0,0) {\nscale{$h_{1,0}$}};
\node (h1) at (1.4,0) {\nscale{}};
\node (h_2_dots) at (1.6,0) {\nscale{$\dots$}};
\node (h_k_1) at (1.8,0) {};
\node (h_k) at (3.2,0) {\nscale{$h_{1,b}$}};
\node (h_k1) at (4.9,0) {\nscale{$h_{1,b+1}$}};
\node (h_k2) at (6.8,0) {\nscale{$h_{1,b+2}$}};
\node (h_k3) at (8.4,0) {\nscale{}};
\node (h_k4_dots) at (8.75,0) {\nscale{$\dots$}};
\node (h_2k1) at (8.9,0) {\nscale{}};
\node (h_2k2) at (10.5,0) {\nscale{$h_{1,2b+2}$}};
\node (h_2k3) at (12.5,0) {\nscale{$h_{1,2b+3}$}};
\node (h_2k4) at (14.5,0) {\nscale{$h_{1,2b+4}$}};
\node (h_2k5) at (14.5,-1.2) {\nscale{$h_{1,2b+5}$}};
\node (h_2k6) at (12.9,-1.2) {};
\node (h_k5_dots) at (12.7,-1.2) {\nscale{$\dots$}};
\node (h_3k4) at (12.4,-1.2) {};
\node (h_3k5) at (10.8,-1.2) {\nscale{$h_{1,3b+5}$}};
\graph { (h0) ->["\escale{$k$}"] (h1) ;
(h_k_1)->["\escale{$k$}"] (h_k) ->["\escale{$y_0$}"] (h_k1)->["\escale{$o_0$}"] (h_k2)->["\escale{$k$}"] (h_k3);
(h_2k1)->["\escale{$k$}"] (h_2k2)->["\escale{$y_1$}"] (h_2k3)->["\escale{$y_0$}"] (h_2k4)->["\escale{$o_1$}"] (h_2k5)->[swap, "\escale{$k$}"] (h_2k6);
(h_3k4)->[swap, "\escale{$k$}"] (h_3k5);
};
\end{tikzpicture}
\end{center}
%
%
For all $j\in \{0,1,2,3\}$, the union $U_\tau$ has the following gadget $D_j$ that provides the event $k_j$:
\begin{center}
\begin{tikzpicture}[baseline=-2pt]
\node at (-1,0) {$D_{j,1}=$};
\foreach \i in {0,...,3} {\coordinate (\i) at (\i*1.5,0);}
\foreach \i in {0,...,3} {\node (p\i) at (\i) {\nscale{$d_{j,\i}$}};}
\graph { (p0) ->["\escale{$o_0$}"] (p1) ->["\escale{$k_j$}"] (p2) ->["\escale{$o_1$}"] (p3);};
\end{tikzpicture}
\end{center}
For all $j\in \{0,\dots, 2m-1\}$, the union $U_\tau$ has the following gadgets $F_j$ and $G_j$ that provide the event~$z_j$:\vspace*{-2mm}
\begin{center}
\begin{tikzpicture}
\begin{scope}
\node at (-1,0) {$F_j=$};
\foreach \i in {0,...,2} {\coordinate (\i) at (\i*1.5,0);}
\foreach \i in {0,...,2} {\node (p\i) at (\i) {\nscale{$f_{j,\i}$}};}
\graph { (p0) ->["\escale{$k_0$}"] (p1) ->["\escale{$z_j$}"] (p2);};
\end{scope}
\begin{scope}[xshift=6cm]
\node at (-1,0) {$G_j=$};
\foreach \i in {0,...,2} {\coordinate (\i) at (\i*1.5,0);}
\foreach \i in {0,...,2} {\node (p\i) at (\i) {\nscale{$g_{j,\i}$}};}
\graph { (p0) ->["\escale{$z_j$}"] (p1) ->["\escale{$o_0$}"] (p2);};
\end{scope}
\end{tikzpicture}
\end{center}
For all $i\in \{0,\dots, m-1\}$, the union $U_\tau$ has the following gadget $M_i$, that uses the variable $X_i$ as event:\vspace*{-2mm}
\begin{center}
\begin{tikzpicture}[scale=0.9]
\begin{scope}[yshift=3cm]
\node at (-1,0) {$M_i=$};
\node (t0) at (0,0) {\nscale{$m_{i,0}$}};
\node (t1) at (2,0) {\nscale{$m_{i,1}$}};
\node (t2) at (3.75,0) {};
\node (h_2_dots) at (4,0) {\nscale{$\dots$}};
\node (tb) at (4.25,0) {};
\node (tb+1) at (6,0) {\nscale{$m_{i,b+1}$}};
\graph {
(t0) ->["\escale{$k_1$}"] (t1) ->["\escale{$X_i$}"] (t2) ;
(tb) ->["\escale{$X_i$}"] (tb+1);};
\end{scope}
\end{tikzpicture}
\end{center}
For all $i\in \{0,\dots, m-1\}$, the union $U_\tau$ has the following gadget $T_i$ that uses the elements of $\zeta_i=\{X_{i_0}, X_{i_1}, X_{i_2}\}$ as events:
\begin{center}
\begin{tikzpicture}[scale=0.9]
\begin{scope}[yshift=3cm]
\node at (-1,0) {$T_{i}=$};
\node (t0) at (0,0) {\nscale{$t_{i,0}$}};
\node (t1) at (1.75,0) {\nscale{$t_{i,1}$}};
\node (t2) at (3.25,0) {};
\node (h_2_dots) at (3.75,0) {\nscale{$\dots$}};
\node (tb) at (4.1,0) {};
\node (tb+1) at (5.75,0) {\nscale{$t_{i,b+1}$}};
\node (tb+2) at (7.75,0) {\nscale{$t_{i,b+2}$}};
\node (tb+3) at (9.5,0) {};
\node (h_k+4_dots) at (9.75,0) {\nscale{$\dots$}};
\node (t2b+1) at (10,0) {};
\node (t2b+2) at (12,0) {\nscale{$t_{i,2b+2}$}};
\node (t2b+3) at (12,-1.2) {\nscale{$t_{i,2b+3}$}};
\node (t2b+4) at (10.1,-1.2) {};
\node (h_k+5_dots) at (9.75,-1.2) {\nscale{$\dots$}};
\node (t3b+2) at (9.5,-1.2) {};
\node (t3b+3) at (7.75,-1.2) {\nscale{$t_{i,3b+3}$}};
\node (t3b+4) at (5.75,-1.2) {\nscale{$t_{i,3b+4}$}};
\graph {
(t0) ->["\escale{$k_2$}"] (t1) ->["\escale{$X_{i_0}$}"] (t2) ;
(tb) ->["\escale{$X_{i_0}$}"] (tb+1) ->["\escale{$z_{2i}$}"] (tb+2)->["\escale{$X_{i_1}$}"] (tb+3);
(t2b+1)->["\escale{$X_{i_1}$}"] (t2b+2)->["\escale{$z_{2i+1}$}"] (t2b+3)->[swap, "\escale{$X_{i_2}$}"] (t2b+4);
(t3b+2)->[swap, "\escale{$X_{i_2}$}"] (t3b+3)->[swap, "\escale{$k_3$}"] (t3b+4);
;};
\end{scope}
\end{tikzpicture}
\end{center}
Altogether, 
\[
U_\tau=U(H_1,D_0,\dots, D_3, F_0,\dots, F_{2m-1}, G_0,\dots, G_{2m-1},M_0,\dots, M_{m-1}, T_0,\dots, T_{m-1}).
\]

\begin{lemma}\label{lem:tau_ppt_essp_implies_model}
If $U_\tau$ has the $\tau$-ESSP, then $\varphi$ has a one-in-three model.
\end{lemma}
\begin{proof}
Since $U_\tau$ has the $\tau$-ESSP, there is a $\tau$-region that solves $\alpha$.
Let $R=(sup, sig)$ be such a region.
In the following we argue, that the set $\{X\in V(\varphi)\vert sig(X)=(0,1)\}$ or the set $\{X\in V(\varphi)\vert sig(X)=(1,0)\}$ is a one-in-three model of $\varphi$.
Since $R$ solves $\alpha$, we have that $sig(k)$ does not occur at $sup(h_{1,2b+4})$.
This implies $sig(k)\not=(0,0)$.
By Lemma~\ref{lem:observations}, we get $sig(k)\in \{(1,0),(0,1)\}$.
In what follows, we let $sig(k)=(0,1)$ and show that $M=\{X\in V(\varphi)\vert sig(X)=(1,0)\}$ defines a one-in-three model of $\varphi$.
The arguments for the case $sig(k)=(1,0)$ are quite similar and lead to the fact that $\{X\in V(\varphi)\vert sig(X)=(0,1)\}$ defines a searched model.

Let $sig(k)=(0,1)$ and $\neg sup(h_{1,2b+4}) \ledge{sig(k)}$.
We argue that this implies $sig(o_0)=sig(o_1)=(b,0)$:
For all $s\in \{0,\dots, b-1\}$, the event $(0,1)$ occurs at $s$ in $\tau$.
Since $sig(k)$ does not occur at $sup(h_{1,2b+4})$, this implies $sup(h_{1,2b+4})=b$.
Moreover, by $sig(k)=(0,1)$ and Lemma~\ref{lem:observations}, we get $sup(h_{1,b})=b$ and $sup(h_{1,b+2})=sup(h_{1,2b+5})=0$.
By $sup(h_{1,2b+4})=b$ and $sup(h_{1,2b+5})=0$, we obtain $sig(o_1)=(b,0)$.
Moreover, $sup(h_{1,b})=b$ and $h_{1,b}\edge{y_0}$ imply $sig^+(y_0)=0$, and by $sup(h_{1,2b+4})=b$ and $\edge{y_0}h_{1,2b+4}$ imply $sig^-(y_0)=0$.
(Recall that $R$ is pure.)
Hence, $sig(y_0)=(0,0)$, which implies $sup(h_{1,b+1})=b$.
Thus, by $sup(h_{1,b+1})=b$ and $sup(h_{1,b+2})=0$, we obtain $sig(o_0)=(b,0)$.

\medskip
The gadgets $D_0,\dots, D_3$ use the signatures of $o_0$ and $o_1$ to determine the signatures of $k_0,\dots, k_3$.
More exactly, $sig(o_0)=sig(o_1)=(b,0)$ implies $sup(d_{j,1})=0$ and $sup(d_{j,2})=b$ for all $j\in \{0,1,2,3\}$.
Consequently, this implies $sig(k_0)=\dots=sig(k_3)=(0,b)$.

Let $j\in \{0,\dots, 2m-1\}$ be arbitrary but fixed.
The gadgets $F_j$ and $G_j$ ensure that $sig(z_j)=(0,0)$:
By $sig(o_0)=(b,0)$ and $sig(k_0)=(0,b)$, we get $sup(f_{j,1})=b$ and $sup(g_{j,1})=b$.
Since $R$ is pure, that is $sig^+(z_j)=0$ or $sig^-(z_j)=0$, by $f_{j,2}\Edge{z_j}$, we get $sig^-(z_j)\geq sig^+(z_j)$.
Similarly, by $\Edge{z_j}g_{j,1}$, we get $sig^+(z_j)\geq sig^-(z_j)$.
Consequently, $sig^-(z_j)=sig^+(z_j)$, which implies $sig(z)=(0,0)$, since $R$ is pure.

\medskip
Let $i\in \{0,\dots, m-1\}$ be arbitrary but fixed.
The gadget $M_i$ ensures for $X_i$ that $sig(X_i)\in \{(1,0), (0,0)\}$:
By $sig(k_1)=(0,b)$, we have $sup(m_{i,1})=b$, which implies $sig^-(X_i)\geq sig^+(X_i)$.
Since $X_i$ occurs b times in a row at $m_{i,1}$, by Lemma~\ref{lem:observations}, this implies $sig(X_i)\in \{(1,0),(0,0)\}$.

The gadget $T_i$ ensures that there is exactly one event $X\in \{X_{i_0}, X_{i_1}, X_{i_2}\}$ such that $sig(X)=(1,0)$:
By $sig(k_2)=sig(k_3)=(0,b)$, we have that $sup(t_{i,1})=b$ and $sup(t_{i,3b+3})=0$.
Consequently, the image of the sub-path $t_{i,1}\edge{X_{i_1}}\dots\edge{X_{i_2}}t_{i,3b+3}$ under $(sup, sig)$ is a path of $\tau$ that starts at $b$ and terminates at $0$.
Hence, there is an event $e$ on this path that satisfies $sig^-(e) > sig^+(e)$.
Since $sig(z_{2i})= sig(z_{2i+1})=(0,0)$, we obtain that $e\in \{X_{i_0}, X_{i_1}, X_{i_2}\}$.
Moreover, since each of $X_{i_0}, X_{i_1}$ and $X_{i_2}$ occurs b times in a row, if $sig^-(e) > sig^+(e)$, then $sig(e)=(1,0)$.
In the following, we argue that if $e\in \{X_{i_0}, X_{i_1}, X_{i_2}\}$ such that $sig(e)=(1,0)$, then $sig(e')\not=(1,0)$ for all $e'\in \{X_{i_0}, X_{i_1}, X_{i_2}\}\setminus\{e\}$.

If $sig(X_{i_0})=(1,0)$, then we get $sup(t_{i,b+1})=0$, by Lemma~\ref{lem:observations}.
By $sig(z_{2i}) =(0,0)$, this implies $sup(t_{i,b+2})=0$ and $sig^-(X_{i_1})=0$.
Thus, by $sig(X_{i_1})\in \{(1,0),(0,0)\}$, we conclude $sig(X_{i_1})=(0,0)$.
By $sup(t_{i,b+2})=0$, $sig(X_{i_1})=(0,0)$ and $sig(z_{2i+1})=(0,0)$, we have that $sup(t_{i,2b+3})=0$.
This implies $sig^-(X_{i_2})=0$ and, thus, $sig(X_{i_2})=(0,0)$.
In particular, we have $sig(X_{i_1})\not=(1,0)$ and $sig(X_{i_2})\not=(1,0)$.

If $sig(X_{i_1})=(1,0)$, then we get $sup(t_{i,b+1})=b$ and $sup(t_{i,2b+3})=0$, by Lemma~\ref{lem:observations} and $sig(z_{2i}) =sig(z_{2i+1})=(0,0)$.
By $sup(t_{i,b+1})=b$, we get $sig(X_{i_0})\not=(1,0)$.
Moreover, just like before, by $sup(t_{i,2b+3})=0$, we have $sig(X_{i_2})\not=(1,0)$.

If $sig(X_{i_2})=(1,0)$, then we get $sig(X_{i_0})\not=(1,0)$ and $sig(X_{i_1})\not=(1,0)$, since $sig(X_{i_0})=(1,0)$ or $sig(X_{i_1})=(1,0)$ imply $sig(X_{i_2})\not=(1,0)$, as just discussed.

Altogether, we have shown that if $R=(sup, sig)$ is a $\tau$-region that solves $\alpha$ such that $sig(k)=(0,1)$, then, for all $i\in \{0,\dots, m-1\}$, there is exactly one event $e\in \{X_{i_0}, X_{i_1}, X_{i_2}\}$ that satisfies $sig(e)=(1,0)$.
As a result, the set $\{X\in V(\varphi)\mid sig(X)=(1,0)\}$ defines a one-in-three model of $\varphi$.
It is noteworthy that we use the pureness of $\tau$ only for the functionality of $H_1$ and (by the signature of $o_1$, implicitly) for $D_0,\dots, D_3$.
That is, once we have that $sig(k_0)=\dots=sig(k_3)=(0,b)$ and $sig(o_0)=(b,0)$, the arguments for the functionality of the remaining gadgets essentially work also for the (impure) $b$-bounded type $\tau_{PT}^b$.
The only difference then is that we can not conclude that $sig(z_j)=(0,0)$, since $sig(z_j)=(m,m)$ would also be possible for $\tau_{PT}^b$.
The same is true for $e'\in \{X_{i_0}, X_{i_1}, X_{i_2}\}\setminus\{e\}$ if $e\in \{X_{i_0}, X_{i_1}, X_{i_2}\}$ such that $sig(e)=(1,0)$.
However, if $sig(e)=(m,m)$, then $s\edge{e}s'$ implies also $sup(s)=sup(s')$, and that is what actually matters in our arguments.
Thus, we will reuse the corresponding gadgets for the type $\tau_{PT}^b$.

If $sig(k)=(1,0)$ and $sup(h_{1, 2b+4})=0$, then one argues similarly that the set $\{X\in V(\varphi)\mid sig(X)=(0,1)\}$ defines a one-in-three model of $\varphi$.
Altogether, this shows that if $U_\tau$ has the $\tau$-ESSP, which implies the $\tau$-solvability of $\alpha$, then $\varphi$ has a one-in-three model.
\end{proof}

For the opposite direction, we have to prove the following lemma:
\begin{lemma}\label{lem:tau_ppt_model_implies_solvability}
If $\varphi$ has a one-in-three model, then $U_\tau$ has the $\tau$-ESSP and the $\tau$-SSP.
\end{lemma}

For the proof of Lemma~\ref{lem:tau_ppt_model_implies_solvability} it is sufficient to show that if $\varphi$ has a one-in-three model $M$, then $U_\tau$ has the $\tau$-ESSP.
Since all introduced gadgets are linear TS, by Lemma~\ref{lem:essp_implies_ssp}, this implies that $U_\tau$ has the $\tau$-SSP, too.
The brut-force approach of this proof would be to explicitly present for every ESSA of $U_\tau$ a $\tau$-region that solves it.
In fact, for some atoms of $U_\tau$, we need to explicitly present regions that solve them.
In particular, this applies to $(k, h_{1, 2b+4})$.
On the other hand, the gadgets and the events of $U_\tau$ meet some regularities that allow us to solve many events homogeneously.
In the following, for the purpose to discover these regularities, we first introduce the notions of consistent and thinly distributed events.
After that we present a lemma that uses these notions and exploits a certain structure of $U_\tau$ to solve most events uniformly.

\begin{definition}[c-consistent]\label{def:c_consistent}
Let $U=U(A_0,\dots, A_n)$ be a union, where $A_i=(S_i,E_i,\delta_i,\iota_i)$ is a linear TS for all $i\in \{0,\dots, n\}$, and let $c\in \mathbb{N}$.
We say an event $e\in E(U)$ is \emph{c-consistent} (in $U$), if the following condition is satisfied for all $i\in \{0,\dots, n\}$:
if $s\edge{e}s'\in A_i$, then $e$ occurs always exactly $c$ times in a row in $A_i$, that is, there are states $s,s'\in \{s_0,\dots,s_c\}\subseteq S_i $ such that $s_0\edge{e}\dots\edge{e} s_c$ and $\neg \edge{e}s_0$ and $\neg s_c\edge{e}$.
\end{definition}

\begin{definition}[thinly distributed]\label{def:thinly_distributed}
Let $U=U(A_0,\dots, A_n)$ be a union, where $A_i=(S_i,E_i,\delta_i,\iota_i)$ is a linear TS for all $i\in \{0,\dots, n\}$, and let $e\in E(U)$ such that $e$ is $c$-consistent for some $c\in\{1,b\}$.
We say $e$ is \emph{thinly distributed} (in $U$) if the following condition is satisfied for all $i\in \{0,\dots, n\}$:
if $e\in E_i$, then there is exactly one path (with pairwise distinct states) $s_0\edge{e}\dots\edge{e}s_c$ in $A_i$.
\end{definition}

\begin{example}
Every event of $U_\tau$ is either $b$-consistent as, for example, $k$ and $X_0,\dots, X_{m-1}$, or $1$-consistent as, for example, $o_0$ and $o_1$.
Moreover, the event $o_1$ occurs once at the edge $h_{1,2b+4}\edge{o_1}h_{1,2b+5}$ and,
for all $j\in \{0,1,2,3\}$, the it occurs once at the edge $d_{j,2}\edge{o_1}d_{j,3}$.
No other gadget of $U_\tau$ applies $o_1$.
Thus, $o_1$ is thinly distributed in $U_\tau$.
Moreover, for all $i\in \{0,\dots, m-1\}$, if $X_i$ occurs in a gadget of $U_\tau$, then it occurs exactly once $b$-times in a row in this gadget.
Hence, $X_i$ is thinly distributed.
\end{example}

\begin{lemma}\label{lem:easy_solvability}
Let $\tau\in \{\tau_{PT}^b,\tau_{PPT}^b\}$.
Let $U=U(A_0,\dots, A_n)$ be a union, where $A_i=(S_i,E_i,\delta_i,\iota_i)$ is a linear TS for all $i\in \{0,\dots, n\}$, such that every event $e\in E(U)$ is $1$-consistent or $b$-consistent, and let $a\in E(U)$ be a thinly distributed event and $q\in S_i$ a state such that $\neg \edge{a}$, where $i\in \{0,\dots, n\}$ is arbitrary but fixed.
If one of the following conditions is satisfied, then there is a $\tau$-region of $U$ that solves $(a,q)$:
\begin{enumerate}

\item\label{lem:easy_solvability_not_or_initial}
$a\not\in E_i$ or $e\in E_i$ and $q$ occurs after $a$;
\item\label{lem:easy_solvability_preceded}
$a\in E_i$ and $a$ occurs after $q$ and there is an event $x\in E_i\setminus \{a\}$ such that $\edge{x}z\edge{a}$ in $A_i$ and
\begin{enumerate}
\item
$x$ is thinly distributed and
\item
for all $j\in \{0,\dots, n\}$, if $a,x\in E_j$, then $x$ does not occur after $a$ in $A_j$ and
\item
if $a$ is $b$-consistent, then $x$ is $1$-consistent.
\end{enumerate}
\end{enumerate}
\end{lemma}
\begin{proof}
(1):
The following $\tau$-region $R=(sup, sig)$ solves $(a,q)$:
For all $j\in \{0,\dots, m-1\}$, if $a\in E_j$, then $sup(\iota_j)=0$, otherwise $sup(\iota_j)=b$;
for all $e\in E(U)$, if $e=a$ and $a$ is $b$-consistent, then $sig(e)=(0,1)$;
if $e=a$ and $e$ is $1$-consistent, then $sig(e)=(0,b)$;
otherwise $sig(e)=(0,0)$.

(2):
The following $\tau$-region $R=(sup, sig)$ solves $(a,q)$:
for all $j\in \{0,\dots, n\}$, if $a\in E_j$ and $x\not\in E_j$, then $sup(\iota_j)=0$, otherwise $sup(\iota_j)=b$;
for all $e\in E(A)$, if $e=a$, then $sig(e)=(0,1)$ if $a$ is $b$-consistent, else $sig(e)=(0,b)$;
if $e=x$, then $sig(e)=(1,0)$ if $x$ is $b$-consistent, else $sig(e)=(b,0)$;
otherwise $sig(e)=(0,0)$.
\end{proof}

Armed with these results, we are now able to provide the proof of Lemma~\ref{lem:tau_ppt_model_implies_solvability}:
\begin{proof}[Lemma~\ref{lem:tau_ppt_model_implies_solvability}]
Let $M$ be a one-in-three model of $\varphi$.
We proceed as follows.
First, we apply Lemma~\ref{lem:easy_solvability} to solve most of $U_\tau$'s ESSA.
After that, we explicitly present $\tau$-regions that solve the remaining atoms and, in particular, solve $\alpha$.
This proves that $U_\tau$ has the $\tau$-ESSP and, by Lemma~\ref{lem:essp_implies_ssp}, the \textsc{$\tau$-Solvability}, too.

Let $e\in E(U_\tau)\setminus\{k, y_0, y_1\}$, let $G$ be a gadget of $U_\tau$ and let $s\in S(G)$ such that $\neg s\edge{e}$, where
all of $e, G$ and $s$ are arbitrary but fixed.
For a start, we notice that $e$ is thinly distributed.
Moreover, recall that, for all $i\in \{0,\dots, m-1\}$, the clause $\zeta_i=\{X_{i_0}, X_{i_1}, X_{i_2}\}$ satisfies $i_0 < i_1 < i_2$.
Consequently, if $e\not\in \{o_0,o_1\}$ or if $e\in \{o_0,o_1\}$ and $G\not=H_1$, then $e$ satisfies Condition~\ref{lem:easy_solvability_not_or_initial} or Condition~\ref{lem:easy_solvability_preceded} of Lemma~\ref{lem:easy_solvability}.
Thus, by Lemma~\ref{lem:easy_solvability}, the atom $(e,s)$ is $\tau$-solvable.

It remains to argue that the remaining atoms are also $\tau$-solvable.
For convenience, we let $I=\{\iota_G\mid \text{$G$ is a gadget of $U$}\}$ be the set of the initial states of the gadgets of $U$.

For a start, we argue for the solvability of $k$.
The following region $R=(sup, sig)$ solves $(k,s)$ for all relevant $s\in S(H_1)$;
in particular, it solves $(k ,h_{1,2b+4})$.
Figure~\ref{fig:example_for_pure_essp} presents a concrete example of $R$ for the union $U_\tau$ that originates from $\varphi$ of Example~\ref{ex:varphi}.
Let's start with the support of the initial states:
if $s\in \{h_{1,0}, f_{0,0}, \dots, f_{2m-1,0}, m_{0,0},\dots, m_{m-1,0}, t_{0,0}, \dots, t_{m-1,0}\}$, then $sup(s)=0$;
if $s\in \{d_{0,0}, \dots, d_{3,0}, g_{0,0}, \dots, g_{2m-1,0}\}$, then $sup(s)=b$.
The signature is defined as follows:
for all $e\in E(U_\tau)$, if $e=k$, then $sig(e)=(0,1)$;
if $e\in \{o_0,o_1\}$, then $sig(e)=(b,0)$;
if $e\in \{k_0,\dots, k_3\}$, then $sig(e)=(0,b)$;
if $e\in M$, then $sig(e)=(1,0)$;
otherwise $sig(e)=(0,0)$.

The following region $R=(sup, sig)$ solves $(k,s)$ for all other relevant states of $U_\tau$:
$sup(h_{1,0})=0$;
for all $s\in I\setminus\{h_{1,0}\}$, $sup(s)=b$;
for all $e\in E(U_\tau)$, if $e=k$, then $sig(e)=(0,1)$;
if $e=y_0$, then $sig(e)=(b,0)$;
otherwise $sig(e)=(0,0)$.
This proves the solvability of $k$.

In the following, we argue that $(o_0,q)$ is solvable for all relevant $q\in S(H_1)$:
The following region $R=(sup, sig)$ solves $(o_0,s)$ for all $s\in \{h_{1, b+2}, \dots, h_{1, 3b+5}\}$:
for all gadgets $G\in U_\tau$, we define $sup(\iota_G)=0$ for $G$'s the initial state $\iota_G$;
for all $e\in E(U_\tau)$, if $e=o_0$, then $sig(e)=(0,b)$;
otherwise $sig(e)=(0,0)$.

The following region $R=(sup, sig)$ solves $(o_0,s)$ for all $s\in \{h_{1,0}, \dots, h_{1, b}\}$:
$sup(h_{1,0})=b$;
for all $s\in I\setminus\{h_{1,0}\}$, we define $sup(s)=0$;
for all $e\in E(U_\tau)$, if $e=o_0$, then $sig(e)=(0,b)$;
if $e=y_0$, then $sig(e)=(b,0)$;
otherwise $sig(e)=(0,0)$.

Similarly, one argues that $(o_1,q)$ is solvable for all relevant $q\in S(H_1)$.
So far, we have proven the solvability of all $e\in E(U_\tau)\setminus\{k,y_0,y_1\}$.
It remains to argue for the solvability of $y_0$ and $y_1$.

The following region $R=(sup, sig)$ solves $(y_0,s)$ for all $s\in \{h_{1,0}, \dots, h_{1, b-1}\}$:
for all $s\in I$, we define $sup(\iota_G)=0$;
for all $e\in E(U_\tau)$, if $e=y_0$, then $sig(e)=(b,0)$;
if $e=k$, then $sig(e)=(0,1)$;
otherwise $sig(e)=(0,0)$.

The following region $R=(sup, sig)$ solves $(y_0,s)$ for all $s\in S(H_1) \setminus \{h_{1,0}, \dots, h_{1, b-1}\}$:
$sup(h_{i,0})=b$ and for all $s\in I\setminus\{h_{1,0}\}$, we define $sup(s)=0$;
for all $e\in E(U_\tau)$, if $e=y_0$, then $sig(e)=(b,0)$;
if $e=y_1$, then $sig(e)=(0,b)$;
otherwise $sig(e)=(0,0)$.

It is easy to see, that $y_1$ is solvable.
\end{proof}

Altogether, since the construction of $U_\tau$ and thus $A_\tau$ is obviously polynomial, by Lemma~\ref{lem:joining}, Lemma~\ref{lem:tau_ppt_essp_implies_model} and Lemma~\ref{lem:tau_ppt_model_implies_solvability} and the NP-completeness of \textsc{CM1in33Sat}, we have finally proven that \textsc{$\tau_{PPT}^b$-ESSP} and \textsc{$\tau_{PPT}^b$-Solvability} are NP-complete for all $b\in \mathbb{N}^+$.

\subsection{NP-hardness of \textsc{$\tau_{PT}^b$-solvability} and \textsc{$\tau_{PT}^b$-ESSP}}\label{sec:tau_pt_solvability}%
%
In the remainder of section, unless stated explicitly otherwise, we assume that $\tau=\tau_{PT}^b$.
The union $U_\tau$ has the following TS $H_0$ that provides the ESSA $\alpha=(k, h_{0, 4b+1})$:
\begin{center}
\begin{tikzpicture}
\node (init) at (-0.75,0) {$H_{0}=$};
\node (h0) at (0,0) {\nscale{$h_{0,0}$}};
\node (h1) at (1.5,0) {};
\node (dots1) at (1.75,0) {\nscale{$\dots$}};
\node (h_b_1) at (2,0) {};
\node (h_b) at (3.5,0) {\nscale{$h_{0,b}$}};
\node (h_b+1) at (5,0) {};
\node (dots_2) at (5.25,0) {\nscale{$\dots$}};
\node (h_2b_1) at (5.5,0) {};
\node (h_2b) at (7,0) {\nscale{$h_{0,2b}$}};
\node (h_2b+1) at (9,0) {\nscale{$h_{0,2b+1}$}};
\node (h_2b+2) at (10.5,0) {};
\node (dots_3) at (10.75,0) {\nscale{$\dots$}};
\node (h_3b) at (11,0) {};
\node (h_3b+1) at (12.5,0) { \nscale{$h_{0,3b+1}$} };
\node (h_3b+2) at (12.5,-1.2) {};
\node (dots_4) at (12.25,-1.2) {\nscale{$\dots$}};
\node (h_4b) at (12,-1.2) { };
\node (h_4b+1) at (10.5,-1.2) { \nscale{$h_{0,4b+1}$} };
\node (h_4b+2) at (8.75,-1.2) {};
\node (dots_5) at (8.5,-1.2) {\nscale{$\dots$}};
\node (h_5b) at (8.25,-1.2) { };
\node (h_5b+1) at (6.5,-1.2) { \nscale{$h_{0,5b+1}$} };
\node (h_5b+2) at (4.8,-1.2) { };
\node (dots_5) at (4.5,-1.2) {\nscale{$\dots$}};
\node (h_6b) at (4.25,-1.2) { };
\node (h_6b+1) at (2.5,-1.2) { \nscale{$h_{0,6b+1}$} };
\graph {
(h0) ->["\escale{$k$}"] (h1);
(h_b_1)->["\escale{$k$}"] (h_b) ->["\escale{$z$}"] (h_b+1);
(h_2b_1)->["\escale{$z$}"] (h_2b)->["\escale{$o_0$}"] (h_2b+1)->["\escale{$k$}"] (h_2b+2);
(h_3b)->["\escale{$k$}"] (h_3b+1)->["\escale{$z$}"] (h_3b+2);
(h_4b)->[swap, "\escale{$z$}"] (h_4b+1)->[swap, "\escale{$o_1$}"] (h_4b+2);
(h_5b)->[swap, "\escale{$o_1$}"] (h_5b+1)->[swap, "\escale{$k$}"] (h_5b+2);
(h_6b)->[swap, "\escale{$k$}"] (h_6b+1);
};
\end{tikzpicture}
\end{center}
For every $j\in \{0,1,2,3\}$, the union $U_\tau$ has the following gadget $C_j$ that provides $k_{j}$:
\begin{center}
\begin{tikzpicture}[yshift=-5cm]
\node (init) at (-1,0) {$C_j=$};
\foreach \i in {0,...,2} {\coordinate (\i) at (\i*1.55,0);}
\foreach \i in {0,...,2} {\node (p\i) at (\i) {\nscale{$c_{j,\i}$}};}
\node (p3) at (4.5,0) {};
\node (hdots_3) at (4.9,0) {\nscale{$\dots$}};
\node (db+2) at (5.25,0) {};
\node (db+3) at (6.75,0) { \nscale{$c_{j,b+2}$} };
\graph { (p0) ->["\escale{$o_{0}$}"] (p1) ->["\escale{$k_j$}"] (p2) ->["\escale{$o_{1}$}"] (p3);
(db+2) ->["\escale{$o_{1}$}"] (db+3);
};
\end{tikzpicture}
\end{center}
Finally, for all $j\in \{0,\dots, 2m-1\}$ and for all $i\in \{0,\dots, m-1\}$, the union $U_\tau$ has the gadgets $F_j, G_j, M_i$ and $T_i$ as defined in Section~\ref{sec:tau_ppt_solvability}.
Altogether, 
\[
U_\tau=U(H_0,C_0,\dots, C_3,F_0,\dots, F_{2m-1}, G_0,\dots, G_{2m-1},M_0,\dots, M_{m-1},T_0,\dots, T_{m-1}).
\]

\begin{lemma}\label{lem:tau_pt_essp_implies_model}
If $U_\tau$ has the $\tau$-ESSP, then $\varphi$ has a one-in-three model.
\end{lemma}
\begin{proof}
Since $U_\tau$ has the $\tau$-ESSP, there is a $\tau$-region of $U_\tau$ that solves $\alpha$.
Let $R=(sup, sig)$ be such a region.
In the following, we argue that either $sig(k_0)=\dots=sig(k_3)=(0,b)$ or $sig(k_0)=\dots=sig(k_3)=(b,0)$.
As already argued at the end of the proof of Lemma~\ref{lem:tau_ppt_essp_implies_model}, by the functionality of the remaining gadgets, this implies that $\{X\in V(\varphi)\mid sig(X)=(1,0)\}$ or $\{X\in V(\varphi)\mid sig(X)=(0,1)\}$ is a one-in-three model of $\varphi$.

Let $ E_{0}=\{ (m,m) \mid 0\le m \leq b\}$.
By definition, if $sig(k)=(m,m) \in E_0$ then $sup(h_{0,3b+1})\geq m$ and $sup(h_{0,5b+1})\geq m$.
Event $(m,m)$ occurs at every state $s\in S_{\tau_{PT}^b}$ that satisfies $s\geq m$.
Hence, by $\neg h_{0,4b+1}\ledge{(m,m)}$, we get $sup(h_{0,4b+1}) < m$.
Since $sup(h_{0,3b+1}) \geq m$ and $sup(h_{0,4b+1}) < m$, we have $sig^-(z) > sig^+(z)$.
Observe, that $z$ is $b$-consistent.
Thus, by Lemma~\ref{lem:observations}, we have $sup(z)=(1,0)$.
Similarly, we get $sig(o_1)=(0,1)$.
This immediately implies $sup(h_{0,2b})=0$ and $sup(h_{0,3b+1})=b$.
Moreover, by $sig(k)=(m,m) $ and $sup(h_{0,3b+1})=b$ we get $sup(h_{0,2b+1})=b$.
By $sup(h_{0,2b})=0$, this implies $sig(o_0)=(0,b)$.
Thus, we have $sig(o_0)=(0,b)$ and $sig(o_1)=(0,1)$.

Otherwise, if $sig(k)\not\in E_0$, then Lemma~\ref{lem:observations} ensures $sig(k)\in \{(1,0), (0,1)\}$.
If $sig(k)=(0,1)$ then we have $sup(h_{0, 4b+1})=b$, since $s\edge{(0,1)}$ for every state $s\in \{0,\dots, b-1\}$ of $\tau_{PT}^b$.
Moreover, again by $sig(k)=(0,1)$ we have $sup(h_{0,b})=sup(h_{0,3b+1})=b$ and $sup(h_{0,2b+1})=sup(h_{0, 5b+1})=0$.
By $sup(h_{0, 3b+1})=sup(h_{0,4b+1})=b$ we have $sig(z)\in E_0$, which together with $sup(h_{0,b})=b$ implies $sup(h_{0, 2b})=b$.
Thus, by $sup(h_{0,2b})=b$ and $sup(h_{0, 2b+1})=0$, it is $sig(o_0)=(b,0)$.
Moreover, by $sup(h_{0, 4b+1})=b$ and $sup(h_{0, 5b+1 })=0$, we conclude $sig(o_1)=(1,0)$.
Hence, we have $sig(o_0)=(b,0)$ and $sig(o_1)=(1,0)$.
Similar arguments show that $sig(k)=(1,0)$ implies $sig(o_0)=(0,b)$ and $sig(o_1)=(0,1)$.

So far we have argued that if $(sup, sig)$ is a $\tau_{PT}^b$-region of $U_\tau$ that solves $\alpha$, then either $sig(o_0)=(0,b)$ and $sig(o_1)=(0,1)$ or $sig(o_0)=(b,0)$ and $sig(o_1)=(1,0)$.
One easily finds out that if $sig(o_0)=(0,b)$ and $sig(o_1)=(0,1)$, then $sup(c_{j,1})=b$ and $sup(c_{j,2})=0$ and thus $sig(k_j)=(b,0)$ for all $j\in \{0,\dots, 3\}$.
Similarly, if $sig(o_0)=(b,0)$ and $sig(o_1)=(1,0)$, then $sup(c_{j,1})=0$, $sup(c_{j,2})=b$ and $sig(k_j)=(0,b)$ for all $j\in \{0,\dots,3\}$.
By the functionality of $F_j,G_j$ this implies $z_j\in E_0$.
Moreover, by the functionality of $M_i$, this implies if $sig(k_1)=(0,b)$, then $sig(X_i)\in \{(1,0)\}\cup E_0$ and if $sig(k_1)=(b,0)$, then $sig(X_i)\in \{(0,1), (0,0)\}$ for all $i\in \{0,\dots, m-1\}$.
Similar to the arguments for $\tau_{PPT}^b$, one argues that the gadgets $T_0,\dots, T_{m-1}$ then ensure that $\{e\in V(\varphi)\mid sig(e)=(0,1)\}$ or $\{e\in V(\varphi)\mid sig(e)=(1,0)\}$ defines a sought model of $\varphi$.
Thus, if $U_\tau$ has the $\tau$-ESSP or is $\tau$-solvable, which implies that $\alpha$ is $\tau$-solvable, then $\varphi$ has a one-in-three model.
\end{proof}

The following lemma is dedicated to the opposite direction:
\begin{lemma}\label{lem:tau_pt_model_implies_solvability}
If $\varphi$ has a one-in-three model, then $U_\tau$ has the $\tau$-ESSP and the $\tau$-SSP.
\end{lemma}
\begin{proof}
In the following, we argue that if $M$ is a one-in-three model of $\varphi$, then $U_\tau$ has the $\tau$-ESSP and thus has also the $\tau$-SSP, since all gadgets are linear.
Notice that if $e$ is an event and $G$ is a gadget of $U_\tau$ such that $e$ does not occur in $G$, then $(e,s)$ is $\tau$-solvable for all $s\in S(G)$.
A solving region $R=(sup, sig)$ is defined as follows:
for all gadgets $G'$ of $U_\tau$, if $e\in E(G')$, then $sup(\iota_{G'})=b$, otherwise $sup(\iota_{G'})=0$;
for all events $e'\in E(U_\tau)$, $if e'=e$, then $sig(e')=(b,b)$;
otherwise $sig(e')=(0,0)$.
Thus, in the following, we only argue for valid atoms $(e,s)$ where $e$ and $s$ occur in the same gadget.

\medskip
Let $e\in E(U_\tau)\setminus\{k, z\}$, let $G$ be a gadget of $U_\tau$ and let $s\in S(G)$ such that $\neg s\edge{e}$, where
all of $e, G$ and $s$ are arbitrary but fixed.
The event $e$ is thinly distributed.
Moreover, if $e\not\in \{o_0,o_1\}$ or if $e\in \{o_0,o_1\}$ and $G\not=H_0$, then $e$ satisfies Condition~\ref{lem:easy_solvability_not_or_initial} or Condition~\ref{lem:easy_solvability_preceded} of Lemma~\ref{lem:easy_solvability}.
Thus, by Lemma~\ref{lem:easy_solvability}, in these cases, the atom $(e,s)$ is $\tau$-solvable.

For convenience, let $I=\{\iota_G\mid \text{$G$ is a gadget of $U_\tau$}\}$ be the set of the initial states of the gadgets of $U_\tau$.

\medskip
To complete the proof for the solvability of $o_0$ and $o_1$, it remains to argue that $(o_0, s)$ and $(o_1,s')$ are solvable for all relevant $s,s'\in S(H_0)$:
By Lemma~\ref{lem:easy_solvability}.\ref{lem:easy_solvability_not_or_initial}, the atoms $(o_0,s)$ and $(o_1,s')$ are solvable for all $s\in \{h_{0, 2b+1},\dots, h_{0,6b+1}\}$ and for all $s'\in \{h_{0, 5b+1},\dots, h_{0,6b+1}\}$.
The following region $R=(sup, sig)$ solves $(o_0,s)$ for all $s\in \{h_{0,0},\dots, h_{0,2b-1}\}$:
$sup(h_{0,0})=b$;
for all $s\in I\setminus\{h_{0,0}\}$, we define $sup(s)=0$;
for all $e\in E(U_\tau)$, if $e=o_0$, then $sig(e)=(0,b)$;
if $e=z$, then $sig(e)=(1,0)$;
otherwise $sig(e)=(0,0)$.

The following region $R=(sup, sig)$ solves $(o_1,s)$ for all $s\in \{h_{0,0},\dots, h_{0,4b}\}\setminus\{h_{0,2b}\}$ and uses the model $M$ of $\varphi$:
We start with the support of the initial states:
$sup(h_{0,0})=0$;
if $s\in \{f_{0,0}, \dots, f_{2m-1,0}, m_{0,0},\dots, m_{m-1,0}, t_{0,0}, \dots, t_{m-1,0}\}$, then $sup(s)=0$;
if $s\in \{c_{0,0}, \dots, c_{3,0}\}\cup\{g_{0,0}, \dots, g_{2m-1,0}\}$, then $sup(s)=b$.
The signature is defined as follows:
for all $e\in E(U_\tau)$, if $e=o_1$, then $sig(e)=(b,b)$;
if $e=z$, then $sig(e)=(0,1)$;
if $e=o_0$, then $sig(e)=(b,0)$;
if $e\in \{k_0,\dots, k_3\}$, then $sig(e)=(0,b)$;
if $e\in M$, then $sig(e)=(1,0)$;
otherwise $sig(e)=(0,0)$.

The following region $R=(sup, sig)$ solves $(o_1,h_{0,2b})$:
$sup(s)=0$ for all $s\in I$;
for all $e\in E(U_\tau)$, if $e=o_1$, then $sig(o_1)=(b,b)$;
if $e=o_0$, then $sig(e)=(0,b)$;
otherwise $sig(e)=(0,0)$.
This proves the solvability of $o_1$.

Since $z$ occurs only in $H_0$, for the solvability of $z$, it remains to argue that $(z,s)$ is $\tau$-solvable for all relevant $s\in S(H_0)$.
The following region $R=(sup, sig)$ does this for all $s\in S(H_0)\setminus\{h_{0,6b+1}\}$ and uses the model $M$ of $\varphi$.
Moreover, this region also solves $(k,s)$ for all $s\in S(H_0)$ and, thus, proves the solvability of $k$:
if $s\in \{h_{0,0}, f_{0,0}, \dots, f_{2m-1,0}, m_{0,0},\dots, m_{m-1,0}, t_{0,0}, \dots, t_{m-1,0}\}$, then $sup(s)=0$;
if $s\in \{c_{0,0}, \dots, c_{3,0}, g_{0,0}, \dots, g_{2m-1,0}\}$, then $sup(s)=b$;
for all $e\in E(U_\tau)$, if $e=z$, then $sig(z)=(b,b)$;
if $e=k$, then $sig(e)=(0,1)$;
if $e\in \{o_0,o_1\}$, then $sig(e)=(b,0)$;
if $e\in \{k_0,\dots, k_3\}$, then $sig(e)=(0,b)$;
if $e\in M$, then $sig(e)=(1,0)$;
otherwise $sig(e)=(0,0)$.

\medskip
One easily finds that $(z, h_{0,6b+1})$ is $\tau$-solvable.
Altogether, this proves that if $\varphi$ is one-in-three satisfiable, then $U_\tau$ has the $\tau$-ESSP.
Since all gadgets are linear, this completes the proof.
\end{proof}

\subsection{NP-hardness of \textsc{$\tau_{PPT}^b$-SSP} and \textsc{$\tau_{PT}^b$-SSP}}\label{sec:tau_ppt_tau_pt_ssp}%
%
In the remainder of this section, unless stated explicitly otherwise, let $\tau\in \{\tau_{PPT}^b, \tau_{PT}^b\}$ be arbitrary but fixed.
The union $U_\tau$ has the following gadget $H_2$ that provides the atom $\alpha = (h_{2,0}, h_{2,b})$:
\begin{center}
\begin{tikzpicture}
\node at (-1,0) {$H_2=$};
\node (h0) at (-0.25,0) {\nscale{$h_{2,0}$}};
\node (h1) at (1,0) {\nscale{}};
\node (hdots_1) at (1.25,0) {\nscale{$\dots$}};
\node (hb_1) at (1.5,0) {};
\node (hb) at (2.75,0) {\nscale{$h_{2,b}$}};
\node (hb+1) at (4.5,0) {\nscale{$h_{2,b+1}$}};
\node (hb+2) at (6,0) {};
\node (hdots_2) at (6.25,0) {\nscale{$\dots$}};
\node (h2b) at (6.5,0) {};
\node (h2b+1) at (8,0) {\nscale{$h_{2,2b+1}$}};
\node (h2b+2) at (10,0) {\nscale{$h_{2,2b+2}$}};
\node (h2b+3) at (11.5,0) {};
\node (hdots_3) at (11.75,0) {\nscale{$\dots$}};
\node (h3b+1) at (12,0) {};
\node (h3b+2) at (13.5,0) {\nscale{$h_{2,3b+2}$}};
\graph {
(h0) ->["\escale{$k$}"] (h1);
(hb_1)->["\escale{$k$}"] (hb) ->["\escale{$o_0$}"] (hb+1)->["\escale{$k$}"] (hb+2);
(h2b) ->["\escale{$k$}"] (h2b+1)->["\escale{$o_2$}"] (h2b+2)->["\escale{$k$}"] (h2b+3);
(h3b+1) ->["\escale{$k$}"] (h3b+2);
};
\end{tikzpicture}
\end{center}
Moreover, the union $U_\tau$ has every gadget that has been defined for $U_{\tau_{PPT}^b}$ in Section~\ref{sec:tau_ppt_solvability} except for $H_1$.
Altogether, $U_\tau$ is defined as follows:
\[
U_\tau=U(H_2,D_0,\dots, D_3, F_0,\dots, F_{2m-1}, G_0,\dots, G_{2m-1},M_0,\dots, M_{m-1}, T_0,\dots, T_{m-1}).
\]

\begin{lemma}\label{lem:tau_ppt_tau_pt_ssp_implies_model}
If $U_\tau$ has the $\tau$-SSP, then $\varphi$ has a one-in-three model.
\end{lemma}
\begin{proof}
Since $U_\tau$ has the $\tau$-SSP, there is a $\tau$-region that solves $\alpha$.
Let $R=(sup, sig)$ be such a region.
We argue that the signature of the variable events define a sought model of $\varphi$:
The event $k$ occurs $b$ times in a row at $h_{2,0}$.
Thus, by Lemma~\ref{lem:observations}, a region $(sup, sig)$ solving $(h_{2,0}, h_{2,b})$ satisfies $sig(k)\in \{(1,0), (0,1)\}$.
If $sig(k)=(1,0)$, then $sup(h_{2,b})=sup(h_{2,2b+1})=b$ and $sup(h_{2,b+1})=sup(h_{2,2b+2})=0$.
This implies $sig(o_0)=sig(o_2)=(b,0)$ and, thus, $sig(k_j)=(0,b)$ for all $j\in \{0,\dots, 3\}$.
Otherwise, if $sig(k)=(0,1)$ then $sup(h_{2,b})=sup(h_{2,2b+1})=0$ and $sup(h_{2,b+1})=sup(h_{2,2b+2})=b$.
This implies $sig(o_0)=sig(o_2)=(0,b)$ and $sig(k_j)=(b,0)$ for all $j\in \{0,\dots, 3\}$.
Just like before, this proves the one-in-three satisfiability of $\varphi$.
\end{proof}

The following lemma addresses the opposite direction:
\begin{lemma}\label{lem:tau_ppt_tau_pt_model_implies_ssp}
If $\varphi$ has a one-in-three model, then $U_\tau$ has the $\tau$-SSP.
\end{lemma}
\begin{proof}
Let $M$ be a one-in-three model of $\varphi$.
We briefly argue, that $U_\tau$ has the $\tau$-SSP.
For start, let $e\in E(U_\tau)\setminus\{k\}$ be arbitrary but fixed.
The event $e$ is thinly distributed.
Moreover, if $s\in S(U_\tau)\setminus S(H_2)$ and $\neg s\edge{e}$, then, by Lemma~\ref{lem:easy_solvability}, $(e,s)$ is $\tau$-solvable.
By Lemma~\ref{lem:essp_implies_ssp}, this implies that if $(s,s')$ is an SSA of $U_\tau$ such that $s,s'\not\in S(H_2)$, then $(s,s')$ is $\tau$-solvable.
Thus, it remains to show that any SSA $(s,s')$ of $U_\tau$ where $s,s'\in S(H_2)$ is $\tau$-solvable, too.
The corresponding regions can be defined similar to those from Section~\ref{sec:tau_ppt_solvability} and Section~\ref{sec:tau_pt_solvability}.
In particular, the atom $(h_{2,0}, h_{2,b})$ can be solved by a region that is defined in accordance to the region $R=(sup, sig)$ of Section~\ref{sec:tau_ppt_solvability} that solves $(k ,h_{1,2b+4})$; one simply has to replace $sup(h_{1,0})=0$ by $sup(h_{2,0})=0$ and to ignore the events $y_0$ and $y_1$.
The resulting region also solves $(s,s')$ if $s\not=s'\in \{h_{2,0},\dots, h_{2,b}\}$ or $s\not=s'\in \{h_{2,b+1},\dots, h_{2,2b+1}\}$ or $s\not=s'\in \{h_{2,2b+2},\dots, h_{2,3b+2}\}$.
Finally, it is easy to see that all states of $\{h_{2,0},\dots, h_{2,b}\}$ are separable from all states of $\{h_{2,b+1},\dots, h_{2,2b+1}\}\cup \{h_{2,2b+2},\dots, h_{2,3b+2}\}$, and that all states of $\{h_{2,b+1},\dots, h_{2,2b+1}\}$ are separable from all states of $\{h_{2,2b+2},\dots, h_{2,3b+2}\}$.
Altogether, this proves that if $M$ has a one-in-three model, then $U_\tau$ has the $\tau$-SSP.
\end{proof}

\subsection{NP-hardness of \textsc{$\tau$-solvability} and \textsc{$\tau$-ESSP} for $\tau=\tau_{\mathbb{Z}PPT}$ and $\tau=\tau_{\mathbb{Z}PT}$}\label{sec:tau_zppt_tau_zpt_solvability}%
%
In the remainder of this section, unless stated explicitly otherwise, let $\tau\in \{\tau_{\mathbb{Z}PT}^b,\tau_{\mathbb{Z}PPT}^b\}$ and let $E_0=\{(m,m) \vert 1 \leq m\leq b\}\cup \{  0 \}$.
The union $U_\tau$ has the following TS $H_3$ that provides the atom $\alpha=(k, h_{3,1,b-1})$:
\begin{center}
\begin{tikzpicture}
\node (init) at (-0.9,0) {$H_3=$};
\node (h0) at (0,0) {\nscale{$h_{3,0,0}$}};
\node (h1) at (1.5,0) {\nscale{}};
\node (h_2_dots) at (1.75,0) {\nscale{$\dots$}};
\node (h_b_2) at (2,0) {};
\node (h_b_1) at (3.5,0) {\nscale{$h_{3,0,b-1}$}};
\node (h_b) at (5.5,0) {\nscale{$h_{3,0,b}$}};
\node (h_b+1) at (0,-1) {\nscale{$h_{3,1,0}$}};
\node (h_b+2) at (1.5,-1) {};
\node (h_b+2_dots) at (1.75,-1) {\nscale{$\dots$}};
\node (h_2b_3) at (2,-1) {};
\node (h_2b_2) at (3.5,-1) {\nscale{$h_{3,1,b-1}$}};
\graph{
(h0) ->["\escale{$k$}"] (h1);
(h0) ->["\escale{$u$}", swap](h_b+1)->["\escale{$k$}"](h_b+2);
(h_b_2)->["\escale{$k$}"] (h_b_1)->["\escale{$k$}"] (h_b);
(h_2b_3)->["\escale{$k$}"] (h_2b_2);
(h_2b_2)->[swap, "\escale{$z$}"] (h_b);
};
\end{tikzpicture}
\end{center}
Moreover, for all $j\in \{0,\dots, m-1\}$, the union $U_\tau$ has the following gadgets $F_j$ and $G_j$ that use the variable $X_j$ as event:
\begin{center}
\begin{tikzpicture}
\begin{scope}
\node at (-0.9,0) {$F_j=$};
\node (f0) at (0,0) {\nscale{$f_{j,0,0}$}};
\node (f1) at (1.5,0) {\nscale{}};
\node (f_2_dots) at (1.75,0) {\nscale{$\dots$}};
\node (f_b_1) at (2,0) {};
\node (f_b) at (3.5,0) {\nscale{$f_{j,0,b-1}$}};
\node (f_b') at (5.25,0) {\nscale{$f_{j,0,b}$}};
\node (f_b+1) at (0,-1) {\nscale{$f_{j,1,0}$}};
\node (f_b+2) at (1.5,-1) {};
\node (f_b+2_dots) at (1.75,-1) {\nscale{$\dots$}};
\node (f_2b_1) at (2,-1) {};
\node (f_2b) at (3.5,-1) {\nscale{$f_{j,1,b-1}$}};
\graph{
(f0) ->["\escale{$k$}"] (f1);
(f0) ->["\escale{$v_j$}", swap](f_b+1)->["\escale{$k$}"](f_b+2);
(f_b_1)->["\escale{$k$}"] (f_b)->["\escale{$k$}"] (f_b');
(f_2b_1)->["\escale{$k$}"] (f_2b);
(f_2b)->[swap, "\escale{$X_{j}$}"] (f_b');
};
\end{scope}
\begin{scope}[xshift= 8cm]
\node at (-0.9,0) {$G_j=$};
\node (f0) at (0,0) {\nscale{$g_{j,0}$}};
\node (f1) at (1.5,0) {\nscale{}};
\node (f_2_dots) at (1.75,0) {\nscale{$\dots$}};
\node (f_b_1) at (2,0) {};
\node (f_b) at (3.5,0) {\nscale{$g_{j,b}$}};
\node (f_b+1) at (5,0) {\nscale{$g_{j,b+1}$}};
\graph{
(f0) ->["\escale{$k$}"] (f1);
(f_b_1)->["\escale{$k$}"] (f_b)->["\escale{$X_j$}"](f_b+1);
};
\end{scope}
\end{tikzpicture}
\end{center}
Finally, for all $i\in \{0,\dots, m-1\}$, the union $U_\tau$ has the following gadget $T_i$ that uses the variables of the clause $\zeta_i=\{X_{i_0}, X_{i_1}, X_{i_2}\}$ as events:
\begin{center}
\begin{tikzpicture}
\node at (-1.4,0) {$T_i=$};
\node (t0) at (-0.6,0) {\nscale{$t_{i,0}$}};
\node (t1) at (0.7,0) {\nscale{}};
\node (t_2_dots) at (0.95,0) {\nscale{$\dots$}};
\node (t_b_1) at (1.2,0) {};
\node (t_b) at (2.5,0) {\nscale{$t_{i,b}$}};
\node (t_b+1) at (4.2,0) {\nscale{$t_{i,b+1}$}};
\node (t_b+2) at (6,0) {\nscale{$t_{i,b+2}$}};
\node (t_b+3) at (7.75,0) {\nscale{$t_{i,b+3}$}};
\node (t_b+4) at (9.5,0) {\nscale{$t_{i,b+4}$}};
\node (t_b+5) at (11,0) {};
\node (t_b+5_dots) at (11.25,0) {\nscale{$\dots$}};
\node (t_2b+3) at (11.5,0) {\nscale{}};
\node (t_2b+4) at (13,0) {\nscale{$t_{i, 2b+4}$}};
\graph{
(t0) ->["\escale{$k$}"] (t1);
(t_b_1)->["\escale{$k$}"] (t_b) ->["\escale{$X_{i_0}$}"] (t_b+1)->["\escale{$X_{i_1}$}"] (t_b+2)->["\escale{$X_{i_2}$}"] (t_b+3)->["\escale{$z$}"] (t_b+4)->["\escale{$k$}"] (t_b+5);
(t_2b+3)->["\escale{$k$}"] (t_2b+4);
};
\end{tikzpicture}
\end{center}
Altogether, \[U_\tau=(H_3,F_0,G_0,\dots, F_{m-1}, G_{m-1}, T_0,\dots,T_{m-1}).\]

\begin{lemma}\label{lem:tau_zppt_tau_zpt_essp_implies_model}
If $U_\tau$ has the $\tau$-ESSP, then $\varphi$ has a one-in-three model.
\end{lemma}
\begin{proof}
Since $U_\tau$ has the $\tau$-ESSP, there is a $\tau$-region, that solves $\alpha$.
Let $R=(sup, sig)$ be such a region.
In the following, we first argue that $sig(k)\in \{(1,0), (0,1)\}$ and $sig(z)\in E_0$.
Secondly, we show that this implies that $M=\{X\in V(\varphi)\mid sig(X)=1\}$ is a one-in-three model of $\varphi$.

Let $(sup, sig)$ be a $\tau$-region that solves $\alpha$, that is, $\neg sup(h_{3,1,b-1})\ledge{sig(k)}$.
If $sig(k)\in E_0$, then we inductively obtain $sup(h_{3,1,0})=sup(h_{3,1,b-1})$.
This contradicts $\neg sup(h_{3,1,b-1})\ledge{sig(k)}$.
Moreover, if $e\in \{0,\dots, b\}$, then $s\edge{e}$ for all $s\in S_\tau$.
Consequently, we have $sig(k)\not\in E_0 \cup \{1,\dots, b\}$.

The event $k$ occurs $b$ times in a row.
Therefore, by Lemma~\ref{lem:observations}, we get $sig(k)\in \{(1,0), (0,1)\}$.
Moreover, if $sig(k)= (1,0)$, then $sup(h_{3,0,b})=0$ and if $sig(k)= (0,1)$, then $sup(h_{3,0,b})=b$.
If $s\in \{0,\dots, b-1\}$ then $s\ledge{(0,1)}$ is true, and if $s\in \{1,\dots, b\}$, then $s\ledge{(1,0)}$ is true.
Consequently, by $\neg sup(h_{3,1,b-1})\ledge{sig(k)}$, if $sig(k)=(0,1)$, then $sup(h_{3,1,b-1})=b$, and if $sig(k)=(1,0)$, then $sup(h_{3,1,b-1})=0$.
For both cases, this implies $sup(h_{3,0,b})=sup(h_{3,1,b-1})$ and thus $sig(z)\in E_0$.

\medskip
We now argue that this makes $M$ a one-in-three model of $\varphi$.
Let $i\in \{0,\dots, m-1\}$ be arbitrary but fixed.
By the definition of $\tau$-regions, if $p_i$ is defined by
\begin{center}
\begin{tikzpicture}
\node (init) at (-1,0) {$p_i=$};
\node (t_b) at (0,0) {\nscale{$sup(t_{i,b})$}};
\node (t_b+1) at (3,0) {\nscale{$sup(t_{i,b+1})$}};
\node (t_b+2) at (6,0) {\nscale{$sup(t_{i,b+2})$}};
\node (t_b+3) at (9,0) {\nscale{$sup(t_{i,b+3})$}};

\graph{
 (t_b) ->["\escale{$sig(X_{i_0})$}"] (t_b+1)->["\escale{$sig(X_{i_1})$}"] (t_b+2)->["\escale{$sig(X_{i_2})$}"] (t_b+3);
 };
\end{tikzpicture}
\end{center}
\noindent
then $p_i$ is a directed labeled path in $\tau$.
By $sig(z)\in E_0$ and $t_{i,b+3}\edge{z}t_{i,b+4}$ we obtain that $sup(t_{i,b+3})=sup(t_{i,b+4})$.
Moreover, $k$ occurs $b$ times in a row at $t_{i,0}$ and $t_{i,b+4}$.
By Lemma~\ref{lem:observations}, this implies if $sig(k)=(0,1)$, then $sup(t_{i,b})=b$ and $sup(t_{i,b+4})=0$.
Similarly, and if $sig(k)=(1,0)$, then $sup(t_{i,b})=0$ and $sup(t_{i,b+4})=b$.
Altogether, we obtain that the following conditions are true:
If $sig(z)\in E_0$ and $sig(k) = (1,0)$, then the path $p_i$ starts at $0$ and terminates at $b$, and if $sig(z)\in E_0$ and $sig(k) = (0,1)$, then the path $p_i$ starts at $b$ and terminates at $0$.
In particular, both cases imply that there has to be at least one event $X\in \{X_{i_0}, X_{i_1}, X_{i_2}\}$ whose signature satisfies $sig(X)\not\in E_0$.
Via the functionality of the gadgets $F_0,G_0,\dots, F_{m-1},G_{m-1} $, our reduction ensures that $X$ is unique.
More exactly, the aim of $F_0,G_0,\dots, F_{m-1},G_{m-1} $ is to restrict the possible signatures for the variable events as follows:
\begin{itemize}
\item
If $sig(k) = (1,0)$, then $X\in V(\varphi)$ implies $sig(X)\in E_0 \cup \{ b \}$, and
\item
if $sig(k) = (0,1)$, then $X\in V(\varphi)$ implies $sig(X)\in E_0 \cup \{ 1 \}$.
\end{itemize}
Before we argue that $F_0,G_0,\dots, F_{m-1},G_{m-1} $ satisfy the announced functionality, we first argue that these restrictions of the signature of $X_{i_0}, X_{i_1}, X_{i_2}$ ensure that there is exactly one variable event $X\in \{X_{i_0}, X_{i_1}, X_{i_2}\}$ with $sig(X)\not\in E_0$.
Remember that, by definition, if $sig(X)\in E_0$ then $sig^-(X) + sig^+(X) = \vert sig(X)\vert = 0$.

\medskip
For a start, let $sig(z)\in E_0$ and $sig(k) = (1,0)$, which implies that $p_i$ starts at $0$ and terminates at $b$.
Moreover, assume $sig(X)\in E_0 \cup \{ b \}$.
By Lemma~\ref{lem:observations}, we obtain
\begin{equation}\label{eq:modulo=b}
(\vert sig(X_{i_0})\vert + \vert sig( X_{i_1} ) \vert +\vert sig( X_{i_2}) \vert)  \equiv b \text{ mod } (b+1)
\end{equation}
If $sig(X_{i_0}), sig(X_{i_1}), sig(X_{i_2})\in E_0$, then $\vert sig(X_{i_0})\vert = \vert sig( X_{i_1} ) \vert =\vert sig( X_{i_2}) \vert=0$.
This contradicts Equation~1.
Hence, there has to be at least one variable event $X\in \{ X_{i_0}, X_{i_1} , X_{i_2}  \}$ such that $sig(X)=b$.
In the following, we argue that $X$ is unique.

\medskip
Assume, for a contradiction, that there are two different variable events $X, Y\in \{ X_{i_0}, X_{i_1} , X_{i_2}  \}$ such that $sig(X)=sig(Y)=b$ and that $sig(Z)\in E_0$ for $Z \in \{ X_{i_0}, X_{i_1} , X_{i_2}  \}\setminus \{X, Y\}$.
By symmetry and transitivity, we obtain

\eject

\hbox{}
\vspace*{-12mm}
\begin{align}
& b \equiv  (\vert sig(X_{i_0})\vert + \vert sig( X_{i_1} ) \vert +\vert sig( X_{i_2}) \vert)  \text{ mod } (b+1)  && \vert (1) \\ 
& (\vert sig(X_{i_0})\vert + \vert sig( X_{i_1} ) \vert +\vert sig( X_{i_2}) \vert)  \equiv 2b \text{ mod } (b+1)  && \vert \text{assumpt.} \\
& b  \equiv 2b  \text{ mod } (b+1)  && \vert (2),(3) \\ 
& 2b  \equiv (b-1)  \text{ mod } (b+1)  && \vert \text{def. }  \equiv  \\ 
& b  \equiv  (b-1) \text{ mod } (b+1)  &&\vert  (4),(5)\\
& \exists m\in \mathbb{Z}: m(b+1)=1 && \vert  (6)
\end{align}
By Equation~7, we get $b=0$, a contradiction.
Similarly, if we assume that $\vert sig(X_{i_0})\vert = \vert sig( X_{i_1} ) \vert =\vert sig( X_{i_2}) \vert=b$, then we obtain
\begin{align}
& (\vert sig(X_{i_0})\vert + \vert sig( X_{i_1} ) \vert +\vert sig( X_{i_2}) \vert)  \equiv 3b \text{ mod } (b+1)  && \vert \text{assumpt.} \\
& b  \equiv 3b  \text{ mod } (b+1)  && \vert (2),(8) \\ 
& 3b  \equiv (b-2)  \text{ mod } (b+1)  && \vert \text{def. }  \equiv \\ 
& b  \equiv  (b-2) \text{ mod } (b+1)  &&\vert  (9),(10)\\
& \exists m\in \mathbb{Z}: m(b+1)=2 && \vert (11)
\end{align}
By Equation~12, we have $b\in \{0,1\}$, which contradicts $b\geq 2$.
Consequently, if $sig(z)\in E_0$ and $sig(k) = (1,0)$ and $sig(X)\in E_0 \cup \{ b \}$, then there is exactly one variable event $X\in \{X_{i_0}, X_{i_1}, X_{i_2}\}$ with $sig(X)\not\in E_0$.

\medskip
Otherwise, if $sig(z) \in E_0$, $sig(k) = (0,1)$, implying that $p_i$ starts at $b$ and terminates at $0$, and $sig(X)\in E_0 \cup \{ 1 \}$, then the following equation is true:
\begin{equation}\label{eq:modulo=0}
(b+ \vert sig(X_{i_0})\vert + \vert sig( X_{i_1} ) \vert +\vert sig( X_{i_2}) \vert)  \equiv 0 \text{ mod } (b+1)
\end{equation}
This implies $\vert sig(X_{i_0})\vert + \vert sig( X_{i_1} ) \vert +\vert sig( X_{i_2}) \vert)  \equiv 1 \text{ mod } (b+1)$.
If there is more than one $X\in \{X_{i_0}, X_{i_1}, X_{i_2}\}$ such that $sig(X)=1$, then $2  \equiv 1 \text{ mod } (b+1)$ or $3  \equiv 1 \text{ mod } (b+1)$ is true.
If $2  \equiv 1 \text{ mod } (b+1)$, then $b=0$, and if $3  \equiv 1 \text{ mod } (b+1)$, then $b\in \{0,1\}$.
Since $b\geq 2$, both cases yield a contradiction.
Consequently, there is exactly one $X\in \{X_{i_0}, X_{i_1}, X_{i_2}\}$ such that $sig(X)=1$, and if $Y\in \{X_{i_0}, X_{i_1}, X_{i_2}\}\setminus \{X\}$, then $sig(Y)\in E_0$.

Under the assumption that the gadgets $F_0,G_0,\dots, F_{m-1}, G_{m-1}$ behave as announced, we have shown the following:
If $(sup, sig)$ is a $\tau$-region of $U_\tau$ such that $sig(k)\in \{(0,1), (1,0)\}$ and $sig(z)\in E_0$, then, for every $i\in \{0,\dots, m-1\}$, there is exactly one variable event $X\in \{X_{i_0}, X_{i_1}, X_{i_2}\}$ such that $sig(X)\not\in E_0$.
As a result, the set $M=\{X\in V(\varphi) \vert sig(X) \not\in E_0\}$ defines a one-in-three model of $\varphi$.

It remains to argue that the gadgets $F_0,G_0,\dots, F_{m-1},G_{m-1} $ behave as announced.
Let $j\in \{0,\dots, m-1\}$.
In the following, we show that if $sig(k)=(1,0)$, then $sig(X_j)\in E_0\cup \{b\}$, and if $sig(k)=(0,1)$, then $sig(X_j)\in E_0\cup \{1\}$.

To begin with, let $sig(k)=(1,0)$.
The event $k$ occurs $b$ times in a row at $f_{j,0,0}$ and $g_{j,0}$ and $b-1$ times in a row at $f_{j,1,0}$.
By Lemma~\ref{lem:observations} this implies $sup(f_{j,0,b})=sup(g_{j,b})=0$ and $sup(f_{j,1,b-1})\in \{0,1\}$.
Clearly, if $sup(f_{j,0,b})=sup(f_{j,1,b-1})=0$ then $sig(X_j)\in E_0$.
We argue that $sup(f_{j,1,b-1})=1$ implies $sig(X_j) = b$.

Assume, for a contradiction, that $sig(X_j)\not=b $.
If $sig(X_j)=(m,m)$ for some $m\in \{1,\dots, b\}$, then $-sig^-(X_j)+sig^+(X_j)=\vert sig(X_j) \vert = 0$.
By Lemma~\ref{lem:observations}, this contradicts $sup(f_{j,0,b})\not=sup(f_{j,1,b-1})$.
If $sig(X_j)=(m,n)$ with $m\not=n$, then $\vert sig(X_j)\vert =0$.
By Lemma~\ref{lem:observations}, we have $sup(f_{j,0,b})=sup(f_{j,1,b-1})-sig^-(X_j)+sig^+(X_j)$, which implies $sig(X_j)=(1,0)$.
But, this contradicts $sup(g_{j,b})\lledge{sig(X_j)}$, since $sup(g_{j,b})=0$ and $\neg 0 \ledge{(1,0)}$ in $\tau$.
Finally, if $sig(X_j) = e \in  \{0,\dots, b-1 \}$, then we have  $1 + e \not\equiv 0 \text{ mod } (b+1)$.
This contradicts $sup(f_{j,1,b-1})\lledge{sig(X_j)}sup(f_{j,0,b})$.
Hence, we have $sig(X_j)=b$.
Overall, it is proven that if $sig(k)=(1,0)$, then $sig(X_j)\in E_0\cup \{b\}$.

To continue, let $sig(k)=(0,1)$.
Similarly to the former case, by Lemma~\ref{lem:observations}, we obtain that $sup(f_{j,0,b})=sup(g_{j,b})=b$ and $sup(f_{j,1,b-1})\in \{b-1,b\}$.
If $sup(f_{j,1,b-1})=b$, then $sig(X_j)\in E_0$.
We show that $sup(f_{j,1,b-1})=b-1$ implies $sig(X_j)=1$.
Assume $sig(X_j)=(m,n)\in E_\tau$.
If $m=n$ or $m>n$, then, by $sup(f_{j,0,b})=sup(f_{j,1,b-1})-sig^-(X_j)+sig^+(X_j)$, we get $sup(f_{j,0,b}) < b$.
This is a contradiction.
If $m < n$ then, by $sup(g_{j,b+1})=sup(g_{j,b})-sig^-(X_j)+sig^+(X_j)$, we get the contradiction $sup(g_{j,b+1}) >  b$.
Hence, $sig(X_j)\in \{0,\dots, b\}$ and $(b-1 + \vert sig(X_j)\vert) \equiv b \text{ mod } (b+1)$.
This implies that $(b+1)$ divides $(\vert sig(X_j)\vert-1)$ and thus $\vert sig(X_j)\vert \equiv 1 \text{ mod } (b+1)$.
Consequently, we obtain $sig(X_j)=1$.
This shows that $sig(k)=(0,1)$ and $z\in E_0$ implies $sig(X_j)\in E_0\cup \{1\}$.
\end{proof}

\newcommand{\freezer}[4]{

\ifstrequal{#4}{0}{
\begin{scope}[nodes={set=import nodes}, xshift= #2cm, yshift=#3 cm]
\coordinate (c00) at (0,0);
\coordinate(c01) at (1,0) ;
\coordinate (c02) at (2,0) ;
\coordinate (c10) at (0,-1) ;
\coordinate (c11) at (2,-1) ;


\node (f00) at (0,0) {\nscale{$f_{#1,0,0}$}};
\node (f01) at (1.5,0) {\nscale{$f_{#1,0,1}$}};
\node (f02) at (3,0) {\nscale{$f_{#1,0,2}$}};
\node (f10) at (0,-1.2) {\nscale{$f_{#1,1,0}$}};
\node (f11) at (2,-1.2) {\nscale{$f_{#1,1,1}$}};
\graph{
(f00) ->[,"\escale{$k$}"] (f01)->[,"\escale{$k$}"] (f02);
(f10) ->[,"\escale{$k$}"] (f11);
(f00) ->[,swap, "\escale{$v_#1$}"] (f10);
(f11) ->[,swap, "\escale{$X_#1$}"] (f02);
};
\end{scope}
}{
\begin{scope}[nodes={set=import nodes}, xshift= #2cm, yshift=#3 cm]

\coordinate (c00) at (0,0);
\coordinate(c01) at (1,0) ;
\coordinate (c02) at (2,0) ;
\coordinate (c10) at (0,-1) ;
\coordinate (c11) at (2,-1) ;


\node (f00) at (0,0) {\nscale{$f_{#1,0,0}$}};
\node (f01) at (1.5,0) {\nscale{$f_{#1,0,1}$}};
\node (f02) at (3,0) {\nscale{$f_{#1,0,2}$}};
\node (f10) at (0,-1.2) {\nscale{$f_{#1,1,0}$}};
\node (f11) at (2,-1.2) {\nscale{$f_{#1,1,1}$}};
\graph{
(f00) ->[,"\escale{$k$}"] (f01)->[,"\escale{$k$}"] (f02);
(f10) ->[,"\escale{$k$}"] (f11);
(f00) ->[,swap, "\escale{$v_#1$}"] (f10);
(f11) ->[,swap, "\escale{$X_#1$}"] (f02);
};
\end{scope}
}
}
\newcommand{\generator}[4]{

\ifstrequal{#4}{0}{
\begin{scope}[nodes={set=import nodes}, xshift = #2cm, yshift= #3cm, ]
\coordinate (c00) at (0,0);
\coordinate(c01) at (1,0) ;
\coordinate(c02) at (2,0) ;
\coordinate(c03) at (3,0) ;
%
\node (g00) at (0,0) {\nscale{$g_{#1,0}$}};
\node (g01) at (1.3,0) {\nscale{$g_{#1,1}$}};
\node (g02) at (2.6,0) {\nscale{$g_{#1,2}$}};
\node (g03) at (3.9,0) {\nscale{$g_{#1,2}$}};
\graph{
(g00) ->[,"\escale{$k$}"] (g01)->[,"\escale{$k$}"] (g02)->[,"\escale{$X_#1$}"] (g03);
};
\end{scope}
}{
\begin{scope}[nodes={set=import nodes}, xshift = #2cm, yshift= #3cm, ]
\coordinate (c00) at (0,0);
\coordinate(c01) at (1.2,0) ;
\coordinate(c02) at (2.4,0) ;
\coordinate(c03) at (3.6,0) ;
%
\node (g00) at (0,0) {\nscale{$g_{#1,0}$}};
\node (g01) at (1.3,0) {\nscale{$g_{#1,1}$}};
\node (g02) at (2.6,0) {\nscale{$g_{#1,2}$}};
\node (g03) at (3.9,0) {\nscale{$g_{#1,2}$}};
\graph{
(g00) ->[,"\escale{$k$}"] (g01)->[,"\escale{$k$}"] (g02)->[,"\escale{$X_#1$}"] (g03);
};
\end{scope}
}
}
\newcommand{\translator}[7]{

\ifstrequal{#7}{1}{
\begin{scope}[nodes={set=import nodes}, xshift=#5cm+0.5 cm, yshift=#6 cm]

\coordinate (c0) at (0,0);
\coordinate (c1) at (1,0) ;
\coordinate (c2) at (2,0) ;
\coordinate (c3) at (3,0) ;
\coordinate (c4) at (4,0) ;
\coordinate (c5) at (5,0) ;
\coordinate(c6) at (6,0) ;
\coordinate (c7) at (7,0) ;
\coordinate (c8) at (8,0) ;


\node (t0) at (0,0) {\nscale{$t_{#1,0}$}};
\node (t1) at (1.5,0) {\nscale{$t_{#1,1}$}};
\node (t2) at (3,0) {\nscale{$t_{#1,2}$}};
\node (t3) at (4.5,0) {\nscale{$t_{#1,3}$}};
\node (t4) at (6,0) {\nscale{$t_{#1,4}$}};
\node (t5) at (7.5,0) {\nscale{$t_{#1,5}$}};
\node (t6) at (9,0) {\nscale{$t_{#1,6}$}};
\node (t7) at (10.5,0) {\nscale{$t_{#1,7}$}};
\node (t8) at (12,0) {\nscale{$t_{#1,8}$}};
\graph{
(t0) ->[,"\escale{$k$}"] (t1)->[,"\escale{$k$}"] (t2)->[,"\escale{$X_#2$}"] (t3)->[,"\escale{$X_#3$}"] (t4)->[,"\escale{$X_#4$}"] (t5)->[,"\escale{$z$}"] (t6)->[,"\escale{$k$}"] (t7)->[,"\escale{$k$}"] (t8);
};
\end{scope}
}{
\begin{scope}[nodes={set=import nodes}, xshift=#5cm+0.5cm, yshift=#6 cm]

\coordinate (c0) at (0,0);
\coordinate (c1) at (1,0) ;
\coordinate (c2) at (2,0) ;
\coordinate (c3) at (3,0) ;
\coordinate (c4) at (4,0) ;
\coordinate (c5) at (5,0) ;
\coordinate(c6) at (6,0) ;
\coordinate (c7) at (7,0) ;
\coordinate (c8) at (8,0) ;


\node (t0) at (0,0) {\nscale{$t_{#1,0}$}};
\node (t1) at (1.5,0) {\nscale{$t_{#1,1}$}};
\node (t2) at (3,0) {\nscale{$t_{#1,2}$}};
\node (t3) at (4.5,0) {\nscale{$t_{#1,3}$}};
\node (t4) at (6,0) {\nscale{$t_{#1,4}$}};
\node (t5) at (7.5,0) {\nscale{$t_{#1,5}$}};
\node (t6) at (9,0) {\nscale{$t_{#1,6}$}};
\node (t7) at (10.5,0) {\nscale{$t_{#1,7}$}};
\node (t8) at (12,0) {\nscale{$t_{#1,8}$}};
\graph{
(t0) ->[,"\escale{$k$}"] (t1)->[,"\escale{$k$}"] (t2)->[,"\escale{$X_#2$}"] (t3)->[,"\escale{$X_#3$}"] (t4)->[,"\escale{$X_#4$}"] (t5)->[,"\escale{$z$}"] (t6)->[,"\escale{$k$}"] (t7)->[,"\escale{$k$}"] (t8);
};
\end{scope}

}
}
\begin{figure}[t!]
\vspace*{1mm}
\centering
\begin{tikzpicture}[scale=0.95]
\begin{scope}
\coordinate (c00) at (0,0);
\coordinate(c01) at (1,0) ;
\coordinate (c02) at (2,0) ;
\coordinate (c10) at (0,-1) ;
\coordinate (c11) at (2,-1) ;

\node (h00) at (0,0) {\nscale{$h_{3,0,0}$}};
\node (h01) at (1.5,0) {\nscale{$h_{3,0,1}$}};
\node (h02) at (3,0) {\nscale{$h_{3,0,2}$}};
\node (h10) at (0,-1) {\nscale{$h_{3,1,0}$}};
\node (h11) at (2,-1) {\nscale{$h_{3,1,1}$}};
\graph{
(h00) ->[,"\escale{$k$}"] (h01)->[,"\escale{$k$}"] (h02);
(h10) ->[,swap, "\escale{$k$}"] (h11);
(h00) ->[,"\escale{$u$}", swap](h10);
(h11) ->[,"\escale{$z$}", swap](h02);
};
\end{scope}
\node (q_0) at (0,1) {\nscale{$q_0$}};
\draw[->, ] (q_0)--(h00)node [pos=0.4, left ] {\escale{$y_0$}};
\freezer{0}{4}{0}{1}
\node (q_1) at (4,1) {\nscale{$q_1$}};
\draw[->, ] (q_1)--(f00)node [pos=0.4, left ] {\escale{$y_1$}};
\draw[->, ] (q_0)--(q_1)node [pos=0.5,above] {\escale{$w_1$}};
\freezer{1}{8}{0}{0}
\node (q_2) at (8,1) {\nscale{$q_2$}};
\draw[->, ] (q_2)--(f00)node [pos=0.4, left ] {\escale{$y_2$}};
\draw[->, ] (q_1)--(q_2)node [pos=0.5,above] {\escale{$w_2$}};
\freezer{2}{12}{0}{0}
\node (q_3) at (12,1) {\nscale{$q_3$}};
\draw[->, ] (q_3)--(f00)node [pos=0.4, left ] {\escale{$y_3$}};
\draw[->, ] (q_2)--(q_3)node [pos=0.5, above] {\escale{$w_3$}};
\freezer{3}{12}{-3.5}{0}
\coordinate (corner1) at (15.5,1) {};
\coordinate(corner2) at (15.5,-2) {};
\node (q_4) at (12,-2) {\nscale{$q_4$}};
\draw[->, ]  (q_3)--(corner1)--(corner2)->(q_4) node [pos=0.5, above] {\escale{$w_4$}};
\draw[->, ] (q_4)--(f00)node [pos=0.4, left] {\escale{$y_4$}};
\freezer{4}{8}{-3.5}{1}
\node (q_5) at (8,-2) {\nscale{$q_5$}};
\draw[->, ]  (q_5)->(f00)node [pos=0.4, left] {\escale{$y_5$}};
\draw[->, ]  (q_4)->(q_5) node [pos=0.5, above] {\escale{$w_5$}};
\freezer{5}{4}{-3.5}{0}
\node (q_6) at (4,-2) {\nscale{$q_6$}};
\draw[->, ]  (q_6)->(f00)node [pos=0.4, left] {\escale{$y_6$}};
\draw[->, ]  (q_5)->(q_6) node [pos=0.5, above] {\escale{$w_6$}};
\translator{0}{0}{1}{2}{1}{-6}{1}
\coordinate (corner3) at (0,-2);
\node (q_7) at (0,-6) {\nscale{$q_7$}};
\draw[->, ] (q_6)--(corner3)->(q_7)node [pos=0.5, left] {\escale{$w_7$}};
\draw[->, ] (q_7)--(t0)node [pos=0.4, above] {\escale{$y_7$}};
\translator{1}{0}{2}{3}{1}{-7.2}{1}
\node (q_8) at (0,-7.2) {\nscale{$q_8$}};
\draw[->, ]  (q_7)->(q_8)node [pos=0.5, left] {\escale{$w_8$}};
\draw[->, ]  (q_8)->(t0)node [pos=0.4, above] {\escale{$y_8$}};
\translator{2}{1}{0}{3}{1}{-8.4}{1}
\node (q_9) at (0,-8.4) {\nscale{$q_9$}};
\draw[->, ]  (q_8)->(q_9)node [pos=0.5, left] {\escale{$w_9$}};
\draw[->, ]  (q_9)->(t0)node [pos=0.4, above] {\escale{$y_9$}};
\translator{3}{2}{4}{5}{1}{-9.6}{2}
\node (q_10) at (0,-9.6) {\nscale{$q_{10}$}};
\draw[->, ]  (q_9)->(q_10)node [pos=0.5, left] {\escale{$w_{10}$}};
\draw[->, ]  (q_10)->(t0)node [pos=0.4, above] {\escale{$y_{10}$}};
\translator{4}{1}{4}{5}{1}{-10.8}{2}
\node (q_11) at (0,-10.8) {\nscale{$q_{11}$}};
\draw[->, ]  (q_10)->(q_11)node [pos=0.5, left] {\escale{$w_{11}$}};
\draw[->, ]  (q_11)->(t0)node [pos=0.4, above] {\escale{$y_{11}$}};
\translator{5}{3}{4}{5}{1}{-12}{2}
\node (q_12) at (0,-12) {\nscale{$q_{12}$}};
\draw[->, ]  (q_11)->(q_12)node [pos=0.5, left] {\escale{$w_{12}$}};
\draw[->, ]  (q_12)->(t0)node [pos=0.4, above] {\escale{$y_{12}$}};
\generator{0}{0}{-14}{1}
\node (q_13) at (0,-13) {\nscale{$q_{13}$}};
\draw[->, ]  (q_13)->(g00)node [pos=0.4, left] {\escale{$y_{13}$}};
\generator{1}{4.75}{-14}{0}
\node (q_14) at (4.75,-13) {\nscale{$q_{14}$}};
\draw[->, ]  (q_14)->(g00)node [pos=0.4, left] {\escale{$y_{14}$}};
\generator{2}{9.5}{-14}{0}
\node (q_15) at (9.5,-13) {\nscale{$q_{15}$}};
\draw[->, ]  (q_15)->(g00)node [pos=0.4, left] {\escale{$y_{15}$}};
\draw[->, ]  (q_12)->(q_13)node [pos=0.5, left] {\escale{$w_{13}$}}->(q_14)node [pos=0.5, above] {\escale{$w_{14}$}}->(q_15)node [pos=0.5, above] {\escale{$w_{15}$}};
\coordinate (corner4) at(15,-13);
\coordinate (corner5) at(15,-15);
\generator{3}{9.5}{-16}{0}
\node (q_16) at (9.5,-15) {\nscale{$q_{16}$}};
\draw[->, ]  (q_16)->(g00)node [pos=0.4, left] {\escale{$y_{16}$}};
\generator{4}{4.75}{-16}{1}
\node (q_17) at (4.75,-15) {\nscale{$q_{17}$}};
\draw[->, ]  (q_17)->(g00)node [pos=0.4, left] {\escale{$y_{17}$}};
\generator{5}{0}{-16}{0}
\node (q_18) at (0,-15) {\nscale{$q_{18}$}};
\draw[->, ]  (q_18)->(g00)node [pos=0.4, left] {\escale{$y_{18}$}};
\draw[->, ]  (q_15)--(corner4)--(corner5)->(q_16) node [pos=0.5, above] {\escale{$w_{16}$}}->(q_17) node [pos=0.5, above] {\escale{$w_{17}$}}->(q_18) node [pos=0.5, above] {\escale{$w_{18}$}};
\end{tikzpicture}\vspace{2mm}
\caption{For all $\tau\in \{\tau_{\mathbb{Z}PPT}^2,\tau_{\mathbb{Z}PT}^2\}$, the joining $J(U_\tau)$ where $\varphi$ corresponds to Example~\ref{ex:varphi}.
}
\label{fig:example} \vspace*{-3mm}
\end{figure}


Conversely, a one-in-three model of $\varphi$ implies the $\tau$-ESSP and the $\tau$-SSP for $U_\tau$:
\begin{lemma}\label{lem:tau_zppt_tau_zpt_model_implies_solvability}
If $\varphi$ has a one-in-three model, then $U_\tau$ has the $\tau$-ESSP and the $\tau$-SSP.
\end{lemma}
\begin{proof}
Let $M$ be a one-in-three model of $\varphi$, and let $I=\{h_{3,0,0}, t_{j,0}, f_{j,0,0},g_{j,0}\mid 0\leq j\leq m-1\}$ be the set of the initial states of the gadgets of $U_\tau$.

We start with the solvability of $k$.
The following $\tau$-region $R=(sup, sig)$ solves $\alpha=(k, h_{3,1,b-1})$ and thus $k$ completely in $H_3$:
for all $s\in I$, $sup(s)=0$;
for all $e\in E(U_\tau)$, if $e=k$, then $sig(e)=(0,1)$;
if $e\in \{z\}\cup (V(\varphi)\setminus M)$, then $sig(e)=0$;
for all $j\in \{0,\dots, m-1\}$, if $e=v_j$ and  $X_j\in M$, then $sig(e)=0$;
for all $j\in \{0,\dots, m-1\}$, if $e=v_j$ and  $X_j\not\in M$, then $sig(e)=1$;
otherwise holds $e \in M\cup\{u\}$, and we define $sig(e)=1$.

Notice that this region solves also a lot SSA of $U_\tau$.
In particular, if $q_0\edge{k}\dots\edge{k}q_b$, then this region solves $(s,s')$ for all $s\not=s' \in \{1_0,\dots, q_b\}$.

The following region $R=(sup, sig)$ solves $(k,s)$ for all remaining relevant states of $U_\tau$:
for all $s\in I$, $sup(s)=0$;
for all $e\in E(U_\tau)$, if $e=k$, then $sig(e)=(0,1)$;
if $e\in \{z\}\cup \{v_0,\dots, v_{m-1}\}$, then $sig(e)=1$;
otherwise, $sig(e)=0$.

We proceed with the solvability of $z$.
Let $i\in \{0,\dots, m-1\}$ be arbitrary but fixed.
Let $j,\ell\in \{0,\dots, m-1\}\setminus\{i\}$ such that $j\not=\ell$ and $X_{i_2}\in E(T_j)$ and $X_{i_2}\in E(T_\ell)$.
The following region solves $(z,s)$ for all $s\in \{h_{3,0,0}\}\cup S(T_i)$:
for all $s\in \{h_{3,0,0},t_{i,0}, t_{j,0}, t_{\ell,0}\}$, $sup(s)=b$;
for all $s\in \{f_{0,0,0}, g_{j,0}\mid j\in \{0,\dots, m-1\}\}$, $sup(s)=1$;
for all $s\in \{t_{j,0}\mid j\in \{0,\dots, m-1\}\}\setminus\{t_{i,0}, t_{j,0}, t_{\ell,0}\}$, $sup(s)=0$;
for all $e\in E(U_\tau)$, if $e=z$, then $sig(z)=(0,b)$;
if $e=X_{i_2}$, then $sig(e)=1$;
if $n\in \{0,\dots, m-1\}$ and $e=v_n$ and $X_{i_2}\in E(F_n)$, then $sig(e)=b$;
if $e=u$, then $sig(e)=1$;
otherwise $sig(e)=0$.
By the arbitrariness of $i$, this proves also the $\tau$-solvability of $(z,s)$ for all relevant $s\in \bigcup_{j=0}^{m-1}S(T_j)$.

Notice that if $s\in \{h_{3,0,0},\dots, h_{3,0,b}\}$ and $s'\in \{h_{3,1,0},\dots, h_{3,1,b-1}\}$ or if $s\in \{f_{i,0,0},\dots, h_{i,0,b}\}$ and $s'\in \{s_{i,1,0},\dots, f_{i,1,b-1}\}$, then this region also solves $(s,s')$.
Thus, altogether, we already have proven the solvability of all states of $H_3,F_0,\dots, F_{m-1}$.

The following region $R=(sup, sig)$ solves $(z,s)$ for all relevant $s\in S(H_3)\setminus\{h_{3,0,0}\}$:
for all $s\in I\setminus\{t_{i,0}\mid i\in \{0,\dots, m-1\}\}$, $sup(s)=0$;
for all $s\in \{t_{i,0}\mid i\in \{0,\dots, m-1\}$, $sup(s)=1$;
for all $e\in E(U_\tau)$, if $e=z$, then $sig(e)=(0,b)$;
if $e=k$, then $sig(k)=1$;
if $e=u$, then $u=2$;
otherwise $sig(e)=0$.

The following region $R=(sup, sig)$ solves $(z,s)$ for all remaining relevant states:
for all $s\in \{h_{3,0,0}\}\cup\{ f_{j,0,0}, g_{j,0}\mid j\in \{0,\dots, m-1\}\}$, $sup(s)=b$;
for all $s\in \{t_{i,0}\mid i\in \{0,\dots, m-1\}$, $sup(s)=0$;
for all $e\in E(U_\tau)$, if $e=z$, then $sig(e)=(0,b)$;
if $e=u$, then $sig(e)=1$;
otherwise, $sig(e)=0$.

We proceed by arguing for the solvability of $u$.
The following region $R=(sup, sig)$ solves $(u,s)$ for all $s\in \{h_{3,0,1}, \dots, h_{3,0,b}\}$:
for all $s\in I$, $sup(s)=0$;
for all $e\in E(U_\tau)$, if $e=u$, then $sig(e)=(0,b)$;
if $e=z$, then $sig(e)=2$;
of $e=k$, then $sig(e)=1$;
otherwise $sig(e)=0$.

If $b >2$, then the following region $R=(sup, sig)$ solves $(u,s)$ for relevant states $s\in S(U_\tau)\setminus \{h_{3,0,1}, \dots, h_{3,0,b}\}$:
for all $s\in I$, if $s=h_{3,0,0}$, then $sup(s)=0$; otherwise, $sup(s)=1$;
for all $e\in E(U_\tau)$, if $e=u$, then $sig(e)=(0,b)$;
if $e=z$, then $sig(e)=1$;
otherwise $sig(e)=0$.
If $b=2$, then we additionally need a slightly modified region that maps $sup(s)=0$ for all $s\in \{t_{j,0}\mid j\in \{0,\dots, m-1\}\}$.
This proves the solvability of $u$.

We proceed with the solvability of the events $v_0,\dots, v_{m-1}$.
Let $i\in \{0,\dots, m-1\}$ be arbitrary but fixed.
The following region $R=(sup, sig)$ solves $(u_i,s)$ for all $s\in \{f_{i,0,1},\dots, f_{i,0,b}\}$:
for all $s\in I$, $sup(s)=0$;
for all $e\in E(U_\tau)$, if $e=u_i$, then $sig(e)=(0,b)$;
if $e=X_i$, then $sig(e)=2$;
if $e=k$, then $sig(e)=1$;
otherwise $sig(e)=0$.

If $b> 2$, then the following region $R=(sup, sig)$ solves $(v_i, s)$ for all remaining relevant states $S(U_\tau)\setminus  \{f_{i,0,1},\dots, f_{i,0,b}\}$:
$sup(f_{i,0,0})=0$;
for all $s\in I\setminus\{f_{i,0,0}\}$, $sup(s)=1$;
for all $e\in E(U_\tau)$, if $e=u_i$, then $sig(e)=(0,b)$;
if $e=X_i$, then $sig(e)=1$;
otherwise $sig(e)=0$.
If $b=2$, then we additionally need a slightly modified region that maps $sup(s)=0$ for all $s\in \{g_{j,0},t_{j,0}\mid j\in \{0,\dots, m-1\}\}$.
This proves the solvability of $v_i$.
Since $i$ was arbitrary, this proves the solvability of all $v_0,\dots, v_{m-1}$.

It is easy to see, that the variable events $X_0,\dots, X_{m-1}$ are solvable.
Thus,  for the sake of simplicity, we refrain from the explicit representation of the corresponding regions.
Moreover, one easily verifies that the remaining regions that complete the $\tau$-ESSP of $U_\tau$ also solve the remaining SSA of $U_\tau$.
Altogether, we have finally proven that if $M$ has a one in three model, then $U_\tau$ has the $\tau$-ESSP and the $\tau$-SSP.
\end{proof}

\section{Polynomial time results}\label{sec:poly_results}%

The following theorem states the main result of this section:
\begin{theorem}\label{the:tractability}
\begin{enumerate}
\item
\textsc{$\tau_{R\mathbb{Z}PT}^b$-ESSP} can be solved in time polynomial in the size of input $A$.
\item\label{the:tau_zppt_tau_zpt_tau_rzpt_ssp}
If $\tau\in \{\tau_{\mathbb{Z}PT}^b, \tau_{\mathbb{Z}PPT}^b, \tau_{R\mathbb{Z}PT}^b\}$, then 
\textsc{$\tau$-SSP} can be solved in time polynomial in the size of input $A$.
\end{enumerate}
\end{theorem}
The contribution of Theorem~\ref{the:tractability} is threefold.
Firstly, \textsc{$\tau$-ESSP} and \textsc{$\tau$-Solvability} are NP-complete for all $\tau\in \{\tau_{\mathbb{Z}PT}^b,\tau_{\mathbb{Z}PPT}^b\}$ by Theorem~\ref{the:hardness_results}.
However, Theorem~\ref{the:tractability}.\ref{the:tau_zppt_tau_zpt_tau_rzpt_ssp} states that \textsc{$\tau$-SSP} is solvable in polynomial-time for these types.
Hence, to the best of our knowledge, Theorem~\ref{the:tractability} discovers the first Petri net types for which \textsc{$\tau$-SSP} and \textsc{$\tau$-ESSP} as well as \textsc{$\tau$-SSP} and \textsc{$\tau$-Solvability} provably have a different computational complexity.

\medskip
Secondly, in \cite{DBLP:conf/stacs/Schmitt96}, Schmitt extended the type $\tau_{PPT}^1$ by the additive group of integers modulo $2$, which leads to the tractable (super-) type to $\tau_{\mathbb{Z}PPT}^1$.
Moreover, in~\cite{DBLP:conf/tamc/TredupR19},  we argued that Schmitts approach transferred to $\tau_{PT}^1$ yields the tractable type $\tau_{\mathbb{Z}PT}^1$.
However, by Theorem~\ref{the:hardness_results}, lifting Schmitts technique to $\tau_{PPT}^b$ and $\tau_{PT}^b$ does not lead to superclasses with a tractable synthesis problem for all $2\leq b\in  \mathbb{N}$.
Hence, Theorem~\ref{the:tractability} proposes the first tractable type of $b$-bounded Petri nets, where $b\geq 2$, so far.
Finally, Theorem~\ref{the:tractability} gives us insight into which of the $\tau$-\emph{net properties}, where $\tau\in \{\tau_{PT}^b,\tau_{PPT}^b\}$, cause the hardness of \textsc{$\tau$-Synthesis} and the corresponding separation problems.
In particular, flow arc relations (events in $\tau$) between places and transitions in a $\tau$-net define conditions when a transition is able to fire.
For example, if $N$ is a $\tau$-net with transition $t$ and place $p$ such that $f(p,t)=(1,0)$ then the firing of $t$ in a marking $M$ requires $M(p)\geq 1$.
By Theorem~\ref{the:tractability}, the hardness of finding a $\tau$-net $N$ for $A$ originates from the potential possibility of $\tau$-nets to satisfy such conditions by multiple markings $M(p)\in \{1,\dots,b\}$.
In fact, the definition of  $\tau_{R\mathbb{Z}PT}^b$ implies that $f(p,t)=(m,n)$ requires $M(p)=m$ for the firing of $t$ and prohibits the possibility of multiple choices.
By Theorem~\ref{the:tractability}, this makes $\tau_{R\mathbb{Z}PT}^b$-synthesis tractable.

While the question of whether there are superclasses of $\tau_{PT}^b,\tau_{PPT}^b$, $b\geq 2$, for which synthesis is doable in polynomial time remains unanswered, the following lemma shows that the type $\tau_{R\mathbb{Z}PT}^b$ yields at least a tractable superclasses of Schmitt's type $\tau_{R\mathbb{Z}PT}^b$ \cite{DBLP:conf/stacs/Schmitt96}.
More generally, if $b < b'$ then the class of $\tau_{R\mathbb{Z}PT}^b$-nets is strictly more comprehensive than the class of $\tau_{R\mathbb{Z}PT}^{b'}$-nets:

\begin{lemma}\label{lem:contribution_of_sharp_nets}
If $b < b'\in \mathbb{N}^+$ and if $\mathcal{T}$ is the set of $\tau_{R\mathbb{Z}PT}^b$ -solvable TS and $\mathcal{T'}$ the set of $\tau_{R\mathbb{Z}PT}^{b'}$-solvable TS then $\mathcal{T }\subset \mathcal{T'}$.
\end{lemma}
\begin{proof}
We present a TS $A$ that is $\tau_{R\mathbb{Z}PT}^{b'}$-solvable but not $\tau_{R\mathbb{Z}PT}^{b}$-solvable:
Let $A=(\{s_0,\dots , s_{b'}\},\{a\},\delta, s_0)$ be the TS with transition function $\delta (s_i, a)=s_{i+1}$ for $i\in \{0,\dots, b'-1\}$ and $\delta (s_{b'},a)=s_0$.
By other words, $A$ is a directed labeled cycle $s_0\edge{a}\dots \edge{a}s_{b'}\edge{a}s_0$ where every transition is labeled by $a$.
Notice, that $A$ has no ESSA.
Hence, it has the $\tau$-ESSP for every type of nets $\tau$.
Consequently, $A$ is $\tau$-solvable if and only if it has the $\tau$-SSP.

Assume, for a contradiction, that $A$ is $\tau_{R\mathbb{Z}PT}^b$-solvable.
By $b <b'$, $A$ provides the SSA $(s_0, s_{b+1})$ and the $\tau_{R\mathbb{Z}PT}^b$-solvability of $A$ implies that there is a $\tau_{R\mathbb{Z}PT}^b$-region $(sup, sig)$ that solves it.
If $sig(a)=(m,n)$ then $sup(s_1)=sup(s_0)-m+n\not=sup(s_0)$ and, by definition of $\tau_{R\mathbb{Z}PT}^b$, $\neg sup(s_1)\ledge{(m,n)}$.
This is a contradiction to $s_1\edge{a}$.
Hence, $sig(a)\in \{1,\dots, b\}$.
By induction, $sup(s_{b+1})=sup(s_0) + (b+1)\cdot sig(a) = sup(s_0) \text{ mod } (b+1)$ implying $sup(s_{b+1})=sup(s_0)$.
Thus, $(sup, sig)$ does not solve $(s_0, s_{b+1})$, which proves that $A$ is not $\tau_{R\mathbb{Z}PT}^{b}$-solvable.

On the contrary, it is easy to see that the $\tau_{R\mathbb{Z}PT}^{b'}$-region $(sup, sig)$, which is defined by $sup(s_0)=0$, $sig(a)=1$ and $sup(s_{i+1})=sup(s_i)+ sig(a)$ for $i\in \{0,\dots, b'-1\}$, solves every SSA of $A$.
Hence, $A$ is $\tau_{R\mathbb{Z}PT}^{b'}$-solvable.
\end{proof}

\subsection{Abstract regions and fundamental cycles}%

In the remainder of this paper, unless explicitly stated otherwise, we assume that $A=(S,E,\delta,\iota)$ is an arbitrary but fixed (non-trivial) TS with at least two states and event set $E=\{ e_1,\dots, e_n \}$.
Recall that $\tau\in \{\tau_{\mathbb{Z}PT}^b, \tau_{\mathbb{Z}PPT}^b, \tau_{R\mathbb{Z}PT}^b\}$ and $b\in \mathbb{N}^+$ are also arbitrary but fixed.

The proof of Theorem~\ref{the:tractability} bases on a generalization of the approach used in~\cite{DBLP:conf/stacs/Schmitt96} that reduces the solvability of ESSA and SSA to the solvability of systems of linear equations modulo $b+1$.
It exploits that the solvability of such systems is decidable in polynomial time:

\begin{lemma}[\cite{DBLP:journals/iandc/GoldmannR02}]\label{lem:complexity_solving_linear_system}
Let $M \in \mathbb{Z}^{k \times n}_{b+1}$ and $c\in \mathbb{Z}^k_{b+1}$.
There is an algorithm that decides in time  $\mathcal{O}(nk\cdot max\{n,k\})$ whether there is an element $x\in \mathbb{Z}^n_{b+1}$ such that $Mx=c$.
\end{lemma}

Essentially, our generalization composes for every ESSA and every SSA $\alpha=(x,y)$ of the TS $A$ a system of equations modulo $b+1$ that has a solution if and only if $\alpha$ is $\tau$-solvable.
Hence, the TS $A$ has the $\tau$-ESSP, respectively the $\tau$-SSP, if and only if every system, defined by the ESSA of $A$, respectively by the SSA of $A$, has a solution.

We proceed by deducing the notion of abstract regions.
Our starting point is the goal to obtain $\tau$-regions $(sup, sig)$ of $A$ as solutions of linear equation systems modulo $b + 1$.
By Definition~\ref{def:region} and the definition of $\tau$, $(sup, sig)$ is a $\tau$-region of $A$ if and only if for every transition $s\edge{e}s'$ it is true that
\begin{equation}\label{eq:region_definition}
sup(s')=(sup(s)-sig^-(e)+sig^+(e)+\vert sig(e) \vert) \text{ mod } (b+1)
\end{equation}
Hence, installing for every transition $s\edge{e}s'$ the corresponding Equation~\ref{eq:region_definition} yields a linear system of equations whose solutions are regions of $A$.
If $(sup, sig)$ is a solution of this system such that $sig(e)=(m,n)\in E_\tau\setminus \{0,\dots, b\}$ for $e\in E(A)$ then, by definition, for every transition $s\edge{e}s'$ it has to be true that $m \leq sup(s)$ and $sup(s')-m+n\leq b$.
Unfortunately, the conditions $m \leq sup(s)$ and $sup(s')-m+n\leq b$ can not be tested in the group $\mathbb{Z}_{b+1}$.
To cope with this obstacle, we abstract from elements $(m,n)\in E_\tau$ by restricting to regions (solutions) that identify $(m,n)$ with the unique element $x\in \{0,\dots, b\}$ such that $x= (n-m) \text{ mod } (b+1)$.
This leads to the notion of \emph{abstract} $\tau$-regions.
\begin{definition}[Abstract Region]\label{def:abstract_region}
A $\tau$-region $(sup, sig)$ of $A=(S,E,\delta,\iota)$ is called \emph{abstract} if the codomain of $sig$ is restricted to the elements of $\mathbb{Z}_{b+1}$, that is, $sig: E\longrightarrow \{0,\dots, b\}$.
If $(sup, sig)$ is an abstract region, then we call $sig$ an \emph{abstract} signature.
\end{definition}

\begin{remark}[Notation of abstract regions]\label{rem:notation_abstract_region}
For the sake of clarity, we denote abstract signatures by $abs$ instead of $sig$ and abstract regions by $(sup, abs)$ instead of $(sup, sig)$.
For convenience, we also identify $abs=(abs(e_1),\dots, abs(e_n))$.
\end{remark}

By definition, two mappings $sup:\{0,\dots, b\} \rightarrow \{0,\dots, b\}$ and $abs: E\rightarrow \{0,\dots, b\}$ define an abstract $\tau$-region if and only if for every transition $s\edge{e}s'$ of $A$ it is true that
\begin{equation}\label{eq:abstract_region_definition}
sup(s')=(sup(s) + abs(e)) \text{ mod } (b+1)
\end{equation}
Obviously, for abstract regions, the Equation~\ref{eq:region_definition} reduces to Equation~\ref{eq:abstract_region_definition}.
Installing for every transition $s\edge{e}s'$ of $A$ its corresponding Equation~\ref{eq:abstract_region_definition} yields a system modulo $b+1$ whose solutions are abstract regions.
However, such systems require to deal with $sup$ and $abs$ simultaneously, which is very inconvenient.
It is better to first obtain $abs$ independently of $sup$ and then to define $sup$ with the help of $abs$.
The following observations show how to realize this idea.

\medskip
By induction and Equation~\ref{eq:abstract_region_definition}, one immediately obtains that $(sup, abs)$ is an abstract region if and only if for every directed labeled path $p=\iota\Edge{e'_1}\dots\Edge{e'_m}s_m$ of $A$ from the initial state $\iota$ to the state $s_m$ the \emph{path equation} holds:
\begin{equation}\label{eq:path_equation}
sup(s_m) = (sup(\iota) + abs(e'_1)+ \dots  + abs(e'_m)) \text{ mod }(b+1)
\end{equation}

In order to exploit Equation~\ref{eq:path_equation}, we first introduce the following notions:

\begin{definition}[Parikh-vector]\label{def:parikh_vector}
Let $p=z_0\edge{a_1}\dots\edge{a_m}z_m$ be a path of the TS $A$ on pairwise distinct states $z_0,\dots, z_m$.
The \emph{Parikh-vector} of $p$ is the mapping $\psi_p:\{e_1,\dots, e_n\}\rightarrow \{0,\dots, b\}$ such that $\psi_p(e)=\vert \{i\in \{0,\dots, m-1\}\mid z_i\edge{e}\}\vert \text{ mod }(b+1)$ for every event $e\in \{e_1,\dots, e_n\}$, that is, $\psi_p$ assigns to $e$ the number of its occurrences on $p$ modulo $b+1$.
For convenience, we identify $\psi_p=(\psi_p(e_1),\dots, \psi_p(e_n))$.
\end{definition}

\begin{definition}[Product]\label{def:product}
If $x=(x_1,\dots, x_n)$ and $y=(y_1,\dots, y_n)$ are two elements of $\mathbb{Z}^n_{b+1}$, then we say $x\cdot y= (x_1\cdot y_1 +\dots + x_n\cdot y_n) \text{ mod } b+1 $ is the \emph{product} of $x$ and $y$.
\end{definition}

Definition~\ref{def:parikh_vector} and Definition~\ref{def:product} allow us to reformulate the path equation~\ref{eq:path_equation} as follows:
\begin{equation}\label{eq:path_equation_reformulated}
sup(s_m)=(sup(\iota) + \psi_p\cdot abs) \text{ mod } (b+1)
\end{equation}

Notice that if $p,p'$ are two different paths from $\iota$ to $s_m$, then $\psi_p\cdot abs =\psi_p\cdot abs$.
Thus, the support $sup$ is fully determined by $sup(\iota)$ and $abs$.
We obtain $sup$ explicitly by $sup(s)=(sup(\iota) + \psi_p\cdot abs) \text{ mod } (b+1)$ for all $s\in S$, where $p$ is an arbitrary but fixed path of $A$ that starts at $\iota$ and terminates at $s$.
Consequently, every abstract signature $abs$ implies $b+1$ different abstract $\tau$-regions of $A$, one for every $sup(\iota)\in \{0,\dots, b\}$.
Altogether, we have argued that the challenge of finding abstract regions of $A$ reduces to the task of finding the abstract signatures of $A$.
In the following, we introduce the notion of fundamental cycles, defined by so-called chords of a spanning tree of $A$, which enables us to find abstract signatures.

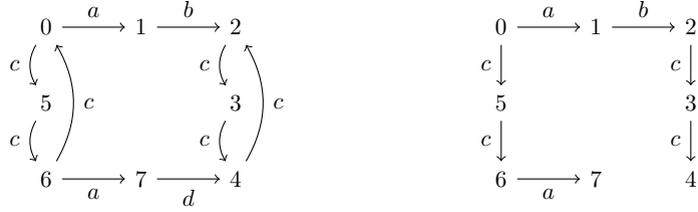
\begin{figure}[t!]
\centering
\begin{tikzpicture}
\begin{scope}
\node (0) at (0,0) {\nscale{$0$}};
\node (1) at (1.25,0) {\nscale{$1$}};
\node (2) at (2.5,0) {\nscale{$2$}};
\node (3) at (2.5,-1) {\nscale{$3$}};
\node (4) at (2.5,-2) {\nscale{$4$}};
\node (5) at (0,-1) {\nscale{$5$}};
\node (6) at (0,-2) {\nscale{$6$}};
\node (7) at (1.25,-2) {\nscale{$7$}};

\path (0) edge [->] node[pos=0.5, above] {\escale{$a$}} (1);
\path (1) edge [->] node[pos=0.5,above] {\escale{$b$}} (2);
\path (2) edge [->, bend right=30] node[pos=0.5,left] {\escale{$c$}} (3);
\path (3) edge [->, bend right=30] node[pos=0.5,left] {\escale{$c$}} (4);
\path (4) edge [->, bend right=30] node[pos=0.5,right] {\escale{$c$}} (2);

\path (0) edge [->, bend right=30] node[pos=0.5,left ] {\escale{$c$}} (5);
\path (5) edge [->, bend right=30] node[pos=0.5, left] {\escale{$c$}} (6);
\path (6) edge [->] node[pos=0.5, below] {\escale{$a$}} (7);
\path (6) edge [->, bend right=30] node[pos=0.5, right] {\escale{$c$}} (0);
\path (7) edge [->] node[pos=0.5, below] {\escale{$d$}} (4);
\end{scope}
\begin{scope}[xshift=6cm]
\node (0) at (0,0) {\nscale{$0$}};
\node (1) at (1.25,0) {\nscale{$1$}};
\node (2) at (2.5,0) {\nscale{$2$}};
\node (3) at (2.5,-1) {\nscale{$3$}};
\node (4) at (2.5,-2) {\nscale{$4$}};
\node (5) at (0,-1) {\nscale{$5$}};
\node (6) at (0,-2) {\nscale{$6$}};
\node (7) at (1.25,-2) {\nscale{$7$}};

\path (0) edge [->] node[pos=0.5, above] {\escale{$a$}} (1);
\path (1) edge [->] node[pos=0.5,above] {\escale{$b$}} (2);
\path (2) edge [->] node[pos=0.5,left] {\escale{$c$}} (3);
\path (3) edge [->] node[pos=0.5,left] {\escale{$c$}} (4);
\path (0) edge [->] node[pos=0.5,left ] {\escale{$c$}} (5);
\path (5) edge [->] node[pos=0.5, left] {\escale{$c$}} (6);
\path (6) edge [->] node[pos=0.5, below] {\escale{$a$}} (7);
\end{scope}
\end{tikzpicture}
\caption{Left: An input TS $A$.
Right: A spanning tree $A'$ of TS $A$.
The unique Parikh vectors $\psi_0,\dots \psi_7$ of $A'$ (written as rows) are given by $\psi_0= (0,0,0,0), \psi_1=(1,0,0,0), \psi_2= (1,1,0,0), \psi_3=(1,1,1,0), \psi_4=(1,1,2,0), \psi_5= (0,0,1,0)$, $\psi_6=(0,0,2,0) $ and $\psi_7=(1,0,2,0) $.
The transitions $\delta_A(7,d)=4$, $\delta_A(4,c)=2$ and $\delta_A(6,c)=0$ of $A$ define the chords of $A'$.
The corresponding fundamental cycles are given by $\psi_t=\psi_7 +(0,0,0,1) -\psi_4 = (0,2,0,1) $ and $\psi_{t'}= \psi_4 + (0,0,1,0)-\psi_2 =(0,0,0,0)$ and  $\psi_{t''}=\psi_6 + (0,0,1,0)-\psi_0 =(0,0,0,0)$.
Hence, if $abs=(x_a,x_b,x_c,x_d)$ then $\psi_t\cdot abs= 0\cdot x_a +2\cdot x_b +0\cdot x_c + x_d=2\cdot x_b + x_d$.
By $\psi_{t'}\cdot abs = \psi_{t''}\cdot abs =0$ for every map $abs$, only the equation $2\cdot x_b + x_d=0$ contributes to the basic part of every upcoming system. } \label{fig:fundamental_cycles}\vspace*{-3mm}
\end{figure}

\begin{definition}[Spanning tree, chord]\label{def:spanning_tree}
A \emph{spanning tree} $A'$ of TS $A$ is a sub-transition system $A'=(S, E', \delta_{A'}, \iota)$ of $A$ with the same set of states $S$, an event set $E'\subseteq E$ and a restricted transition function $\delta_{A'}$ such that first $\delta_{A'}(s,e)=s'$ entails $\delta_A(s,e)=s'$ and, moreover, for every $s\in S$ there is \emph{exactly} one path $p=\iota\edge{e_1} \dots \edge{e_m} s$ in $A'$.
Every transition $s\edge{e}s'$ of $A$ which is not in $A'$ is called a \emph{chord} (of $A'$).
\end{definition}

\begin{remark}[Parikh-vector of a state in the spanning tree]
For every $s\in S$, by $\psi_s$ we denote the Parikh-vector $\psi_p$ of the unique path $p=\iota\edge{e_1} \dots \edge{e_m} s$ in $A'$.
\end{remark}

Notice that the underlying undirected graph of $A'$ is a tree in the common graph-theoretical sense.
The chords of $A'$ are exactly the edges that induce a cycle in the underlying undirected graph of $A'$.
This gives rise to the following notion of fundamental cycles:

\begin{definition}[Fundamental cycle]\label{def:fundamental_cycle}
Let $t=s\edge{e}s'$ be a chord of $A'$.
The \emph{fundamental cycle} of $t$ is the mapping $\psi_t:\{e_1,\dots, e_n\}\rightarrow \{0,\dots, b\}$ that is defined as follows for all $i\in \{1,\dots, n\}$:
\[\psi_t(e_i)=
\begin{cases}
\psi_s(e_i)-\psi_{s'}(e_1) \text{ mod } b+1, & \text{if } e_i\not=e\\
\psi_s(e_i) -\psi_{s'}(e_i) +1 \text{ mod } b+1, & \text{else.}\\
\end{cases}
\]
For convenience, we identify $\psi_t=(\psi_t(e_1),\dots, \psi_t(e_n))$.
\end{definition}

By the following lemma, we can use the fundamental cycles to generate abstract signatures of $A$:

\begin{lemma}\label{lem:fundamental_cycles}
If $A'$ is a spanning tree of a TS $A$ with chords $t_1,\dots, t_k$ then $\emph{abs}\in \mathbb{Z}^n_{b+1}$ is an abstract signature of $A$ if and only if $\psi_{t_i} \cdot \emph{abs}  = 0$ for all $i\in \{1,\dots, k\}$.
Two different spanning trees $A'$ and $A''$ provide equivalent systems of equations.
\end{lemma}

\begin{proof}
We start with proving the first statement.
\textit{If}:
Let $\emph{abs}\in \mathbb{Z}^n_{b+1}$ such that $\psi_{t_i} \cdot abs  = 0$ for all $i\in \{1,\dots, k\}$ and $sup(\iota)\in \{0,\dots, b\}$.
Let $sup(\iota)\in \{0,\dots, b\}$ be arbitrary but fixed and, for all $s\in S$, let $sup(s)=sup(\iota)+\psi_s\cdot abs$.
We show that $(sup, sig)$ is an abstract region of $A$, that is, for all edges $t=s\edge{a}s'$ of $A$ holds $sup(s')=sup(s) +abs(a) \text{ mod } b+1$:
By definition, we have $sup(s)=sup(\iota)+\psi_s\cdot abs$ and $sup(s')=sup(\iota)+\psi_{s'}\cdot abs$.
If $t$ is not a chord, then $\psi_{s'}(a)=\psi_s(a)+1 \text{ mod }b+1$ and $\psi_{s'}(e)=\psi_s(e)$ for all $e\in \{e_1,\dots, e_n\}\setminus\{a\}$.
This implies $sup(s')=sup(\iota)+\psi_s\cdot abs + abs(a) \text{ mod }b+1$ and thus $sup(s')=sup(s) +abs(a) \text{ mod } b+1$.

\medskip
Otherwise, if $t$ is a chord of $A'$, then it holds $\psi_t(a)=\psi_s(a) -\psi_{s'}(a)+1$ and the following implications (considered modulo $b+1$) are true:
\begin{align*}
 0 & = \psi_t \cdot abs  &\Longleftrightarrow  \\
 0 &= \sum_{i=1}^n((\psi_s(e_i)-\psi_{s'}(e_i))\cdot abs(e_i) +abs(a)&\Longleftrightarrow  \\
  0 &= \sum_{i=1}^n \psi_s(e_i)\cdot abs(e_i)- \sum_{i=1}^n \psi_{s'}(e_i)\cdot abs(e_i) +abs(a)&\Longleftrightarrow  \\
  \psi_{s'}\cdot abs  &= \psi_s\cdot abs+abs(a)&\Longleftrightarrow  \\
   sup(\iota)+ \psi_{s'}\cdot abs  &=  sup(\iota) + \psi_s\cdot abs+abs(a)&\Longleftrightarrow  \\
 sup(s')  & = sup(s) + abs(a) &
 \end{align*}

Hence, $abs$ is an abstract signature of $A$ and the proof shows how to get a corresponding abstract region $(sup,abs)$ of $A$.

\textit{Only-if}:
If $abs$ is an abstract region of $A$ then we have $sup(s') = sup(s) + abs(e)$ for every transition in $A$.
Hence, if $t=s\edge{e}s'$ is a chord of a spanning tree $A'$ of $A$ then working backwards the equivalent equalities above proves $\psi_t\cdot abs =0$.

The second statement is implied by the first:
If $A'$, $A''$ are two spanning trees of $A$ with fundamental cycles $\psi^{A'}_{t_1},\dots, \psi^{A'}_{t_k}$ and $\psi^{A''}_{t'_1},\dots, \psi^{A''}_{t'_k}$, respectively, then we have for $abs\in \mathbb{Z}^n_{b+1}$ that $\psi^{A'}_{t_i}\cdot abs =0 ,i\in \{1,\dots, k\}$ if and only if $abs$ is an abstract signature of $A$ if and only if $\psi^{A''}_{t'_i}\cdot abs =0 ,i\in \{1,\dots, k\}$.
\end{proof}

In the following, justified by Lemma~\ref{lem:fundamental_cycles}, we assume $A'$ to be a fixed spanning tree of $A$ with chords $t_1,\dots, t_k$.
By $M_{A'}$ we denote the system of equations that consists of $\psi_{t_i} \cdot abs  = 0$ for all $i\in \{1,\dots, k\}$.
A spanning tree of $A$ is computable in polynomial time:
As $\delta_A$ is a function, $A$ has at most $ \vert E\vert\vert S\vert^2$ edges and $A'$ contains $\vert S\vert -1$ edges.
Thus, by $2 \leq \vert S\vert$, $A'$ has at most $\vert E\vert\vert S\vert^2 -1$ chords.
Consequently, a spanning tree $A'$ of $A$ is computable in time $\mathcal{O}(\vert E\vert\vert S \vert^3)$~\cite{DBLP:journals/networks/Tarjan77}.

To get polynomial time solvable systems of equations, we have restricted ourselves to equations like Equation~\ref{eq:path_equation} or its reformulated version Equation~\ref{eq:path_equation_reformulated}.
This restriction results in the challenge to compute abstract signatures of $A$.
By Lemma~\ref{lem:fundamental_cycles}, abstract signatures of $A$ are solutions of $M_{A'}$.
We get an (abstract) $\tau$-region $(sup, abs)$ of $A$ from $sup(\iota)$ and $abs$ by defining $sup(\iota)$ and $sup(s)=sup(\iota)+\psi_s \cdot abs $ for all $s\in S$.
However, if $(s, s')$ is an SSA of $A$ then $sup(s)\not=sup(s')$ is not implied.
Moreover, by definition, to solve an ESSA $(e,s)$, we need (concrete) $\tau$-regions $(sup, sig)$ such that $sig:E \longrightarrow E_\tau$.
The next section shows how to extend $M_{A'}$ to get such solving $\tau$-regions.

\subsection{The Proof of Theorem~\ref{the:tractability}}%

This section shows how to extend $M_{A'}$ for a given (E)SSA $\alpha$ to get a system $M_\alpha$, whose solution yields a region solving $\alpha$ if there is one.
But first we need the following lemma that tells us how to obtain abstract regions from (concrete) regions:

\begin{lemma}\label{lem:concrete_to_abstract}
If $(sup, sig)$ is a $\tau$-region of a TS $A=(S,E,\delta,\iota)$ then we obtain a corresponding abstract $\tau$-region $(sup, abs)$ by defining $abs$ for $e\in E$ as follows:
If $sig(e)=(m,n)$ then $abs(e)=-m+n \text{ mod } (b+1)$ and, otherwise, if $sig(e)\in \{0,\dots, b\}$ then $abs(e)=sig(e)$.
\end{lemma}
\begin{proof}
We have to show that $s\edge{e}s'$ in $A$ entails $sup(s)\ledge{abs(e)}sup(s')$ in $\tau$.
If $abs(e)=sig(e)\in \{0,\dots, b\}$ this is true as $(sup, sig)$ is a $\tau$-region.

If $sig(e)=(m,n)$ then, by definition, we have $sup(s')=sup(s)-m+n \text{ mod } (b+1)$ implying $sup(s')-sup(s)= -m+n \text{ mod } (b+1)$.
By $abs(e)=-m+n \text{ mod } (b+1)$ and symmetry, we get $-m+n = abs(e) \text{ mod } (b+1)$ and, by transitivity, we obtain $sup(s')-sup(s) = abs(e) \text{ mod } (b+1)$ which implies $sup(s')= sup(s) + abs(e) \text{ mod } (b+1)$.
Thus $sup(s)\ledge{abs(e)}sup(s')$.
\end{proof}

If $\alpha$ is an SSA $(s,s')$ then we only need to assure that the (abstract) region $(sup, abs)$ built on a solution of $M_{A'}$ satisfies $sup(s)\not=sup(s')$.
By $sup(s)=sup(\iota) + \psi_s\cdot abs$ and $sup(s')=sup(\iota) + \psi_{s'}\cdot abs$, it is sufficient to extend $M_{A'}$ in a way that ensures $\psi_s\cdot abs\not= \psi_{s'}\cdot abs$.
The next lemma proves this claim.

\begin{lemma}\label{lem:ssp}
If $\tau\in \{\tau_{\mathbb{Z}PT}^b,\tau_{\mathbb{Z}PPT}^b,\tau_{R\mathbb{Z}PT}^b\}$ then an SSA $(s,s')$ of $A=(S,E,\delta,\iota)$ is $\tau$-solvable if and only if there is an abstract signature $abs$ of $A$ with $\psi_{s}\cdot abs \not = \psi_{s'}\cdot abs $.
\end{lemma}
\begin{proof}%
\textit{If}:
If $abs$ is an abstract signature with $\psi_{s}\cdot abs \not = \psi_{s'}\cdot abs $ then the $\tau$-region $(sup, abs)$ with $sup(\iota)=0$ and $sup(s)=\psi_{s}\cdot abs$ satisfies $sup(s)\not=sup(s')$.
\textit{Only-if}:
If $(sup, sig)$ is a $\tau$-region then we obtain a corresponding abstract $\tau$-region $(sup, abs)$ as defined in Lemma~\ref{lem:concrete_to_abstract}.
Clearly, $abs$ is an abstract signature and satisfies the path equations.
Consequently, by $sup(s_0)+\psi_{s}\cdot abs= sup(s) \not= sup(s')=sup(s_0)+\psi_{s'}\cdot abs $, we have that $\psi_{s}\cdot abs \not=\psi_{s'}\cdot abs $.
\end{proof}

The next lemma applies Lemma~\ref{lem:ssp} to get a polynomial time algorithm which decides the $\tau$-SSP if $\tau\in \{\tau_{\mathbb{Z}PT}^b,\tau_{\mathbb{Z}PPT}^b,\tau_{R\mathbb{Z}PT}^b\}$.

\begin{lemma}\label{lem:ssp_tractability}
If $\tau\in \{\tau_{\mathbb{Z}PT}^b,\tau_{\mathbb{Z}PPT}^b,\tau_{R\mathbb{Z}PT}^b\}$ then to decide whether a TS $A=(S,E,\delta,\iota)$ has the $\tau$-SSP is doable in time $\mathcal{O}(\vert E\vert^3 \cdot \vert S \vert^6\cdot)$.
\end{lemma}
\begin{proof}
If $\alpha=(s,s')$ is an SSA of $A$ then the (basic) part $M_{A'}$ of $M_\alpha $ consists of at most  $\vert E\vert\cdot \vert S\vert^2 - 1$ equations for the fundamental cycles.
To satisfy $\psi_{s}\cdot abs \not = \psi_{s'}\cdot abs $, we add the equation $(\psi_{s}-\psi_{s'})\cdot abs = q$, where initially $q=1$, and get (the first possible) $M_\alpha$.
A solution of $M_\alpha$ provides an abstract region satisfying $\psi_{s} \not= \psi_{s'}$.
By Lemma~\ref{lem:ssp}, this proves the solvability of $\alpha$.
If $M_\alpha$ is not solvable then we modify $M_\alpha$ to $M_\alpha'$ simply by incrementing $q$ and try to solve $M_\alpha'$.
Either we get a solution or we modify $M_\alpha'$ to $M_\alpha''$ by incrementing $q$ again.
By Lemma~\ref{lem:ssp}, if $(s,s')$ is solvable then there is a $q\in \{1,\dots, b\}$ such that the corresponding (modified) system has a solution.
Hence, after at most $b$ iterations we can decide whether $(s,s')$ is solvable or not.
Consequently, we have to solve at most $b$ linear systems with at most $ \vert E\vert \cdot  \vert S\vert^2 $ equations for $(s,s')$.
The value $b$ is not part of the input.
Thus, by Lemma~\ref{lem:complexity_solving_linear_system}, this is doable in $\mathcal{O}(\vert E\vert^3 \cdot \vert S\vert^4)$ time.
We have at most $\vert S \vert^2$ different SSA  to solve.
Hence, we can decide the $\tau$-SSP in time $\mathcal{O}(\vert E\vert^3 \cdot \vert S\vert^6)$.
\end{proof}

As a next step, we let $\tau=\tau_{R\mathbb{Z}PT}^b$ and prove the polynomial time decidability of $\tau$-ESSP.
Let $\alpha$ be an ESSA $(e,s)$ and let $s_1,\dots, s_k$ be the sources of $e$ in $A$.
By definition, a $\tau$-region $(sup, sig)$ solves $\alpha$ if and only if $sig(e)=(m,n)$ and $\neg sup(s)\ledge{sig(e)}$ for a $(m,n)\in E_\tau$.
By definition of $\tau$, every element $(m,n)\in E_\tau$ occurs at exactly one state in $\tau$ and this state is $m$.
Hence, $sup(s_1)=\dots=sup(s_k)=m$ and $sup(s)\not=m$.
We base the following lemma on this simple observation.
It provides necessary and sufficient conditions that an \emph{abstract} region must fulfill to imply a \emph{solving} (concrete) region.

\begin{lemma}\label{lem:essp_tau_4}
Let $\tau=\tau_{R\mathbb{Z}PT}^b$ and $A=(S,E,\delta,\iota)$ be a TS and let $s_1\edge{e}s'_1,\dots, s_k\edge{e}s'_k$ be the $e$-labeled transitions in $A$, that is, if $s'\in S\setminus \{s_1,\dots, s_k\}$ then $\neg s'\edge{e} $.
The atom $(e,s)$ is $\tau$-solvable if and only if there is an event $(m,n)\in E_\tau$ and an abstract region $(sup,abs)$ of $A$ such that the following conditions are satisfied:
\begin{enumerate}
\item
$abs(e) = -m+n \text{ mod } (b+1)$,
\item
$\psi_{s_1}\cdot abs  = m - sup(\iota) \text{ mod } (b+1)$,
\item
$(\psi_{s_1}-\psi_{s_i})\cdot abs = 0 \text{ mod } (b+1)$ for $i\in \{2,\dots, k\}$
\item
$(\psi_{s_1}-\psi_s)\cdot abs \not= 0 \text{ mod } (b+1)$.

\end{enumerate}
\end{lemma}
\begin{proof}
\textit{If}:
Let $(sup, abs)$ be an abstract region that satisfies the conditions $1$-$4$.
We obtain a $\tau$-solving region $(sup, sig)$ with (the same support and) the signature $sig$ defined by $sig(e')=abs(e')$ if $e'\not=e$ and $sig(e')=(m, n)$ if $e'=e$.
To argue that $(sup, sig)$ is a $\tau$-region we have to argue that $q\edge{e'}q'$ in $A$ implies $sup(q)\ledge{sig(e')}sup(q')$.
As $(sup, abs)$ is an abstract region this is already clear for transitions $q\edge{e'}q'$ where $e'\not=e$.
Moreover, $(sup, abs)$ satisfies $\psi_{s_1}\cdot abs = m - sup(\iota) \text{ mod } (b+1)$ and the path equation holds, that is, $sup(s_1)=sup(\iota) + \psi_{s_1}\cdot abs \text{ mod } (b+1)$ which implies $sup(s_1)=m$.
Consequently, by definition of $\tau$, we have $sup(s_1) \ledge{(m,n)} n $ in $\tau$.
Furthermore, by $abs(e)= -m+n \text{ mod } (b+1)$ we have $m+abs(e) = n\text{ mod } (b+1)$.
Hence, by $sup(s_1)\ledge{abs(e)}sup(s'_1)$, we conclude $sup(s'_1)=n$ and, thus, $sup(s_1)\ledge{(m,n)}sup(s'_1)$.
By $(\psi_{s_1}-\psi_{s_i})\cdot abs = 0 \text{ mod } (b+1)$ for $i\in \{2,\dots, k\}$, we obtain that $sup(s_1)=\dots =sup(s_k)=m$.
Therefore, similar to the discussion for $s_1\edge{e}s'_1$, we obtain by $sup(s_i)\ledge{abs(e)}sup(s'_i)$ that the transitions $sup(s_i)\ledge{(m,n)}sup(s'_i)$ are present in $\tau$ for $i\in \{2,\dots, k\}$.
Consequently, $(sup, sig)$ is a $\tau$-region.

Finally, by $(\psi_{s_1}-\psi_s)\cdot abs \not = 0 \text{ mod } (b+1)$, have that $sup(s_1)\not=sup(s)$ and thus $\neg sup(s)\ledge{sig(e)}$.
This proves $(e,s)$ to be $\tau$-solvable by $(sup, sig)$.

\textit{Only-if}:
Let $(sup, sig)$ be a $\tau$-region that solves $(e,s)$ implying, by definition, $\neg sup(s)\ledge{sig(e)}$.
We use $(sup, sig)$ to define a corresponding abstract $\tau$-region $(sup, abs)$ in accordance to Lemma~\ref{lem:concrete_to_abstract}.
If $sig(e)\in \{0,\dots,b\}$ then $sup(s)\ledge{sig(e)}$, a contradiction.
Hence, it is $sig(e)=(m,n)\in E_\tau$ such that $sup(s_i)\ledge{(m,n)}$ for $i\in \{1,\dots, k\}$ and $\neg sup(s)\ledge{(m,n)}$.
This immediately implies $sup(s)\not=sup(s_1)$ and, hence, $(\psi_{s_1}-\psi_s)\cdot abs \not= 0 \text{ mod } (b+1)$.
By $sup(s_i)\ledge{(m,n)}sup(s'_i)$ and definition of $\tau$, we have that $sup(s_i)=m$ and $sup(s'_i)=n$ for $i\in \{1,\dots, k\}$ implying $(\psi_{s_1}-\psi_{s_i})\cdot abs = 0 \text{ mod } (b+1)$ for $i\in \{2,\dots, k\}$.
Moreover, by $sup(s_1)\ledge{abs(e)}sup(s'_1)$ we have $ abs(e) = sup(s'_1)-sup(s_1)\text{ mod } (b+1)$.
Hence, it is $abs(e)= -m+n \text{ mod } (b+1)$.
Finally, by the path equation, we have $sup(s_1)=sup(\iota)+\psi_{s_1}\cdot abs \text{ mod } (b+1)$ which with $sup(s_1)=m$ implies $\psi_{s_1}\cdot abs  =  m - sup(\iota)\text{ mod } (b+1)$.
This proves the lemma.
\end{proof}

The proof of the following lemma exhibits a polynomial time decision algorithm for the $\tau_{R\mathbb{Z}PT}^b$-ESSP:
Given a TS $A=(S,E,\delta,\iota)$ and a corresponding ESSA $\alpha$, the system $M_{A'}$ is extended to a system $M_\alpha$.
If $M_\alpha$ has a solution $abs$, then it implies a region $(sup, abs)$ satisfying the conditions of Lemma~\ref{lem:essp_tau_4} and thus implies the $\tau$-solvability of $\alpha$.
Conversely, if $\alpha$ is solvable, then there is an abstract region $(sup, abs)$ that satisfies the four conditions of by Lemma~\ref{lem:essp_tau_4}.
The abstract signature $abs$ is the solution of a corresponding equation system $M_\alpha$.
Hence, we get a solvable $M_\alpha$ if and only if $\alpha$ is solvable.
We argue that the number of possible systems is bounded polynomially in the size of $A$.
The solvability of every system is also decidable in polynomial time.
Consequently, by the at most $\vert E\vert \cdot \vert S\vert $ ESSA to solve, this yields the announced decision procedure.

\begin{lemma}\label{lem:essp_tractability}
If a TS $A=(S,E,\delta,\iota)$ has the $\tau_{R\mathbb{Z}PT}^b$-ESSP is decidable in time $\mathcal{O}(\vert E\vert^4\cdot \vert S\vert^5)$.
\end{lemma}
\begin{proof}
To estimate the computational complexity of deciding the $\tau_{R\mathbb{Z}PT}^b$-ESSP for $A$ observe that $A$ has at most $\vert S\vert \cdot \vert E\vert$ ESSA to solve.
Hence, the maximum costs of deciding the $\tau_{R\mathbb{Z}PT}^b$-ESSP for $A$ equals $\vert S\vert \cdot \vert E \vert$ times the maximum effort for a single atom.

In order to decide the $\tau$-solvability of a single ESSA $(e,s)$, we compose systems in accordance to Lemma~\ref{lem:essp_tau_4}.
The maximum costs can be estimated as follows:
The (basic) part $M_{A'}$ of $M_\alpha$ has at most $\vert E\vert \cdot \vert  S \vert^2$ equations.
Moreover, $e$ occurs at most at $\vert S\vert -1$ states.
This makes at most $\vert S\vert$ equations to ensure that $e$'s sources will have the same support, the third condition of Lemma~\ref{lem:essp_tau_4}.
According to the first and the second condition, we choose an event $(m,n)\in E_\tau$, a value $sup(\iota)\in \{0,\dots, b\}$, define $abs(e)=-m+n \text{ mod } (b+1)$ and add the corresponding equation $\psi_{s_1} \cdot abs = m - sup(\iota)$.
For the fourth condition we choose a fixed value $q \in \{1,\dots, b\}$ and add the equation $(\psi_{s_1} - \psi_{s})\cdot abs =q$.
Hence, the system has at most $2 \cdot  \vert E\vert  \cdot  \vert S\vert^2 $ equations.

By Lemma~\ref{lem:complexity_solving_linear_system}, one checks in time $\mathcal{O}(\vert E\vert^3 \cdot  \vert S\vert^4 )$ if such a system has a solution.
Notice, we use that $2\cdot  \vert E\vert\cdot  \vert S\vert^2  = max\{\vert E\vert , 2 \cdot  \vert E\vert \cdot \vert S\vert^2  \}$.
There are at most $(b+1)^2$ possibilities to choose a corresponding $(m,n)\in E_\tau$ and only $b+1$ possible values for $x$ and for $q$, respectively.
Hence, for a fixed atom $(e,s)$, we have to solve at most $(b+1)^4$ such systems and $b$ is not part of the input.
Consequently, we can decide in time $\mathcal{O}(\vert E\vert^3\cdot   \vert S\vert^4 )$ if $(e,s)$ is solvable.
$A$ provides at most $\vert S \vert\cdot  \vert E \vert$ ESSA.
Hence, the $\tau_{R\mathbb{Z}PT}^b$-ESSP for $A$ is decidable in time $\mathcal{O}(\vert E\vert^4 \cdot  \vert S\vert^5)$.
\end{proof}

The following lemma completes the proof of Theorem~\ref{the:tractability} and, moreover, shows that \textsc{$\tau_{R\mathbb{Z}PT}^b$-Synthesis} is solvable in polynomial time.

\begin{corollary}\label{cor:tractability_synthesis}
There is an algorithm that constructs, for a TS $A=(S,E,\delta,\iota)$, a $\tau_{R\mathbb{Z}PT}^b$-net $N$ with a state graph $A_N$ isomorphic to $A$ if it exists in time $\mathcal{O}(\vert E\vert ^3 \cdot \vert S\vert^5  \cdot max\{ \vert E\vert, \vert S\vert  \})$.
\end{corollary}
\begin{proof}
By \cite{DBLP:series/txtcs/BadouelBD15}, if $\mathcal{R}$ is a set of regions of $A$ containing for each ESSP and SSA of $A$ a solving region, respectively,  then the $\tau$-net $N^\mathcal{R}_A=(\mathcal{R}, E(A), f, M_0)$, where $f((sup,sig),e)=sig(e)$ and $M_0((sup,sig))=sup(\iota)$ for $(sup, sig)\in \mathcal{R}, e\in E(A)$, has a state graph isomorphic to $A$.
Hence, the corollary follows from Lemma~\ref{lem:ssp_tractability} and Lemma~\ref{lem:essp_tractability}.
\end{proof}

\begin{example}
We pick up our running example TS $A$ and its spanning tree of Figure~\ref{fig:fundamental_cycles}.
We present two steps of the method given by Lemma~\ref{lem:essp_tractability} for the type $\tau^2_4$ and check $\tau^2_4$-solvability of the ESSA $(c,1)$.

For a start, we choose $(m,n)=(0,1)$ and $sup(0)=0$ and determine $abs(c)=-0+1=1$ which yields $abs=(x_a, x_b, 1 , x_d)$.
We have to add $\psi_0\cdot abs=m-sup(0)=0$ which, by $\psi_0=(0,0,0,0)$, is always true and do not contribute to the system.
Moreover, for $i\in \{0,2,3,4,5,6\}$, we add the equation $(\psi_0-\psi_i)\cdot abs =0$.
We have $\psi_0-\psi_6=(0,0,-2,0)$ and $(0,0,-2,0)\cdot abs= 0\cdot x_a - 0\cdot x_b - 2 -0\cdot x_d=0$  yields a contradiction.
Hence, $(c,1)$ is not solvable by a region $(sup, sig)$ where $sup(0)=0$ and $sig(c)=(0,1)$.
Similarly, we obtain that the system corresponding to $sup(0)\in \{1,2\}$ and $sig(c)=(0,1)$ is also not solvable.

\medskip
For another try, we choose $(m,n)=(2,2)$ and $sup(0)=2$.
In accordance to the first and the second condition of Lemma~\ref{lem:essp_tau_4} this determines $abs=(x_a, x_b, 0 , x_d)$ and yields the equation $\psi_0 \cdot abs=m-sup(0)=2-2=0$ which is always true.
For the fourth condition, we pick $q=2$ and add the equation  $(\psi_0-\psi_1)\cdot abs= 2\cdot x_a=2$.
Finally, for the third condition, we add for $i\in \{0,2,3,4,5,6\}$ the equation $(\psi_0-\psi_i)\cdot abs =0$ and obtain the following system of equations modulo $(b+1)$:
\begin{align*}
\psi_t\cdot abs 		&= 			&  	&2\cdot x_b 		& \  				& + x_d	&=0 \\
(\psi_0-\psi_1)\cdot abs 	&=  2\cdot x_a 	& \ 	&			&\				&\		&= 2 \\
(\psi_0-\psi_2)\cdot abs 	&= 2\cdot x_a 	& +\ 	&2\cdot x_b 		&\				&\		&= 0\\
(\psi_0-\psi_3)\cdot abs 	&= 2\cdot x_a 	& +\	& 2\cdot x_b 		& + 2\cdot 0			&\		&= 0
\end{align*}
\begin{align*}
(\psi_0-\psi_4)\cdot abs 	&= 2\cdot x_a 	& +\	& 2\cdot x_b 		& + 1\cdot 0			&\		&= 0\\
(\psi_0-\psi_5)\cdot abs 	&= 			& 	& 			&  2\cdot 0			&\		&= 0\\
(\psi_0-\psi_6)\cdot abs 	&= 			& \	& 			& 1\cdot 0			&\		&= 0
\end{align*}
This system is solvable by $abs=(1,2,0,2)$.
We construct a region in accordance to the proof of Lemma~\ref{lem:essp_tau_4}:
By $sup(0)=2$ we obtain $sup(1)=2+\psi_1\cdot abs=2+(1,0,0,0)\cdot (1,2,0,2)=0$.
Similarly,  by $sup(i)=2+\psi_i\cdot abs$ for $i\in \{2,\dots, 7\}$ we obtain $sup(2)=sup(3)=sup(4)=sup(5)=sup(6)=2$ and $sup(7)=0$.
Hence, by defining $sig(c)=(2,2)$, $sig(a)=1$, $sig(b)=2$ and $sig(d)=2$ we obtain a fitting $\tau_{R\mathbb{Z}PT}^b$-region $(sup, sig)$ that solves $(c,1)$.
\end{example}

\section{Conclusion}\label{sec:conclusion}%

In this paper, for all $b\in \mathbb{N}$, we completely characterize the computational complexity of \textsc{$\tau$-SSP} and \textsc{$\tau$-ESSP} and \textsc{$\tau$-Solvability} for the types of pure $b$-bounded P/T-nets, $b$-bounded P/T-nets and their corresponding $\mathbb{Z}_{b+1}$-extensions.
This answers an open problem posed by Schlachter et al. in~\cite{DBLP:conf/concur/SchlachterW17}.

\smallskip
Some open problems in the field of Petri net synthesis concern the computational complexity of $\tau$-\emph{synthesis up to language equivalence} (\textsc{$\tau$-Language Synthesis}) and $\tau$-\emph{synthesis from modal TS} (\textsc{$\tau$-Modal Synthesis}):

\smallskip
\textsc{$\tau$-Language Synthesis} is the task to find for a given TS $A=(S,E,\delta,\iota)$ a $\tau$-net $N$ whose state graph $A_N$ has the same language as $A$, that is, $L(A_N)=L(A)$.
If there is a sought $\tau$-net $N$ for $A$, then $A$ is called $\tau$-solvable up to language equivalence.
To attack this problem, in~\cite[p.~164]{DBLP:series/txtcs/BadouelBD15}, the language $L(A)$ of $A$ is viewed as the TS $L_A=(L(A), E,\delta_L,\varepsilon)$ where $\delta_L(w,e)=we$ if and only if $we\in L(A)$.
By the result of~\cite[p.~164]{DBLP:series/txtcs/BadouelBD15}, there is a $\tau$-net $N$ that solves $A$ up to language equivalence if and only if the TS $L_A$ has the $\tau$-ESSP.
Since there might be exponentially (or even infinite) many paths in $A$, computing $L_A$ and then checking the ESSP yields an algorithm that, in general, is at least exponential in the size of $A$.
Anyway, the exact computational complexity of $\tau$-language synthesis has not yet been proven, and, so far, there has been also no lower bound.
For $\tau\in \{\tau_{PT}^b,\tau_{PPT}^b\}$, our results imply a lower bound, to be seen as follows:
If $A=s_0\edge{e_1}s_1\edge{e_2}\dots\edge{e_n}s_n$ is a linear TS, then $L_A=\varepsilon\edge{e_1}e_1\edge{e_2}\dots\edge{e_n}e_1\dots e_n$ (the states of $L_A$ are $e_1$ and $e_1e_2$ and $\dots$ and $e_1\dots e_n$).
In particular, it is easy to see that $A$ and $L_A$ are isomorphic.
Consequently, by~\cite[p.~164]{DBLP:series/txtcs/BadouelBD15}, a linear TS $A$ is $\tau$-solvable up to language equivalence if and only if it has the $\tau$-ESSP.
Thus, by Theorem~\ref{the:hardness_results}, $\tau$-language synthesis is NP-hard, since there is a trivial reduction from \textsc{$\tau$-ESSP} to \textsc{$\tau$-language synthesis}.

\smallskip
$\tau$-modal synthesis~\cite{DBLP:conf/concur/SchlachterW17} is the task to find for a given \emph{modal} TS $M$ a $\tau$-net $N$ such that the state graph $A_N$ \emph{implements} $A$:
A \emph{modal} TS $M=(S, E, \delta_{must}, \delta_{may}, s_0)$ has a set of states $S$, events $E$, an initial state $s_0$, a (partial) function $\delta_{must}:S\times E\rightarrow S$ that defines the \emph{must}-edges and a (partial) function $\delta_{may}:S\times E\rightarrow S$ that defines the \emph{may} edges of $A$;
moreover, $\delta_{must}$ and $\delta_{may}$ satisfy that if $\delta_{must}(s,e)=s'$, then $\delta_{may}(s,e)=s'$, that is, every must-arc is a may-arc, but not every may-arc is necessarily a must-arc.
A TS $A$ that has the same event set as $M$ \emph{implements} $M$ if a relation $R\subseteq M(S)\times A(S)$ exists such that $(s_{0,M}, \iota)\in R$ and for all $(s, q)\in R$ and $e\in E(M)=E(A)$ the following holds:
\begin{enumerate}
\item
If $\delta_{must}(s, e)=s'$, then there is a $q'\in S(A)$ such that $\delta_A(q,e)=q'$ and $(s',q')\in R$.
\item
If $\delta_A(q,e)=q'$, then there is a $s'\in S(M)$ such that $\delta_{may}(s,e)=s'$ and $(s',q')\in R$.
\end{enumerate}
If there is a searched net $N$ for $M$, then $M$ is called $\tau$-\emph{implementable}.
The computational complexity of $\tau$-modal synthesis has been stated as an open problem in~\cite{DBLP:conf/concur/SchlachterW17}.
While at least an (exponential) upper bound is given in~\cite{DBLP:conf/concur/SchlachterW17}, a lower bound has not yet been stated.
Our results imply a lower bound for $\tau\in \{\tau_{PT}^b,\tau_{PPT}^b\}$.
This can bee seen as follows:
Every TS $A$ can be interpreted as a modal TS where the must-edges and the may-edges coincide.
For such a TS, the just introduced implementation relation then reduces to the well known relation of bisimulation~\cite[p.~22]{DBLP:series/eatcs/Gorrieri17}.
Moreover, it is also known that deterministic TS $A_0$ and $A_1$ are bisimilar if and only if they are language equivalent (also-called \emph{trace equivalent})~\cite[p.~26]{DBLP:series/eatcs/Gorrieri17}.

\smallskip
Altogether, we have justified that a linear TS $A$ has the $\tau$-ESSP if and only if it is $\tau$-solvable up to language equivalence if and only if, interpreted as modal TS, it is implementable by the state graph $A_N$ of a $\tau$-net $N$.
Thus, for $\tau\in\{\tau_{PT}^b,\tau_{PPT}^b\}$, the following theorem is a corollary of Theorem~\ref{the:hardness_results} and, at least, gives lower bounds for the computational complexity of both $\tau$-language synthesis and $\tau$-modal synthesis:
\begin{theorem}
Let $\tau\in \{\tau_{PT}^b,\tau_{PPT}^b\}$.
Deciding for a TS $A$ if it is $\tau$-solvable up to language equivalence or deciding for a modal TS $M$ if it is $\tau$-implementable is NP-hard.
\end{theorem}

It remains for future work to settle the exact complexity of $\tau$-language synthesis and $\tau$-modal synthesis.
Moreover, one might investigate if \textsc{$\tau$-Solvability} and \textsc{$\tau$-ESSP} remain NP-complete for $1$-grade TS if $\tau\in \{\tau_{\mathbb{Z}PT}^b,\tau_{\mathbb{Z}PPT}^b\}$.


\end{document}